\documentclass[12pt]{article} 

\usepackage{amsfonts,amsmath,amsthm,graphicx,float,placeins,subcaption,enumerate,
  color,natbib,verbatim, hyperref, url}
\usepackage{xr-hyper}
\usepackage{amssymb}
\usepackage{lscape,longtable}
\usepackage[a4paper, left=0.9in, right=0.9in, top=0.7in, bottom=0.8in]{geometry}

\usepackage[dvipsnames]{xcolor} 
\usepackage{caption} 
\usepackage{setspace}
\usepackage{booktabs} 

\usepackage{lineno} 
\newcommand*\patchAmsMathEnvironmentForLineno[1]{%
  \expandafter\let\csname old#1\expandafter\endcsname\csname
  #1\endcsname
  \expandafter\let\csname oldend#1\expandafter\endcsname\csname
  end#1\endcsname
  \renewenvironment{#1}%
                   {\linenomath\csname old#1\endcsname}%
                   {\csname oldend#1\endcsname\endlinenomath}}%
\newcommand*\patchBothAmsMathEnvironmentsForLineno[1]{%
  \patchAmsMathEnvironmentForLineno{#1}%
  \patchAmsMathEnvironmentForLineno{#1*}}%
\AtBeginDocument{%
  \patchBothAmsMathEnvironmentsForLineno{equation}%
  \patchBothAmsMathEnvironmentsForLineno{align}%
  \patchBothAmsMathEnvironmentsForLineno{flalign}%
  \patchBothAmsMathEnvironmentsForLineno{alignat}%
  \patchBothAmsMathEnvironmentsForLineno{gather}%
  \patchBothAmsMathEnvironmentsForLineno{multline}%
}

\usepackage{enumitem}

\graphicspath{{figures/}}
\allowdisplaybreaks

\newcommand{\0}{\mathbf{0}}

\newcommand{\I}{\mathbf{I}}

\newcommand{\e}{\mathbf{e}}

\newcommand{\U}{\mathbf{U}}

\newcommand{\X}{\mathbf{X}}

\newcommand{\Z}{\mathbf{Z}}

\newcommand{\diff}{\mathrm{d}}

\newcommand{\bbeta}{\boldsymbol{\beta}}
\newcommand{\brho}{\boldsymbol{\rho}}

\newcommand{\bomega}{\boldsymbol{\omega}}

\newcommand{\hbeta}{{\hat{\boldsymbol{\beta}}}}
\newcommand{\cbeta}{{\check{\boldsymbol{\beta}}}}

\newcommand{\ag}{\text{ag}}
\newcommand{\pl}{\text{pl}}

\newcommand{\tdelta}{\tilde\delta}
\newcommand{\tpi}{\tilde\pi}
\newcommand{\onesqN}{\frac{1}{\sqrt{N}}}

\newcommand{\oneN}{\frac{1}{N}}

\newcommand{\cA}{\mathcal{A}}
\newcommand{\Fn}{\mathcal{F}_N}
\newcommand{\RN}{\mathcal{R}_N}

\newcommand{\sumN}{\sum_{i=1}^{N}}

\newcommand{\tp}{^{\top}}
\newcommand{\tr}{\operatorname{tr}}

\renewcommand{\Pr}{\mathbb{P}}
\newcommand{\Exp}{\mathbb{E}}
\newcommand{\Nor}{\mathbb{N}}
\newcommand{\Var}{\mathbb{V}}

\newcommand{\bI}{\mathbb{I}}

\newtheorem{assumption}{Assumption}
\newtheorem{lemma}{Lemma}
\newtheorem{proposition}{Proposition}
\newtheorem{corollary}{Corollary}
\newtheorem{theorem}{Theorem}

\theoremstyle{definition}
\newtheorem{remark}{{Remark}}

\newcommand{\red}{\color{red}}

\newcommand{\blue}{\color{blue}}
\newcommand{\why}[1]{{\blue WHY: \red #1}}
\newcommand{\jy}[1]{{\textcolor{violet}{JY: #1}}}
\newcommand{\qkd}[1]{{\textcolor{orange}{QKD: #1}}}

\hypersetup{hidelinks}

\renewcommand{\why}[1]{\ignorespaces}
\renewcommand{\jy}[1]{\ignorespaces}
\renewcommand{\qkd}[1]{\ignorespaces}

\begin{document}

\def\spacingset#1{\renewcommand{\baselinestretch}%
{#1}\small\normalsize} \spacingset{1}

\title{Rare-Feature-Adaptive Subsampling for Big Data Regression}

\ifdefined\BLINDVERSION
\author{}
\else
\author{
Qingkai Dong$^{1}$, Jun Yan$^{1}$, Qiang Zhang$^{2,3}$, and HaiYing Wang$^{1}$\\[0.3em]
\small $^{1}$Department of Statistics, University of Connecticut, Storrs, CT, USA\\
\small $^{2}$Vickie and Jack Farber Vision Research Center at Wills Eye, Wills Eye Hospital, Philadelphia, PA, USA\\
\small $^{3}$Sidney Kimmel Medical College at Thomas Jefferson University, Philadelphia, PA, USA
}
\fi

\date{}
\maketitle

\begin{abstract}
Rare binary features cause instability and slow convergence in
regression. Subsampling massive datasets may discard observations
carrying rare-feature information. We show that rare-feature
coefficient convergence rates depend on feature prevalence and the
inclusion probabilities assigned to rare observations. This finding
motivates subsampling designs that adapt sampling probabilities to feature
rarity. We first study optimal rare-census designs, which extend the
 negative subsampling approach by retaining every rare
observation. However,
they may face practical constraints. We then propose a computationally
efficient rarity-adaptive L-optimal design that incorporates inverse prevalence
scaling. Because implementing optimal designs requires pilot estimates
whose quality depends on adequate representation of rare features, we
propose a balanced design, a one-step rarity-adaptive design based
only on binary features. It can be used alone, as a pilot,
or combined with designs targeting other objectives. For
implementation in two stages, we develop a union-sample estimator that
reuses pilot observations while accounting for sample overlap and
establish its asymptotic normality. Simulations and applications show
improved stability and estimation efficiency over existing methods.
One application evaluates designs that jointly protect
rare outcomes and rare features.

\end{abstract}

\newpage
\spacingset{1.9}
\section{Introduction}

Rare features play an important role in many scientific domains and
represent infrequent but consequential characteristics, such as rare
diseases, exposures, or underrepresented populations. These features
are common in genetics, electronic health records, low-prevalence
epidemiologic exposures, and social surveys
\citep[e.g.,][]{doerken2019penalized, king2001logistic, pizzo2019rare,
chen2024tree, shyalika2024comprehensive}. Although rare, they often
carry valuable information for scientific inference and policy
decisions. Their sparsity, however, creates substantial difficulties
for statistical modeling. When such features, typically coded as
binary covariates, take the value one in only a small fraction of
observations, instability, inefficiency, and identifiability issues
arise. For example, in high-dimensional linear models, rare features
lead to extremely sparse design matrices, estimation inconsistency,
and variable selection failure \citep{yan2021rare}. Hypothesis testing
loses power for signals from rare features
\citep{mukherjee2015hypothesis}.

Importantly, these challenges can persist with larger sample sizes
and can be exacerbated in large-scale data analyses.  Subsampling is a
common strategy for scalable regression, reducing computational cost
by fitting a model on a subset of the full data
\citep{wang2018optimal}. This is useful for models requiring iterative
optimization, such as generalized linear models, where computation is
expensive for massive datasets. However, subsampling can worsen the
rare-feature problem by failing to preserve rare observations.  When
observations carrying rare-feature information are omitted or included
too infrequently, the subsample may not permit reliable estimation,
thereby aggravating the instability and inefficiency already present in
full-data analysis.

Many studies have developed methodological and theoretical results for
rare-outcome (rare-event) analyses \citep[e.g.,][]{king2001logistic,
owen2007infinitely}. These studies, together with
optimal subsampling methods for rare outcomes in logistic and Cox
regression \citep{wang2021nonuniform, keret2023analyzing}, focus on
outcome rarity rather than rarity arising from the covariates. Recent
studies analyze rare covariates \citep{yan2021rare, chen2024tree,
mishra2024taro}, primarily in high dimensions where feature
aggregation is often used to alleviate rarity.  The asymptotic
behavior of regression estimators with rare features, the additional
difficulties introduced by subsampling, and the design of subsampling
procedures that preserve rare features remain less understood.

To address these gaps, this paper studies subsampling for regression
with rare features. We establish asymptotic results showing that
rare-feature prevalence slows the convergence rates of the associated
coefficient estimators and that subsampling further alters these
rates through the inclusion probabilities assigned to rare
observations. Protecting rare observations therefore improves both
rare-feature and overall estimation accuracy. This yields a general
design principle: sampling probabilities should adapt to feature
rarity, with observations expressing less prevalent features receiving
larger inclusion probabilities while respecting the sampling budget.
We call designs that follow this principle and attain the resulting
target rate rarity-adaptive (RA) designs. The negative subsampling
approach for rare outcomes
\citep{wang2021nonuniform,keret2023analyzing} can be extended to
rare-feature settings by retaining every rare observation and sampling
from the remainder. We call the resulting designs rare-census (RC)
designs. When complete retention is feasible, RC designs form a constrained
subclass of RA designs.

Our main design contribution is an RA framework that protects less
prevalent features more strongly while accommodating practical
constraints. Within this framework, we derive optimal RC designs that
retain every rare observation and allocate the remaining budget under
A-optimal or L-optimal criteria \citep{wang2018optimal}. When feasible,
RC preserves the full-data convergence rate for rare-feature
coefficients but does not generally recover their full-data
variance. In practice, however, RC can exhaust the budget for nonrare
observations or become infeasible. It also requires a prespecified
rare feature set.  Beyond this constrained case, we build on existing
optimal subsampling designs \citep{wang2018optimal,yu2020quasi}. We
show that the A-optimal design (Aopt) can attain the target rate but
requires dense calculations, whereas the computationally simpler
L-optimal design (Lopt) does not generally provide the required rarity
adaptivity. We therefore propose a rarity-adaptive L-optimal design
(RA-Lopt), which incorporates inverse prevalence scaling into Lopt
without requiring a prespecified rare feature set.

The preceding optimal designs depend on pilot estimates
\citep{wang2018optimal}, which can be unstable when rare features are
poorly represented. We therefore propose the one-step balanced design
(BL), whose probabilities use only binary feature information and
require no model or parameter estimates. BL is itself an RA design,
requires no prespecified rare feature set, and can be used alone, as a
pilot, or as a component of another design.

For two-step
implementation, existing aggregation procedures combine pilot and
second-step information. We
develop a union-sample estimator that incorporates both samples while
accounting for their overlap. We also provide designs for simultaneous
rare outcomes and rare features. Numerically, the proposed designs
compare favorably with existing alternatives and can approach
full-data performance with a much smaller subsample.

The remainder of the paper is organized as
follows. Section~\ref{sec:prelim} introduces the notation and model
framework. Section~\ref{sec:asymptotics} develops the asymptotic
theory for full-data and subsample
estimators. Section~\ref{sec:subsampling-designs} develops the RA
framework, including optimal RC designs, RA-Lopt, and BL.
Section~\ref{sec:implementation} presents feasible implementations,
the union-sample estimator and its theoretical guarantees, and designs
for simultaneous rare features and outcomes.
Sections~\ref{sec:simulation} and~\ref{sec:real} report simulations
and applications, respectively. Section~\ref{sec:discussion}
concludes. All proofs are provided in Appendix~\ref{app:proofs}.

\section{Preliminaries}
\label{sec:prelim}

Consider a regression problem where $Y\in \mathbb{R}$ is the response
variable and $\X \in \mathbb{R}^{d}$ is the feature vector. Define the function
$\psi(Y,\X,\bbeta)=\{Y - \mu(\X\tp \bbeta)\}\X$, where $\mu(\cdot)$ is a
continuously differentiable and strictly increasing function. Let
$\Psi_{0}(\bbeta) = \Exp\{\psi(Y,\X,\bbeta)\}$ and assume that the true
parameter $\bbeta_0$ is the unique root of $\Psi_{0}(\bbeta)=\0$. Let
the full data be a sample of size $N$, denoted by
$\Fn = \left\{(Y_{i}, \X_{i}) \right\}_{i=1}^{N}$, where each pair
$(Y_{i}, \X_{i})$ is an independent copy of $(Y, \X)$.  The full-data
estimator $\hbeta_{N}$ is the solution to 
$ \Psi_N(\bbeta) = N^{-1} \sumN \psi(Y_{i}, \X_{i}, \bbeta) = \0$.
Estimating equations of this form include canonical generalized
linear models, such as logistic regression, and many Z-estimation
problems.

To study the estimation challenges posed by rare features, we
partition $\X=(\U\tp,\Z\tp)\tp$, where $\U\in\mathbb R^{d_u}$ contains
nonrare features, including an intercept term, and
$\Z\in\{0,1\}^{d_r}$ contains rare features. Each component of $\Z$ is a
Bernoulli random variable; the $j$-th rare feature $Z_{\cdot j}$ takes the
value~$1$ (feature expressed) with prevalence
$\rho_{j,N}=\Exp(Z_{\cdot j})$, $j=1,\dots,d_r$. Write
$\Z_i=(Z_{i1},\ldots,Z_{id_r})\tp$ for the rare-feature vector of
observation $i$. The rare-feature regime considered in this paper
allows $\rho_{j,N}\to0$ while $N\rho_{j,N}\to\infty$ as
$N\to\infty$. In other words, the number of observations expressing a
rare feature diverges as $N\to\infty$, but at a slower rate than
$N$. Note that assuming rare features are binary is appropriate because
categorical variables are typically coded as dummy variables in regression
models.

We use standard notation throughout. The Euclidean norm of a vector is
denoted by $\|\cdot\|$, and the Frobenius norm of a matrix by
$\|\cdot\|_F$. The notation $^{\otimes2}$ denotes the outer product,
i.e., $\X^{\otimes2}=\X\X\tp$. The first derivative of a function
$f(x)$ with respect to $x$ is denoted by $\dot f(x)$. We write $\I$
for the identity matrix, $\mathbb I(\cdot)$ for the indicator
function, $\Exp(\cdot)$ for expectation, and $\Var(\cdot)$ for
variance.  Convergence in probability and convergence in distribution
are denoted by $\xrightarrow{p}$ and $\xrightarrow{d}$,
respectively. We use $o_{p}(1)$ and $O_{p}(1)$ to represent a random
quantity that converges to $0$ in probability and is bounded in
probability, respectively.

\section{Theoretical Results under Rare Features}
\label{sec:asymptotics}

The goal of this section is to characterize the statistical cost of
rare features before and after subsampling. We first analyze the
full-data estimator, which provides the baseline convergence rate for
rare-feature coefficients. We then analyze a general subsample
estimator and show how this baseline rate is modified by the
inclusion probabilities assigned to observations expressing rare
features.

\subsection{Asymptotics of the Full-Data Estimator}

We first examine the full-data estimator $\hbeta_{N}$. In
regression asymptotics without rare features, the matrices
$M_{N}(\bbeta) = - N^{-1} \sumN \dot\psi(Y_{i},\X_{i},\bbeta)$
and
$\Lambda_{N}(\bbeta) =N^{-1} \sumN \psi^{\otimes2}(Y_{i},\X_{i},\bbeta)$
are typically assumed to have positive-definite limits,
and all coefficient estimators achieve the
$N^{-1/2}$ convergence rate. With rare features, this classical
scaling is no longer appropriate: for example, the diagonal entry of
$M_N(\bbeta_0)$ corresponding to the $j$-th rare-feature coordinate,
i.e., $-N^{-1} \sumN \{\dot\mu(\X_i\tp \bbeta_{0}) Z_{ij}\}$, is of order
$\rho_{j,N}$ and converges to zero. We therefore introduce the diagonal scaling matrix
$D_{N} = \operatorname{diag}\{\mathbf{1}_{d_{u}}\tp, (\brho_{N}^{1/2})\tp\}$,
where $\brho_{N} = (\rho_{1,N}, \dots, \rho_{d_{r},N})\tp$ and the
power $^{1/2}$ is taken element-wise. We make the following regularity assumptions.

\begin{assumption} \label{asmp:regularity-1}
The nonrare features satisfy $\Exp\|\U\|^4<\infty$. For rare features,
$\rho_{j,N}>0$, $\rho_{j,N}\to0$, and $N\rho_{j,N}\to\infty$, as $N\to\infty$,
for each $j=1,\ldots,d_r$.
\end{assumption}

\begin{assumption} \label{asmp:regularity-2}
(i) For all $\bbeta\in\Theta$, where $\Theta$ is a compact set,
$\sup_{N}\Exp\{\|D_{N}^{-1}\psi(Y,\X,\bbeta)\|^{2}\}<\infty$. (ii) There
exists $\kappa>0$ such that
$N^{-\kappa/2} \Exp\|D_{N}^{-1} \psi(Y, \X, \bbeta_{0})\|^{2+\kappa} \to 0$.
(iii) There exists a function $L(x)\ge0$ satisfying
$\sup_{N}\Exp\{L^2(\X)\|D_{N}^{-1}\X\|^2\}<\infty$ such that, for all
$\bbeta_{1},\bbeta_{2} \in \Theta$,
\begin{align*}
  &|\mu(\X\tp\bbeta_{1})-\mu(\X\tp\bbeta_{2})|
    \le L(\X)\|\bbeta_{1}-\bbeta_{2}\|,
    \quad 
    |\dot\mu(\X\tp\bbeta_{1})-\dot\mu(\X\tp\bbeta_{2})|
    \le L(\X)\|\bbeta_{1}-\bbeta_{2}\|.
\end{align*}
(iv) For any $\epsilon>0$, there exists a constant $c(\epsilon)>0$ such
that, for all sufficiently large $N$,
$\inf_{\bbeta\in\Theta:\|\bbeta-\bbeta_{0}\|\ge \epsilon}\|D_{N}^{-2}\Psi_{0}(\bbeta)\|\ge c(\epsilon)$.
\end{assumption}

\begin{assumption} \label{asmp:regularity-3}
Define $M_{r,N}(\bbeta) = D_{N}^{-1}M_{N}(\bbeta)D_{N}^{-1}$ and
$\Lambda_{r,N}(\bbeta) = D_{N}^{-1}\Lambda_{N}(\bbeta)D_{N}^{-1}$.
Assume $M_{r,N}(\bbeta_{0}) \xrightarrow{p} M_{r}(\bbeta_{0}) $ and
$\Lambda_{r,N}(\bbeta_{0}) \xrightarrow{p} \Lambda_{r}(\bbeta_{0})$,
where $M_{r}(\bbeta_{0})$ and $\Lambda_{r}(\bbeta_{0})$ are finite and
positive-definite matrices.
\end{assumption}

Assumptions~\ref{asmp:regularity-1}--\ref{asmp:regularity-3} are
standard regularity conditions adapted to rare-feature scaling.
Assumption~\ref{asmp:regularity-1} formalizes the rare-feature regime:
for each rare feature, the expected number of observations expressing
it goes to infinity but grows slower than $N$. The prevalences
need not be of the same order.  Assumption~\ref{asmp:regularity-2}
states moment, smoothness, and scaled identification conditions.
Assumption~\ref{asmp:regularity-3} requires that appropriately
rescaled matrices $M_{r,N}(\bbeta_0)$ and $\Lambda_{r,N}(\bbeta_0)$
are full rank and have positive-definite limits.

\begin{theorem}
  \label{thm:full}
Under Assumptions~\ref{asmp:regularity-1}--\ref{asmp:regularity-3}, as $N \to \infty$,
\begin{equation*}
\sqrt{N} D_N \left(\hbeta_{N} - \bbeta_{0}\right) \xrightarrow{d}
\Nor\left(\0,
  M_{r}^{-1}(\bbeta_{0}) \Lambda_{r}(\bbeta_{0}) M_{r}^{-1}(\bbeta_{0})
\right).
\end{equation*}
\end{theorem}

Theorem~\ref{thm:full} extends the classical asymptotic normality
result to settings with rare features. The estimator of the $j$-th
rare-feature coefficient converges at rate $(N\rho_{j,N})^{-1/2}$,
where $N\rho_{j,N}$ is the expected number of observations with
$Z_{ij}=1$. In contrast, nonrare
coordinates of $\hbeta_N$ retain the usual $N^{-1/2}$ convergence
rate. When $d_r=0$, $D_N$ reduces to the identity matrix and the usual
root-$N$ consistency for all coefficients is recovered.  Relative to
the nonrare coordinates, the rare-feature coordinates can therefore
contribute more substantially to the overall coefficient estimation
error. This contribution can be particularly large when the model
contains many rare features or when their prevalence is extremely low.

This result has direct implications for subsampling. Because the
information about a rare-feature coefficient is tied to observations
expressing that feature, a subsample that underrepresents or omits
them can further reduce estimation accuracy. Therefore, inclusion
probabilities assigned to rare observations are central to subsampling
performance. We next study how these probabilities affect
identifiability and subsample estimator convergence.

\subsection{Theory of General Subsample Estimators}

\subsubsection{Poisson Sampling and Its Identifiability Risk}
\label{sec:poisson-subsampling}

We draw a subsample of expected size $n$ from $\Fn$ using
sampling probabilities $\{\pi(Y_{i}, \X_{i})\}_{i=1}^{N}$ satisfying
$0 < \pi(Y_{i}, \X_{i}) \le 1$ and
$\sumN \pi(Y_{i}, \X_{i}) = n$. Observation $i$ is selected if
$V_{i} < \pi(Y_{i}, \X_{i})$, where $V_{i} \sim \text{Uniform}(0,1)$
is independent of $\Fn$. This is known as Poisson sampling
\citep{wang2022sampling}. Let
$\delta_{i} = \bI\{V_{i} < \pi(Y_{i}, \X_{i})\}$ be the inclusion
indicator, so that the actual subsample size is $\sumN \delta_{i}$,
which satisfies $\Exp(\sumN \delta_{i}) = n$. We write
$\pi(Y, \X) = \min\{h_{N} \varphi(Y,\X),1\}$, where $h_{N}=n/N$ is the
base subsampling rate and $\varphi(Y,\X)$ is the weight function that
reflects the relative importance of each observation.  The form of
$\varphi(Y,\X)$ is determined by the subsampling design.  The baseline
choice $\varphi(Y,\X)=1$ gives uniform subsampling, under which all
observations are treated as equally important. A subsample estimator
$\cbeta_n$ is defined as the solution to
\begin{align}
 \label{eq:ee-sub}
  \Psi_{n}(\bbeta) =
  \oneN \sumN \left\{\frac{\delta_{i}}{\pi(Y_{i}, \X_{i})}
  \psi(Y_{i}, \X_{i}, \bbeta) \right\}  = \0.
\end{align}
For the $j$-th rare feature, define
\begin{align}
  \label{eq:omega-N-j}
  \omega_{j,N}
  =
  \Exp\left\{\pi(Y,\X)\mid Z_{\cdot j}=1\right\}.
\end{align}
This quantity represents the expected inclusion probability assigned to
observations expressing the $j$-th rare feature. A larger
$\omega_{j,N}$ indicates that the subsampling design more strongly
protects observations carrying that rare feature.

A coefficient identifiability issue arises when a subsample contains
no observation expressing a rare feature. In this case, the
corresponding coefficient cannot be estimated from the subsample. We
characterize the probability of this issue as follows.

\begin{proposition}
\label{prop:lower-bound}
Define the event $\cA = \cup_{j=1}^{d_{r}}\{ \sumN\delta_{i}Z_{ij}=0 \}$, which
occurs when one or more rare features have no expressed observations in
a subsample. When there exists a constant $\zeta<1$ such that
$\pi(Y,\X) \le \zeta$ almost surely, we have
\begin{align}
\label{eq:lower-bound}
  \Pr(\cA) \ge \max_{1 \le j \le d_{r}} \left\{ 
  \exp{\left(-  \frac{N \rho_{j,N} \omega_{j,N}}{1-\zeta}  \right)}
   \right\}.
\end{align}
\end{proposition}

Proposition~\ref{prop:lower-bound} shows how the inclusion probability
assigned to rare observations affects identifiability. The lower bound
for $\Pr(\cA)$ in \eqref{eq:lower-bound} may not decrease as $N$ grows
because it is also affected by $\rho_{j,N}$ and $\omega_{j,N}$. Under
uniform subsampling, $\omega_{j,N}=h_N$, so the bound is controlled by
$n\rho_{j,N}$. For $N=10^5$, $d_r=1$, $\rho_{1,N}=10^{-3}$, and
$n=10^3$, the lower bound is approximately $0.36$, meaning at least
$36\%$ of subsamples fail to produce an estimate. This risk grows with
multiple rare features, as shown in Section~\ref{sec:simulation}.

\subsubsection{Asymptotics of the Subsample Estimator}
\label{sec:asymptotics-subsample}

We next investigate the asymptotic behavior of the subsample estimator
$\cbeta_n$. To facilitate the presentation, define
$\Omega_{N} = \operatorname{diag}\{h_{N}^{1/2}\mathbf{1}\tp_{d_{u}}, (\bomega_{N}^{1/2})\tp\}$,
where $\bomega_{N} = (\omega_{1,N}, \dots, \omega_{d_{r},N})\tp$ with
elements $\omega_{j,N}$ defined in Equation~\eqref{eq:omega-N-j}.  The
matrix $\Omega_ND_N$ gives the scaling induced jointly by subsampling design and
feature rarity. Its diagonal entries are $h_N^{1/2}$ for
nonrare coordinates and $(\rho_{j,N}\omega_{j,N})^{1/2}$ for the
$j$-th rare-feature coordinate. For simplicity, define
$\Lambda_{\pi,N}(\bbeta) = N^{-1}\sumN \{\pi^{-1}(Y_{i},\X_{i})\psi^{\otimes2}(Y_{i},\X_{i},\bbeta)\}$
and
$K_{N}(\X) = \Omega_{N}D_{N}^{-1}\X\X\tp D_{N}^{-1}\Omega_{N}^{-1}$.
We make the following regularity assumptions on the sampling
probabilities:

\begin{assumption}
  \label{asmp:samp-prob-1}
(i) There exists a constant $\alpha>0$ such that
$\varphi(Y,\X) \ge \alpha$ almost surely. (ii) Assume
$\Omega_{N}M_{r,N}(\bbeta_{0})\Omega_{N}^{-1} \xrightarrow{p} M_{\pi}(\bbeta_{0})$
and
$\Omega_{N}D_{N}^{-1}\Lambda_{\pi,N}(\bbeta_{0}) D_{N}^{-1}\Omega_{N} \xrightarrow{p} \Lambda_{\pi}(\bbeta_{0})$,
where $M_{\pi}(\bbeta_{0})$ is finite and non-singular and
$\Lambda_{\pi}(\bbeta_{0})$ is finite and positive-definite.
\end{assumption}

\begin{assumption}
  \label{asmp:samp-prob-2}
  (i) For $L(x)$ defined in
Assumption~\ref{asmp:regularity-2},
$\sup_{N} \Exp \left\| L(\X) K_{N}(\X) \right\|_{F} < \infty$. (ii) There
exists $\kappa_{2}>0$ such that
$N^{-\kappa_{2}/2} \Exp\|\delta\pi^{-1}(Y,\X)\Omega_{N} D_{N}^{-1} \psi(Y, \X, \bbeta_{0})\|^{2+\kappa_{2}} \to 0$. (iii)
Assume
$n^{-1} \Exp \left\|\dot\mu(\X\tp\bbeta_{0}) K_{N}(\X) \right\|^{2}_{F} \to 0$.
\end{assumption}

\begin{theorem}
\label{thm:sub1}
Under Assumptions~\ref{asmp:regularity-1}--\ref{asmp:samp-prob-2},
if the base subsampling rate $h_{N}$ satisfies
$Nh_{N}\rho_{j,N}\to \infty$ for $j=1,\dots,d_{r}$, then
\begin{equation*}
    \sqrt{N} \Omega_{N} D_{N}
    (\cbeta_{n} - \bbeta_{0})
  \xrightarrow{d}
  \Nor\left(\0,
    M_{\pi}^{-1}(\bbeta_{0}) \Lambda_{\pi}(\bbeta_{0}) \{M_{\pi}^{-1}(\bbeta_{0})\}\tp\right),
\end{equation*}
where $M_{\pi}(\bbeta_0)$ and $\Lambda_{\pi}(\bbeta_0)$ are defined in
Assumption~\ref{asmp:samp-prob-1}.
\end{theorem}

Theorem~\ref{thm:sub1} shows that, beyond the base sampling rate,
rare-feature estimation depends on the inclusion probabilities
assigned to observations expressing the rare feature. To see this,
partition $\cbeta_n=(\cbeta_{n,u}\tp,\cbeta_{n,r}\tp)\tp$ into its
nonrare and rare-feature coordinates. For $\cbeta_{n,u}$, the
corresponding diagonal entries of $\Omega_ND_N$ are $h_{N}^{1/2}$, so
their convergence rate is $(N h_N)^{-1/2} = n^{-1/2}$. This agrees
with standard results in the subsampling literature; see, for example,
M-estimation \citep{wang2022sampling} and the Cox model
\citep{zhang2024approximating}. By contrast, the convergence rate for
the $j$-th element of $\cbeta_{n,r}$ is
$(N\rho_{j,N}\omega_{j,N})^{-1/2}$. This indicates that the convergence
rate of $\cbeta_{n,r,j}$ improves as $\omega_{j,N}$ increases;
thus, a subsampling design should make $\omega_{j,N}$ large.  Note that
$N\rho_{j,N}\omega_{j,N}=\Exp(\sumN\delta_iZ_{ij})$ is the expected number of
sampled observations expressing the $j$-th rare feature, which is bounded by
both the expected subsample size $n$ and the expected full-data count
$N\rho_{j,N}$.  Therefore, a desirable subsampling design can choose
$\varphi(Y,\X)$ so that $N\rho_{j,N}\omega_{j,N}\asymp \min\{n,N\rho_{j,N}\}$,
which implies $\omega_{j,N}$ is of order $\min\{h_N\rho_{j,N}^{-1},1\}$. We
refer to designs satisfying this order
$\omega_{j,N}\asymp\min\{h_N\rho_{j,N}^{-1},1\}$ as rarity-adaptive (RA)
designs. Equivalently, the expected inclusion probability assigned to
observations expressing a rare feature increases in inverse proportion
to its prevalence until reaching one. 

\section{Subsampling Designs}
\label{sec:subsampling-designs}

This section constructs RA subsampling designs that achieve the RA order
$\omega_{j,N}\asymp\min\{h_N\rho_{j,N}^{-1},1\}$. We first study optimal RC
designs, which
retain every rare observation and optimally allocate the remaining
sampling budget. We then evaluate Aopt and Lopt under the rarity
adaptivity objective, and propose RA-Lopt.  Finally, we introduce BL,
which does not require pilot estimates and can be used alone, as a
pilot design, or as a component of another design. The designs in this section
are formulated based on the observed full-data set $\Fn$; their
implementation details are developed in Section~\ref{sec:feasible}.

\subsection{Optimal Rare-Census Designs}
\label{sec:rc-designs}

Negative subsampling has been used for rare outcomes, with all rare
cases or events retained and the remaining observations subsampled
\citep{wang2021nonuniform,keret2023analyzing}. It is natural
to extend this strategy to rare feature settings. Let
$\RN=\{i:\Z_i\neq\0,\ i=1,\dots,N\}$ denote the index set of
observations expressing at least one rare feature (rare observations),
and let $|\RN|$ denote its size. An RC design satisfies
$\pi(Y_i,\X_i)=1$ for $i\in\RN$ and $0<\pi(Y_i,\X_i)\le1$ for $i\in\RN^c$. 

For Theorem~\ref{thm:sub1} to apply, the sampling budget for nonrare
observations, $n-|\RN|$, must be of order $n$, because
Assumption~\ref{asmp:samp-prob-1}(i) implies $n-|\RN|\ge\alpha(1-|\RN|/N)n$,
where $|\RN|/N\xrightarrow{p}0$.  RC feasibility also implies
$N\rho_{j,N}=O(n)$, so $\omega_{j,N}=1\asymp\min\{h_N\rho_{j,N}^{-1},1\}$.
Thus, RC designs form a constrained subclass of RA designs, requiring complete
retention of rare observations and a sufficiently large sampling budget for
nonrare observations.

However, retaining all rare observations and attaining the full-data
rate does not guarantee recovery of the full-data variance.
Let $\Var_N(\cdot)$ denote the finite-$N$ sandwich covariance
approximation for an estimator. For $\cbeta_n$ and $\hbeta_N$ in
Theorems~\ref{thm:sub1} and~\ref{thm:full}, respectively, we have
$\Var_{N}(\cbeta_{n}) =
N^{-1}M_N^{-1}(\bbeta_0)\Lambda_{\pi,N}(\bbeta_0)M_N^{-1}(\bbeta_0)$
and
$\Var_N(\hbeta_{N}) =
N^{-1}M_N^{-1}(\bbeta_0)\Lambda_N(\bbeta_0)M_N^{-1}(\bbeta_0)$.
For an RC design, the subsample
estimator covariance has the following decomposition:

\begin{corollary}
  \label{cor:rc-variance-decomposition}
Suppose the conditions of Theorem~\ref{thm:sub1} hold.  Under an RC
design, $\Var_{N}(\cbeta_{n})$ can be decomposed as follows:
\begin{align}
 \Var_{N}(\cbeta_{n}) 
  = \Var_{N}(\hbeta_{N})
   +
    \oneN
    M_{N}^{-1}(\bbeta_{0})
     Q_{N,u}(\bbeta_{0,u})
    M_{N}^{-1}(\bbeta_{0}),\label{eq:rc-variance-decomposition}
\end{align}
where
$Q_{N,u}(\bbeta_{0,u}) = N^{-1} \sum_{i \in \RN^{c}} [\pi(Y_{i},\X_{i})^{-1}\{1-\pi(Y_{i},\X_{i})\} \{Y_{i} - \mu(\U_{i}\tp\bbeta_{0,u})\}^{2} J_u\U_{i}\U_{i}\tp J_u\tp]$,
and $J_u=(\I_{d_u},\0_{d_u\times d_r})\tp\in\mathbb{R}^{d\times d_u}$
so that $\X_i=J_u\U_i$ for every $i\in\RN^c$.
\end{corollary}

Equation~\eqref{eq:rc-variance-decomposition} separates the covariance
of the full-data estimator from the additional covariance induced by
subsampling observations in $\RN^c$. Although $Q_{N,u}(\bbeta_{0,u})$
is nonzero only on the nonrare feature block, multiplication by
$M_N^{-1}(\bbeta_0)$ can couple the nonrare and rare
coordinates. Consequently, for the $j$-th rare feature estimator
$\cbeta_{n,r,j}$, the additional variance
$N^{-1}\{\e_{r,j}\tp M_N^{-1}(\bbeta_0)\}Q_{N,u}(\bbeta_{0,u}) \{M_N^{-1}(\bbeta_0)\e_{r,j}\}$,
where $\e_{r,j}\in\mathbb R^d$ is the unit vector for rare feature
coordinate $j$, need not be zero. The variances of nonrare coefficient
estimators are always inflated by subsampling. This result differs from the
rare-outcome setting of \citet{wang2021nonuniform}, where retaining
all positive cases can make variance inflation negligible for all
coefficients when the subsample size greatly exceeds the number of
positive cases. In the rare-feature setting, even when the subsample
size greatly exceeds $|\RN|$, an RC design generally cannot recover
the full-data variance, even for rare coordinates, without stronger
assumptions on feature dependence.

As discussed above, the allocation of the remaining sampling budget
over $\RN^c$ is important because it affects the variances of both nonrare
and rare feature coefficient estimators. A uniform RC design assigns
equal sampling probabilities to observations in $\RN^c$. Although
simple, it ignores differences among individual observations, and we
use it as a baseline RC design. We aim to derive optimal RC designs in
the sense of minimizing the variance of the subsample estimator or its
transformation.

\begin{proposition}
\label{prop:rc-optimal-scores}
Assume that
$|Y_i-\mu(\U_i\tp\bbeta_{0,u})|
\|L M_N^{-1}(\bbeta_0)J_u\U_i\|>0$
for every $i\in\RN^c$, where $L$ is a matrix of compatible dimensions. The RC weight function minimizing
$\tr\{\Var_{N}(L\cbeta_n)\}$ satisfies
$\varphi(Y_i,\U_i) \propto |Y_i-\mu(\U_i\tp\bbeta_{0,u})| \left\|L M_N^{-1}(\bbeta_0)J_u\U_i\right\|$,
$i\in\RN^c$. Taking $L=\I$ gives the RC A-optimal
(RC-Aopt) weight function
$\varphi^{\text{RC-Aopt}}(Y_i,\U_i) \propto |Y_i-\mu(\U_i\tp\bbeta_{0,u})| \left\|M_N^{-1}(\bbeta_0)J_u\U_i\right\|$,
$i\in\RN^c$.  Taking $L=M_N(\bbeta_0)$ gives the RC L-optimal
(RC-Lopt) weight function
$\varphi^{\text{RC-Lopt}}(Y_i,\U_i) \propto |Y_i-\mu(\U_i\tp\bbeta_{0,u})|\|\U_i\|$,
$i\in\RN^c$.
\end{proposition}

\begin{remark}
The weight functions in Proposition~\ref{prop:rc-optimal-scores} are
stated up to proportionality. A common normalizer is required to
obtain valid inclusion probabilities with the target expected
subsample size. Its implementation is described in
Section~\ref{sec:feasible}.
\end{remark}

The residual term $|Y_i-\mu(\U_i\tp\bbeta_{0,u})|$ in the RC-Aopt and RC-Lopt weight functions depends
on $\bbeta_{0,u}$ but not on the rare feature coefficients
$\bbeta_{0,r}$. Moreover, RC-Lopt does not require calculating
$M_N^{-1}(\bbeta_0)$. Its sampling probabilities therefore avoid
dependence on pilot estimates of the rare feature coefficients and do
not require estimating $M_N(\bbeta_0)$. This reduces sensitivity to
pilot instability and the computational burden.

The complete retention of $\RN$, however, can be restrictive as it requires
$n>|\RN|$. Even when it is feasible, a
small remaining budget $n-|\RN|$ can reduce estimation accuracy for all
coefficients. An RC design also requires the analyst to prespecify the rare
feature set, which may be difficult in practice, such as when binary
features have heterogeneous prevalences. In particular, a relatively
rare feature that is not included in $\RN$ is not protected by
an RC design. This can lead to inaccurate or even
nonidentifiable coefficient estimates.

\subsection{Optimal Rarity-Adaptive Designs}
\label{sec:ra-designs}

We now consider designs that attain the RA target without enforcing
complete retention. We first examine whether commonly used uniform (Uni), Aopt,
and Lopt designs achieve this target. Uni is not an RA design because
$\varphi^{\text{Uni}}(Y,\X)=1$ implies $\omega_{j,N}=h_N$ for every rare
feature.  Aopt and Lopt are derived by optimizing covariance criteria; the
following proposition presents them in the current setup for completeness.

\begin{proposition}
  \label{prop:original-optimal}
Assume that $|Y_{i} - \mu(\X_{i}\tp\bbeta_0)| \left\|L M_{N}^{-1}(\bbeta_{0}) \X_{i}\right\|>0$
for every $i=1,\dots,N$, where $L$ is a matrix of compatible dimensions.
The weight function $\varphi$ for the design
that minimizes $\tr\{\Var_N(L\cbeta_n)\}$ subject to
$0<\pi(Y_i,\X_i)\le1$ and $\sumN\pi(Y_i,\X_i)=n$ satisfies
$\varphi(Y_{i},\X_{i}) \propto |Y_{i} - \mu(\X_{i}\tp\bbeta_0)| \left\|L M_{N}^{-1}(\bbeta_{0}) \X_{i}\right\|$,
$i=1,\dots,N$. Taking $L=\I$ and
$L=M_N(\bbeta_0)$ gives the Aopt and Lopt weight functions,
respectively:
\begin{align}
  \varphi^{\text{Aopt}}(Y_{i},\X_{i})
  &\propto |Y_{i} - \mu(\X_{i}\tp\bbeta_0)|
  \left\|M_{N}^{-1}(\bbeta_{0}) \X_{i}\right\|,
  \label{eq:aopt-score}\\
  \varphi^{\text{Lopt}}(Y_{i},\X_{i})
  &\propto |Y_{i} - \mu(\X_{i}\tp\bbeta_0)| \|\X_{i}\|, \quad i=1,\dots,N.
  \label{eq:lopt-score}
\end{align}
\end{proposition}

The $j$-th rare coordinate of $M_N^{-1}(\bbeta_0)\X_i$ in Aopt's
weight function \eqref{eq:aopt-score} can be of order
$\rho_{j,N}^{-1}$ when $Z_{ij}=1$, so Aopt can exhibit the desired
rarity amplification and attain
$\omega_{j,N}\asymp\min\{h_N\rho_{j,N}^{-1},1\}$. However, Aopt's
rarity adaptivity may not be realized in practice. The Aopt weights
depend on $\bbeta_0$ through both $\mu(\X\tp\bbeta_0)$ and
$M_N^{-1}(\bbeta_0)$. A feasible implementation therefore usually
estimates these quantities from a pilot sample.
If rare features are poorly represented in
the pilot sample, both the pilot estimator for $\bbeta_{0,r}$ and the
approximation to $M_N^{-1}(\bbeta_0)$ may be unstable, thereby
distorting the Aopt sampling probabilities.

Lopt is computationally simpler than Aopt because it avoids
$M_N^{-1}(\bbeta_0)$ in Aopt's norm calculation, but this
simplification removes the inverse prevalence amplification. In the
weight function \eqref{eq:lopt-score}, when $Z_{ij}=1$, the $j$-th
rare coordinate contributes only one unit to $\|\X_i\|$ rather than an
inverse prevalence amplification. This contribution can be
dominated by the nonrare coordinates. Consequently, the Lopt weights
generally cannot achieve
$\omega_{j,N}\asymp\min\{h_N\rho_{j,N}^{-1},1\}$, and Lopt therefore
does not adequately adapt to feature rarity. The numerical results in
Section~\ref{sec:simulation} illustrate the resulting loss in
estimation accuracy.

To restore rarity adaptivity without using
$M_N^{-1}(\bbeta_0)$, we adopt the rescaled transformation
$L=D_{N}^{-2}M_{N}(\bbeta_{0})$, which leads to the RA-Lopt
subsampling design:

\begin{proposition}
\label{prop:ra-lopt}
Assume that
$|Y_i-\mu(\X_i\tp\bbeta_0)|\|D_N^{-2}\X_i\|>0$ for every
$i=1,\dots,N$. Subject to the constraints $0<\pi(Y_i,\X_i)\le 1$ and
$\sumN\pi(Y_i,\X_i)=n$, the weight function minimizing
$\tr\{\Var_{N}(D_{N}^{-2}M_{N}(\bbeta_{0})\cbeta_n)\}$ satisfies
\begin{align}
  \label{eq:ra-lopt}
  \varphi^{\text{RA-Lopt}}(Y_{i},\X_{i})
  \propto
  |Y_{i} - \mu(\X_{i}\tp\bbeta_0)|
  \left\|D_N^{-2} \X_{i}\right\|, \qquad i=1,\dots,N.
\end{align}
\end{proposition}

The matrix $D_N^{-2}$ scales the coordinate for rare feature $j$ by a
factor of order $\rho_{j,N}^{-1}$. Thus, when $Z_{ij}=1$, this
coordinate contributes a term of order $\rho_{j,N}^{-1}$ to
$\|D_N^{-2}\X_i\|$. Through this inverse prevalence scaling, RA-Lopt
can attain 
$\omega_{j,N}\asymp\min\{h_N\rho_{j,N}^{-1},1\}$. Compared with Aopt,
RA-Lopt replaces the dense multiplication $M_N^{-1}(\bbeta_0)\X$ with
the diagonal scaling $D_N^{-2}\X$. Therefore, evaluating the RA-Lopt
weights for $N$ observations costs $O(Nd)$, the same complexity order
as Lopt, whereas evaluating the Aopt weights costs $O(Nd^2)$. This
replacement avoids estimating $M_N(\bbeta_0)$, although RA-Lopt still
requires a pilot estimate of $\bbeta_0$ for the residual term.

In practice, the unknown prevalences $\rho_{j,N}$ in $D_N$ must be
replaced by their empirical counterparts. Applying this scaling to all
binary features avoids the need to prespecify the rare-feature set.

\subsection{Balanced Design}
\label{sec:bl}

All optimal designs introduced above depend on unknown quantities and
require pilot estimates for feasible implementation. Their performance
can depend strongly on pilot quality, especially when rare
features are poorly represented in the pilot sample. Implementation in
two stages is also less convenient in practice. We therefore propose
BL, whose sampling probabilities use only binary feature information
and do not depend on any unknown
parameters. As a standalone design, BL is one-step, avoids
instability from pilot estimation, and does not require a prespecified
rare feature set. BL can also be used as a pilot design or as a
component of another design when additional sampling objectives are
incorporated.

Let $\mathcal B\subset\{1,\dots,d\}$ denote the set of indices corresponding to
binary features in $\X$, and define $\bar\X_{\cdot k}=N^{-1}\sumN X_{ik}$ for
$k\in\mathcal B$.  For observation $i$, define the balance score as
\begin{equation*}
  b(\X_i) = \sum_{k \in \mathcal{B}}
  \frac{\left| X_{ik} - \bar{\X}_{\cdot k}\right|}{\bar{\X}_{\cdot k}},
  \quad i=1,\dots,N,
\end{equation*}
where $X_{ik}$ is the value of binary feature $k$. The contribution of
binary feature $k$ to $b(\X_i)$ is
$(1-\bar\X_{\cdot k})/\bar\X_{\cdot k}$ when $X_{ik}=1$ and $1$ when
$X_{ik}=0$. Thus, an observation expressing a binary feature with
lower empirical prevalence receives a larger contribution.
Each feature's contribution to the balance score has the same total
across observations at the two feature levels. Specifically, for
every $k\in\mathcal B$,
\begin{align*}
  \sum_{i:X_{ik}=1}\frac{\left| X_{ik} - \bar{\X}_{\cdot k}\right|}{\bar{\X}_{\cdot k}}
  =
  \sum_{i:X_{ik}=0}\frac{\left| X_{ik} - \bar{\X}_{\cdot k}\right|}{\bar{\X}_{\cdot k}}.
\end{align*}
In this sense, the contribution of two levels of a single binary
feature to the score is balanced. Note that the total scores
$\sum_{i:X_{ik}=1}b(\X_i)$ and $\sum_{i:X_{ik}=0}b(\X_i)$ need not be
equal because the contributions from the remaining binary features
depend on their dependence structure.

As a standalone design, BL uses the weight function
$\varphi^{\mathrm{BL}}(\X_i)\propto b(\X_i)$. A desirable property of
BL is that it attains the RA target order:
\begin{proposition}
\label{prop:bl-order}
Let $\pi^{\mathrm{BL}}(\X_{i})=\min\{h_N\tau_1b(\X_i),1\}$,
$i=1,\dots,N$, where $\tau_1>0$ is chosen so that
$\sumN\pi^{\mathrm{BL}}(\X_{i})=n$. This design satisfies
$\omega_{j,N}^{\mathrm{BL}} = \Exp\{\pi^{\mathrm{BL}}(\X)| Z_{\cdot j}=1\} \asymp \min\{h_N\rho_{j,N}^{-1},1\}$,
$j=1,\dots,d_r$.
\end{proposition}

Because its sampling probabilities do not depend on unknown
quantities, BL can be used to draw a pilot sample for estimating the
unknown quantities required by a subsequent design. The balance score
can also be incorporated into another sampling weight. For example,
combining it with Lopt gives
\begin{equation}
  \varphi^{\text{BL-Lopt}}(Y_i,\X_i)
  \propto b(\X_i)|Y_i-\mu(\X_i\tp\bbeta_0)|\|\X_i\|,
\end{equation}
which combines the prevalence information with the residual and
feature magnitude information used by classical Lopt. For the double
rarity problem in Section~\ref{sec:double-rarity}, BL can be
combined with case-control sampling to jointly emphasize rare features
and rare outcomes.

\begin{remark}
For standalone BL, the sampling probability
$\pi^{\text{BL}}(\X)=\min\{h_N\varphi^{\text{BL}}(\X),1\}$ does not
depend on the response $Y$. Consequently, one can solve the unweighted
estimating equation
$N^{-1} \sumN \{\delta_{i} \psi(Y_i, \X_i, \bbeta)\} = \0$ instead of
the weighted estimating equation in \eqref{eq:ee-sub} without
introducing estimation bias. Under a correctly specified model, the
unweighted estimator is typically more efficient than its weighted
counterpart, although it is more sensitive to model misspecification
\citep{WangWangXiong2024}. This option does not apply to
response-dependent combinations such as BL-Lopt or BL with case
control sampling.
\end{remark}

\section{Implementation and Theoretical Guarantees}
\label{sec:implementation}

Several designs developed in Section~\ref{sec:subsampling-designs}
depend on unknown quantities and therefore require pilot estimates for
implementation. The corresponding feasible inclusion probabilities are
introduced in Section~\ref{sec:feasible}.
Section~\ref{sec:union} develops a union-sample estimator that accounts
for overlap between the pilot and second-step samples and establishes
its asymptotic properties. Section~\ref{sec:double-rarity} provides
practical procedures for the double rarity setting, where rare features
and rare outcomes occur simultaneously.

\subsection{Feasible Sampling Probabilities}
\label{sec:feasible}

Because the sampling probabilities for RC-Aopt, RC-Lopt, RA-Lopt, and
BL-Lopt, as well as those for classical Aopt and Lopt, depend on
unknown quantities, we estimate the required components of $\bbeta_0$
and, when needed, $M_N^{-1}(\bbeta_0)$ from a pilot sample. The pilot
estimator $\cbeta_{\pl}$ solves the estimating equation
$\Psi_{\pl}(\bbeta)= N^{-1} \sumN \{\delta_{\pl,i}\pi_{\pl}^{-1}(Y_i,\X_i) \psi(Y_i,\X_i,\bbeta)\}=\0$,
where $\delta_{\pl,i}=\bI\{V_{\pl,i}<\pi_{\pl}(Y_i,\X_i)\}$,
$V_{\pl,i}\sim\text{Uniform}(0,1)$ is independent of $\{V_i\}_{i=1}^N$
(defined in Section~\ref{sec:poisson-subsampling}) and $\Fn$, and the
expected pilot subsample size is $n_{\pl}$. The pilot estimate of
$M_N(\bbeta_0)$ is calculated as $\check M_{\pl}(\cbeta_{\pl}) =
N^{-1}\sumN\{\delta_{\pl,i}\pi_{\pl}^{-1}(Y_i,\X_i)
\dot\mu(\X_i\tp\cbeta_{\pl})\X_i\X_i\tp\}$.

Next, we derive feasible second-step sampling probabilities.
Let $s_i$ denote the unnormalized feasible weight for observation $i$
under the selected sampling design. For example, the RA-Lopt weight based on
Equation~\eqref{eq:ra-lopt} is
$s_i=|Y_i-\mu(\X_i\tp\cbeta_{\pl})|\|D_N^{-2}\X_i\|$, $i=1,\dots,N$.  We cannot
always set $\pi(Y_i,\X_i)$ exactly proportional to $s_i$, that is,
$\pi(Y_i,\X_i)=h_N\tau_0s_i$ for some $\tau_0>0$, because this
quantity may exceed one, especially for rare observations. Instead, we
define $\pi(Y_i,\X_i)=\min\{h_N\tau_1s_i,1\}$, where $\tau_1$ is
chosen to satisfy $\sumN\pi(Y_i,\X_i)=n$. Specifically, let
$s_{(1)}\le\cdots\le s_{(N)}$ be the sorted values of $s_i$ and let
$g$ be the number of values $h_{N}\tau_{1} s_{i}$ that exceed one. Solving
$\sum_{i=1}^{N-g}h_N\tau_1s_{(i)}+g=n$ gives
$\tau_1=\{N(1-g/n)\}/(\sum_{i=1}^{N-g}s_{(i)})$. The resulting
probabilities satisfy $0\le\pi(Y_i,\X_i)\le1$ and
$\sumN\pi(Y_i,\X_i)=n$. When $g=0$, the probabilities reduce to
$\pi(Y_i,\X_i)=h_N\tau_0s_i$, where $\tau_0=N/(\sumN s_i)$, so
$\pi(Y_i,\X_i)$ remains exactly proportional to $s_i$.

Two practical concerns remain. First, calculating $\tau_{1}$ requires
sorting $s_i$'s on the full data, which can be computationally
expensive. We can use the pilot sample to provide an approximation and
define 
\begin{align*}
  \check\pi(Y_{i},\X_{i}) = \min\{h_{N}\check\varphi(Y_{i},\X_{i}), 1\},
  \;\text{ where }
  \check\varphi(Y_{i},\X_{i}) 
  = \frac{\min\{s_{i}, s_{\pl, \gamma}\}}
  {\oneN\sumN \left[\frac{\delta_{\pl,i}}{\pi_{\pl}(Y_i,\X_i)}
      \min\{s_i,s_{\pl, \gamma}\}\right]},
\end{align*}
and $s_{\pl, \gamma}$ is the sample $\gamma$-quantile of the $s_{i}$
values calculated on the pilot sample, with tuning parameter
$\gamma \in (0, 1)$. For the case $g=0$, $\tau_{0}$ is approximated
by $N \{\sumN \delta_{\pl,i}\pi_{\pl}^{-1}(Y_{i},\X_{i})s_{i}\}^{-1}$.
Second,
because the values $\check\varphi(Y_i,\X_i)$ are computed using pilot
estimates, pilot estimation error can propagate and be amplified by
inverse probability weighting, especially when the corresponding
oracle weight is close to zero. We therefore apply the defensive adjustment
$\tilde\varphi(Y,\X)=(1-\alpha_0)\check\varphi(Y,\X)+\alpha_0$, where
$0<\alpha_0<1$ is a small constant. This defensive adjustment also
enforces strict positivity of the inclusion probabilities. After these
adjustments, we use
$\{\tpi(Y_{i},\X_{i})\}_{i=1}^{N} = \{\min\{h_{N}\tilde\varphi(Y_{i}, \X_{i}),1\}\}_{i=1}^{N}$
as the working subsampling probabilities.  The truncation at one
ensures valid inclusion probabilities. Because of these practical
adjustments, the sum $\sumN \tpi(Y_{i},\X_{i})$ may not equal $n$
exactly. The above procedure applies directly to Aopt, Lopt and
RA-Lopt. For RC-Aopt and RC-Lopt, the same calculation is applied over
$i\in\RN^c$ using the remaining budget $n-|\RN|$.

\subsection{Union-Sample Aggregation}
\label{sec:union}

Beyond the construction of the feasible second-step sampling
probabilities, one may also reuse the pilot information to improve
parameter estimation.  Existing aggregation procedures achieve this in
two common ways: one pools the pilot and second-step samples and
solves a combined estimating equation \citep{wang2018optimal}, while
the other constructs a weighted average of the pilot and second-step
estimators \citep{wang2019more}. Implementation details of these two
procedures are given in Appendix~\ref{supp:ag-implement}.  These procedures can improve upon
the second-step estimator alone when the two stages contribute
distinct information. However, when both stages favor rare
observations, their samples can overlap substantially. Ignoring this overlap
can cause aggregation methods to misweight shared observations, reducing
efficiency and distorting covariance estimates.

To account for the overlap, we union the inclusion indicators and probabilities
of the two samples to
build the estimating equation, so that each selected observation
enters the estimating equation only once.  Specifically, observation
$i$ is considered in both sampling stages, with pilot sampling
probability $\pi_{\pl}(Y_i,\X_i)$ and feasible second-step sampling
probability $\tpi(Y_i,\X_i)$, and with corresponding inclusion
indicators $\delta_{\pl,i}$ and $\tdelta_i=\bI\{V_i<\tpi(Y_i,\X_i)\}$. The union inclusion
indicator is
$\tdelta_{\ag,i} = \bI(\delta_{\pl,i}=1 \text{ or
} \tdelta_i=1) = \delta_{\pl,i}+(1-\delta_{\pl,i})\tdelta_i$,
$i=1,\dots,N$.  The feasible union estimator $\cbeta_{\ag}$ is the
solution to
\begin{align}
  \label{eq:ag-ee}
  \tilde\Psi_{\ag}(\bbeta) =
  \oneN 
  \sumN \left\{\frac{\tdelta_{\ag,i}}{\tpi_{\ag}(Y_{i}, \X_{i})}
  \psi(Y_{i},\X_{i},\bbeta) \right\}  = \0,
\end{align}
where the working denominator
$\tpi_{\ag}(Y_{i}, \X_{i}) = \pi_{\pl}(Y_i,\X_i) +\{1-\pi_{\pl}(Y_i,\X_i)\}\tpi(Y_i,\X_i)$.
This working denominator $\tpi_{\ag}(Y_i, \X_i)$ is
preferable to pooling the samples from both stages and using
$\pi_{\pl}(Y_i, \X_i)$ and $\tpi(Y_i, \X_i)$ separately, because
its reciprocal acts as a smoothed weight that can improve
estimation efficiency \citep{kim2013weighting, kim2023note}. The
empirical advantage of the proposed union-sample estimator is
illustrated in Figure~\ref{fig:rmse-compare-ag}.

Because the feasible second-step probability $\tpi(Y_i,\X_i)$ depends
on the pilot sample, the working denominator $\tpi_{\ag}(Y_i,\X_i)$ in Equation~\eqref{eq:ag-ee} is
not the exact inclusion probability.  Consequently,
$\tilde\Psi_{\ag}(\bbeta)$ is not centered at $\Psi_0(\bbeta)$.  To
study the asymptotic behavior of $\cbeta_{\ag}$, we introduce an
oracle version of the union procedure by setting the pilot quantities
in the feasible sampling probabilities to their limits. For example,
the consistent pilot estimator $\cbeta_{\pl}$ is replaced by the true
parameter $\bbeta_0$. Throughout, the
superscript $0$ denotes these oracle quantities. For instance, let
$\pi^0(Y,\X)$ denote the oracle counterpart of $\tpi(Y,\X)$. The
corresponding oracle union probability is
$\pi_{\ag}^0(Y,\X) = \pi_{\pl}(Y,\X) + \{1-\pi_{\pl}(Y,\X)\}\pi^0(Y,\X)$.
Formal definitions of $\Omega_{\ag,N}^0$ and
$\Lambda_{\ag,N}^0(\bbeta)$ are provided in Appendix~\ref{supp:thm-ag}.
To control the additional error introduced
by using the feasible second-step sampling probabilities, we impose
the following assumption:

\begin{assumption}
  \label{asmp:ag}
  (i) The pilot estimator $\cbeta_{\pl}$ is $\sqrt{n_{\pl}}$-consistent, i.e.,
  $\|\cbeta_{\pl} - \bbeta_{0}\|=O_{p}(n_{\pl}^{-1/2})$. (ii) The feasible
  and oracle second-step sampling probabilities satisfy
  $|\tpi(Y,\X) - \pi^{0}(Y,\X)| \le h_{N} L_{\pl}(Y,\X) \|\cbeta_{\pl} -
  \bbeta_{0}\|$ for a nonnegative function $L_{\pl}(Y,\X)$. Let
  $h_{\pl,N}=n_{\pl}/N$ and assume
  \begin{align*}
    \sup_N h_{\pl,N}^{-1}
    \Exp\{\pi_{\pl}(Y,\X)L_{\pl}(Y,\X)G_N(Y,\X,\bbeta_0)\}
    +
    \sup_N
    \Exp\{L_{\pl}(Y,\X)T_N(Y,\X,\bbeta_0)\}
    <\infty,
  \end{align*}
  where $G_N(Y,\X,\bbeta_0)$ and $T_N(Y,\X,\bbeta_0)$ are nonnegative
  envelope functions defined in
  Appendix~\ref{supp:thm-ag}.
\end{assumption}

\begin{theorem}
  \label{thm:ag}
Assume the oracle sampling probability $\pi^{0}_{\ag}(Y,\X)$ satisfies
Assumptions~\ref{asmp:samp-prob-1}-\ref{asmp:samp-prob-2}. In
particular, assume 
$\Omega_{\ag,N}^{0}M_{r,N}(\bbeta_{0})(\Omega_{\ag,N}^{0})^{-1}
\xrightarrow{p} M_{\ag}^{0}(\bbeta_{0})$  and
$\Omega_{\ag,N}^{0}D_{N}^{-1}\Lambda_{\ag,N}^{0}(\bbeta_{0})
    D_{N}^{-1}\Omega_{\ag,N}^{0}
    \xrightarrow{p} \Lambda_{\ag}^{0}(\bbeta_{0})$,
where $M_{\ag}^{0}(\bbeta_{0})$ is finite and nonsingular and
$\Lambda_{\ag}^{0}(\bbeta_{0})$ is finite and positive definite.
Under Assumptions~\ref{asmp:regularity-1}--\ref{asmp:regularity-3}
and~\ref{asmp:ag}, if $h_{N} \sqrt{n_{\pl}} \to \infty$,
$\sqrt{h_{\pl,N}} / h_{N} \to 0$, and $N h_{N} \rho_{j,N}\to\infty$ for
$j=1,\dots,d_{r}$, we have
\begin{align*}
  \sqrt{N}\Omega^{0}_{\ag,N}D_{N}
  (\cbeta_{\ag}-\bbeta_{0})
\xrightarrow{d}
\Nor\left(\0, \{M_{\ag}^{0}(\bbeta_{0})\}^{-1}
  \Lambda_{\ag}^{0}(\bbeta_{0}) 
  [\{ M_{\ag}^{0}(\bbeta_{0})\}^{-1}]\tp \right).
\end{align*}
\end{theorem}

Theorem~\ref{thm:ag} shows that the proposed feasible implementation
does not change the oracle asymptotic distribution. The condition
$h_N\sqrt{n_{\pl}}\to\infty$ controls the error from estimating the
second-step probabilities, while $\sqrt{h_{\pl,N}}/h_N\to0$ implies that
the pilot sampling rate $h_{\pl,N}$ is of a smaller order than the
second-step sampling rate $h_{N}$, so the oracle union sampling rate
remains of the same order of magnitude as $h_{N}$.
Nevertheless, the union
estimator can improve rare-feature estimation efficiency when the
pilot sample contributes rare observations not captured by the
second-step sample. This effect is reflected in the rare-feature
scaling component $\omega_{\ag,j}^{0}$ in $\Omega_{\ag,N}^{0}$, where
\begin{align*}
   \omega_{\ag,j}^{0}
  =
  \Exp\{\pi_{\ag}^{0}(Y,\X)\mid Z_{\cdot j}=1\}
  =
  \Exp[\pi_{\pl}(Y,\X)+\{1-\pi_{\pl}(Y,\X)\}\pi^{0}(Y,\X)
  \mid Z_{\cdot j}=1].
\end{align*}
It determines the convergence rate
$(N\rho_{j,N}\omega_{\ag,j}^{0})^{-1/2}$ for the $j$-th rare-feature
coordinate of $\cbeta_{\ag}$. The term $\omega_{\ag,j}^{0}$ includes
the additional pilot contribution beyond that from the oracle
second-step sample alone. This contribution can be substantial when
the pilot design, such as BL, strongly protects rare
observations; see the empirical evidence in
Section~\ref{sec:simulation}.

\subsection{Double Rarity Problem in Logistic Regression}
\label{sec:double-rarity}

A particularly challenging scenario arises when rare features and a rare
binary outcome occur simultaneously, as in many binary classification
problems involving massive data. Let $Y\in\{0,1\}$ be the binary outcome, and let
$N_1=\sumN Y_i$ and $N_0=N-N_1$ denote the numbers of cases and
controls, respectively, with $N_1/N_0\to 0$. Existing analyses, such as
\citet{wang2021nonuniform}, show that, without considering rare
features, the convergence rate of the full-data estimator is tied to
$N_1$ rather than $N$. A sampling method must therefore preserve both
rare outcomes and rare features. To address this ``double rarity''
problem, we propose the following two practical procedures:

\begin{itemize}[leftmargin=*]
\item When $N_1\ll n$, we recommend negative subsampling: retain all
cases and apply BL within the control group using the remaining
budget $n-N_1$. Specifically, assign an inclusion probability of
one to every case and choose the control inclusion probabilities
proportional to $b(\X_i)$, with the normalizer selected so that
$\sum_{i:Y_i=0}\pi(Y_i,\X_i)=n-N_1$.
\item Otherwise, we combine the balance score with a case-control (CC)
sampling scheme. Specifically, set
$\varphi(Y_i,\X_i)\propto b(\X_i)/(2N_{Y_i})$,
$i=1,\dots,N$, where $N_{Y_i}=N_1$ when $Y_i=1$ and $N_{Y_i}=N_0$ when
$Y_i=0$. Section~\ref{sec:iris} illustrates its application.
\end{itemize}

\section{Numerical Results}
\label{sec:simulation}

\subsection{Simulation Setup}

We use simulations to illustrate the theoretical results and compare
the proposed designs with existing
methods. Table~\ref{tab:compared-methods} lists the competing
methods. Methods requiring a pilot are referred to as two-step
methods. The pilot estimates $\cbeta_{\pl}$ and
$\check{M}_{\pl}^{-1}(\cbeta_{\pl})$ are obtained from a pilot sample
drawn using either the Uni or BL design, denoted by (Uni) and (BL),
respectively. RA-Lopt estimates the rare-coordinate entries of $D_N$
using the square roots of corresponding full-data means. The final
estimate for each two-step method is the union-estimator. For sample
size comparability, one-step methods use expected sample size
$n+n_{\pl}$.
RC methods require a sampling budget large enough to retain all rare
observations and sample nonrare observations. We therefore examine
settings where RC is feasible as well as settings where the rare-observation
count exceeds the available budget.

\begin{table}[tb]
\centering
\caption{Compared Designs.}
    \label{tab:compared-methods}
\small
\begin{tabular}{llll}
\toprule
   Design & Description & Weight function $\varphi(Y,\X)$ proportional to \\
\midrule
 Uni (baseline) &  Uniform sampling & $1$ \\
  BL (new) & Balance score & $b(\X)$ \\
  RC-Uni{} (baseline) & Uniformly sampling from $\RN^{c}$ &  $1$ for $\RN^{c}$ \\
  Aopt &  Traditional A-optimal & $|Y - \mu(\X\tp \cbeta_{\pl})| \|\check{M}_{\pl}^{-1}(\cbeta_{\pl})\X\|$ \\
  Lopt &  Traditional L-optimal & $|Y - \mu(\X\tp\cbeta_{\pl})| \|\X\|$ \\
  RC-Aopt (new) &  Aopt sampling from $\RN^{c}$ & $|Y_{i} - \mu(\U_{i}\tp\cbeta_{\pl,u})| \left\|\check{M}_{\pl}^{-1}(\cbeta_{\pl}) J_{u} \U_{i}\right\|$ for $\RN^{c}$ \\
  RC-Lopt (new) &  Lopt sampling from $\RN^{c}$ &  $|Y_{i} - \mu(\U_{i}\tp\cbeta_{\pl,u})| \left\|\U_{i}\right\|$ for $\RN^{c}$ \\
  RA-Lopt (new) &  Rescaled L-optimal & $|Y - \mu(\X\tp \cbeta_{\pl})| \| D^{-2}_{N} \X\|$ \\
  BL-Lopt (new) &  Balance score $\times$ L-optimal & $ |Y - \mu(\X\tp\cbeta_{\pl})| \|\X\| b(\X)$ \\
\bottomrule
\end{tabular}
\end{table}

We consider logistic regression as a representative model from the
class of estimating equations studied in this paper. The full-data
sample size is $N=2.5\times 10^{5}$. Observations are independent and
identically distributed. The nonrare features $\U$ follow a
multivariate normal distribution with zero means, unit variances, and
$\operatorname{Corr}(U_i,U_j)=0.5^{|i-j|}$,
$i,j \in \{1,\dots,d_{u}\}$. The rare features $\Z$ have marginal
distribution $\operatorname{Bernoulli}(\rho)$ and are partitioned into
independent blocks of size at most five, with within-block pairwise
correlation $0.3$.  The main comparisons use $d_u\in\{15,50\}$,
$d_r\in\{2,10,20\}$, $\rho\in\{1\%,0.2\%, 0.1\%\}$, pilot sizes
$n_{\pl}\in\{300,1{,}000\}$, and second-step sizes ranging from
$2{,}000$ to $50{,}000$. All true coefficients, including the
intercept, are set to $0.5$. The defensive parameter is 
$\alpha_{0}=0.02$.  The additional convergence study in
Figure~\ref{fig:convergence-pilot-ra}(A) uses $d_u=5$ and $d_r=1$ and
varies $N$ from $10^{4}$ to $10^{6}$ while holding $N\rho=500$.

Under these settings, the outcome is relatively balanced 
between $0$ and $1$.  We generated $1,000$ full datasets and performed
each subsampling method once per full dataset. Estimation accuracy is
measured by the empirical root mean squared error (RMSE) across
$1,000$ replications, reported separately for rare and nonrare
feature coefficients. We also report the empirical identifiability
failure rates (\%), defined as the percentage of replicates
lacking an estimable coefficient for a given feature. For two-step methods, pilot and second-step failures are reported separately.

\subsection{Simulation Results}

We first illustrate the convergence rate established in
Theorem~\ref{thm:full}. Figure~\ref{fig:convergence-pilot-ra}(A) shows
that the variance of the rare-feature coefficient estimator scales as
$(N\rho)^{-1}$ rather than $N^{-1}$.  Table~\ref{tab:failure-prob}
illustrates the identifiability problem identified in
Proposition~\ref{prop:lower-bound}. At $\rho=0.1\%$ and $n=3{,}500$,
failure rates can be large and increase sharply with $d_r$. Using a BL
pilot eliminates pilot failures. However, Lopt(BL) retains
nonnegligible second-step failure rates because Lopt does not
adequately protect rare features, whereas Aopt(BL) has no
identifiability failures.

Figure~\ref{fig:convergence-pilot-ra}(B) illustrates the effect of pilot
design in a setting where all displayed methods have zero failures in the 1,000
replications. Across the displayed sample sizes and methods, the
rare-feature RMSE under a Uni pilot is much larger than that under
a BL pilot, resulting in larger overall RMSE. Thus,
even when identifiability failures are absent, the finite-sample
performance of a two-step method depends strongly on the pilot design
used to construct its second-step sampling probabilities.

\begin{figure}[t]
  \centering
  \includegraphics[width=1\textwidth]
  {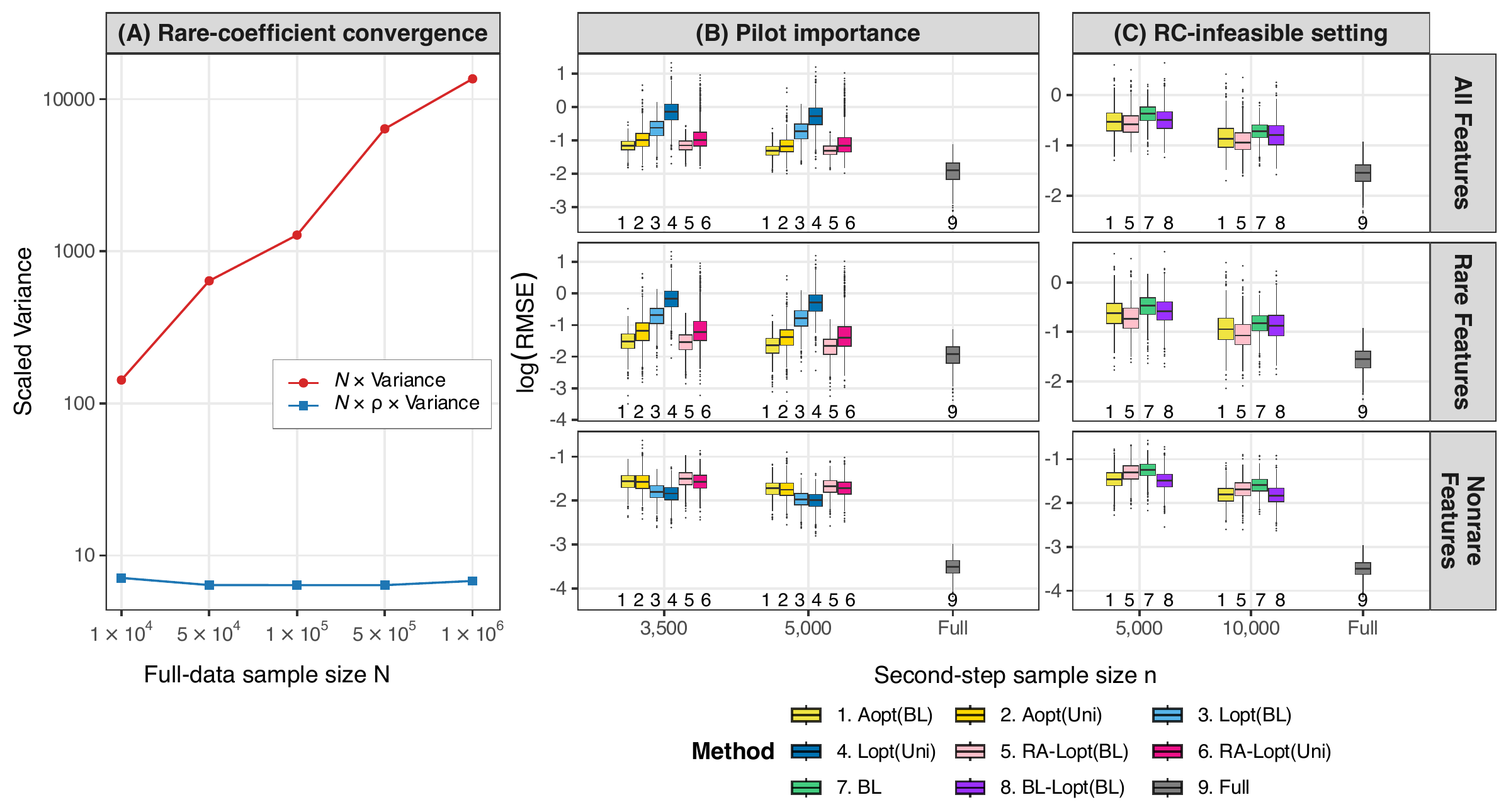}
  \caption{(A): Full-data estimator variance with $d_u=5$ and $d_r=1$,
varying $N$ from $10^{4}$ to $10^{6}$ while fixing $N\rho=500$. (B):
Comparing BL and Uni pilots in a regime with $n_{\pl}=1{,}000$,
$d_u=15$, $d_r=5$, and $\rho=1\%$.  (C): Comparing RA methods in a
regime with $n_{\pl}=300$, $d_u=15$, $d_r=10$, and $\rho=1\%$. Rows in
(B) and (C) show log(RMSE) for all, rare-feature, and nonrare feature
coefficients, respectively, and ``Full'' denotes the full-data
estimator.}
  \label{fig:convergence-pilot-ra}
\end{figure}

\begin{table}[tp]
\centering
\caption{Empirical identifiability failure rates (\%) at $\rho=0.1\%$,
and $n=3,500$. Method suffixes ``(BL)'' and ``(Uni)'' identify the pilot design. For two-step methods, entries $y(x)$ give the
second-step (pilot) failure rates.}
 \label{tab:failure-prob}
\begin{tabular}{r rrrrr}
\toprule
$d_r$ & \multicolumn{5}{c}{Method} \\
\cmidrule(lr){2-6}
 & Uni & Aopt(BL) & Aopt(Uni) & Lopt(BL) & Lopt(Uni) \\
\midrule
\addlinespace
\multicolumn{6}{c}{$n_{\pl}=300$} \\
\addlinespace
2 & 4.5 & 0(0) & 0(87.8) & 4.8(0) & 0(87.8) \\
10 & 17.6 & 0(0) & 0(99.2) & 23.7(0) & 0(99.2) \\
20 & 33.7 & 0(0) & 0(100) & 36.2(0) & 0(100) \\
\addlinespace
\multicolumn{6}{c}{$n_{\pl}=1,000$} \\
\addlinespace
2 & 2.8 & 0(0) & 0(58.6) & 6.4(0) & 0.2(58.6) \\
10 & 9.8 & 0(0) & 0(92.8) & 26.6(0) & 0(92.8) \\
20 & 20.9 & 0(0) & 0(99.2) & 44.3(0) & 0(99.2) \\
\bottomrule
\end{tabular}
\end{table}

\begin{table}[!tb]
\centering
\caption{Mean number of union-sample observations with $Z_{\cdot 1}=1$
under settings with $d_u=15$, $\rho=1\%$ and
$n_{\pl}=1{,}000$. Averages are over successful replications. The
expected full-data count of $Z_{\cdot 1}=1$ is $2,500$.}
\label{tab:rare-count-sim}
\begin{tabular}{r rr rr rr r}
\toprule
$n$ & \multicolumn{7}{c}{Method} \\
\cmidrule(lr){2-8}
 & BL & Uni & Aopt(BL) & Aopt(Uni) & Lopt(BL) & Lopt(Uni) & RA-Lopt(BL)  \\
\midrule
\addlinespace
\multicolumn{8}{c}{$d_r=2$} \\
\addlinespace
2,000 & 967 & 30 & 580 & 406 & 340 & 33 & 665 \\
10,000 & 2,454 & 110 & 1,248 & 1,148 & 410 & 124 & 1,375 \\
30,000 & 2,457 & 311 & 1,666 & 1,584 & 587 & 351 & 1,772 \\
\addlinespace
\multicolumn{8}{c}{$d_r=10$} \\
\addlinespace
2,000 & 346 & 30 & 267 & 223 & 135 & 33 & 316 \\
10,000 & 1,188 & 109 & 827 & 776 & 210 & 123 & 949 \\
30,000 & 2,086 & 309 & 1,419 & 1,326 & 398 & 349 & 1,557 \\
\bottomrule
\end{tabular}
\end{table}

\begin{figure}[t]
  \centering
  \includegraphics[width=\textwidth]
  {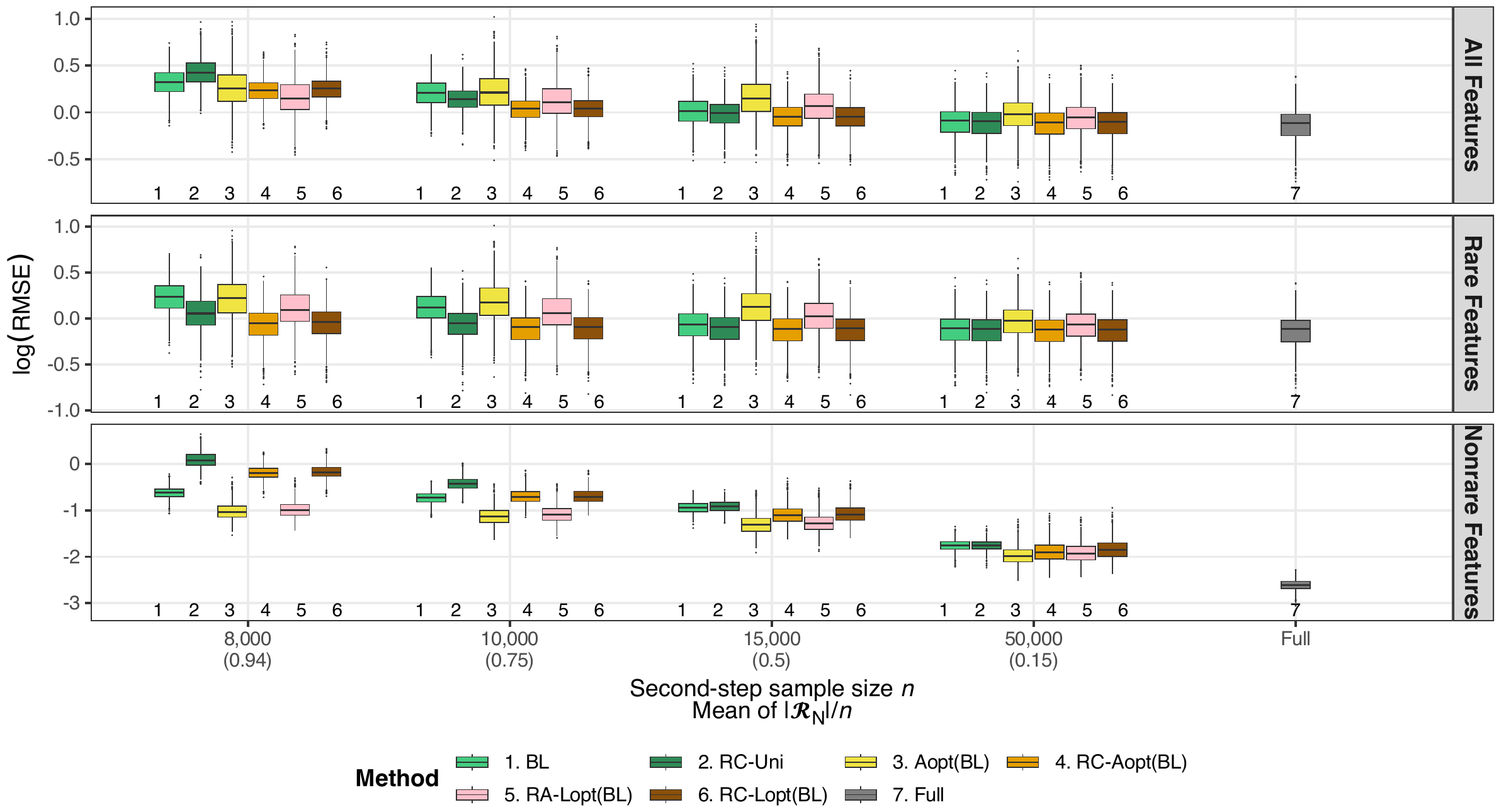}
  \caption{Rare-census (RC) trade-off under the setting
 $n_{\pl}=1{,}000$, $d_u=50$, $d_r=20$, and
$\rho=0.2\%$. The x-axis gives the second-step sample size $n$, with
the empirical mean of $|\mathcal R_N|/n$ shown in parentheses.}
  \label{fig:rc-ra-tradeoff}
\end{figure}

Table~\ref{tab:rare-count-sim} helps explain the differences in
identifiability and estimation accuracy. Relative to a Uni
pilot, a BL pilot increases the mean union-sample count with
$Z_{\cdot1}=1$ for Aopt and Lopt. Through union aggregation, it also
partly offsets Lopt’s lack of explicit rare-feature protection, as
discussed after Theorem~\ref{thm:ag}. Methods retaining more rare
observations generally have fewer identifiability failures.

Having examined pilot choice and rare-observation counts, we next
compare the estimation accuracy of RA designs, using a BL pilot for
all two-step methods. Figure~\ref{fig:convergence-pilot-ra}(C) compares
overall, rare-feature, and nonrare feature RMSE in a setting where
RC is infeasible. At both $n=5{,}000$ and
$n=10{,}000$, RA-Lopt(BL) has the lowest mean overall and rare-feature
RMSE, whereas BL-Lopt(BL) has the lowest mean nonrare feature RMSE,
closely followed by Aopt(BL). Although the BL method provides strong
representation of rare observations, incorporating residual information
can further improve estimation accuracy by prioritizing observations
that are informative for the regression coefficients.

Figure~\ref{fig:rc-ra-tradeoff} compares RC and other RA designs when
RC is feasible. Their relative accuracy depends on the fraction of
the second-step budget consumed by the rare-observation census,
$|\RN|/n$. The mean of this fraction decreases from approximately
0.94 to 0.15 as $n$ increases. At the highest budget pressure, RC methods have a
substantial loss in nonrare feature RMSE, while RA-Lopt(BL) has the
lowest overall RMSE at $n=8{,}000$.  At budget pressure 0.75 or lower,
optimal RC methods outperform the other RA methods on overall RMSE;
however, no design dominates across all regimes.

Finally, Figure~\ref{fig:rmse-compare-ag} compares the union, pooled,
weighted, and second-step estimators for RA-Lopt(BL) under sample
overlap.  The overlap rate is
$(\sumN \delta_{\pl,i}\tdelta_i)/(\sumN \tdelta_i)$, the proportion of
second-step observations appearing in the pilot sample under each
replication, and we report its average across replications. The
estimated MSE is computed from the finite-$N$ sandwich covariance
expressions and then averaged across replications. All three
aggregation methods improve upon the second-step estimator, but the
union estimator has the lowest empirical MSE and its estimated
MSE is close to the empirical MSE. The pooled and weighted estimators
substantially underestimate MSE.

\begin{figure}[!t]
    \centering    \includegraphics[width=\textwidth]{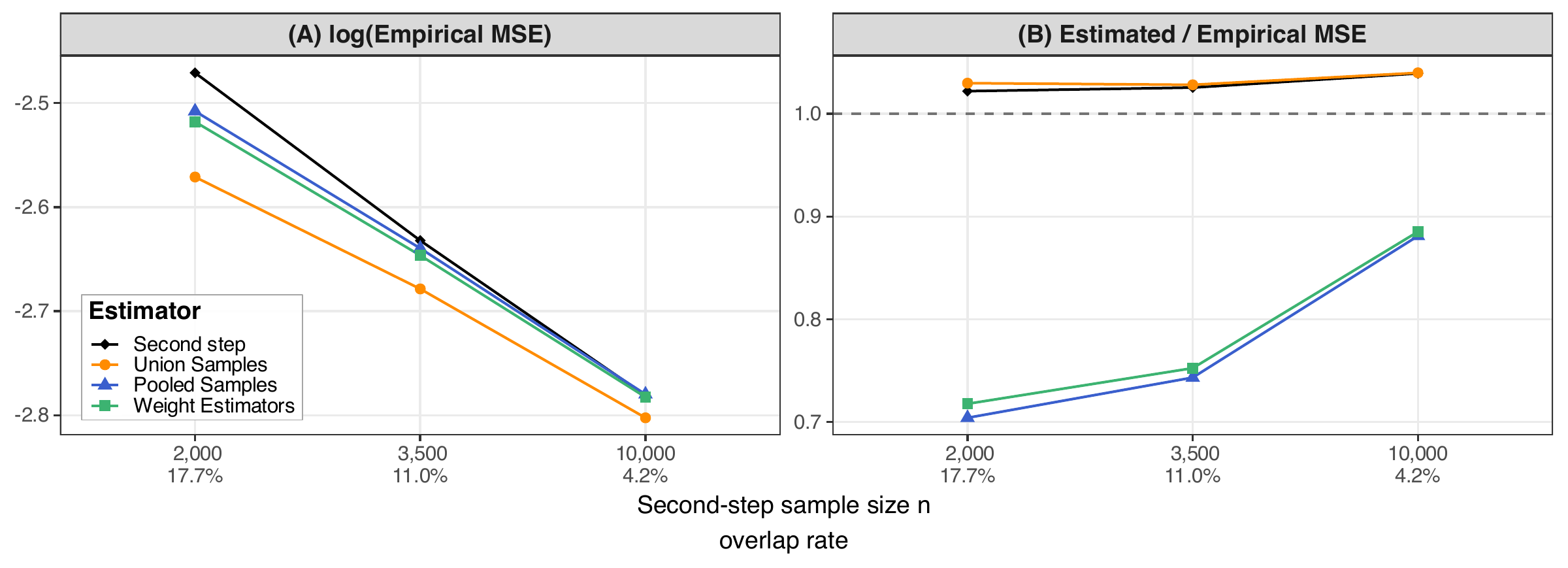}
\caption{Comparison of aggregation estimators for the setting with
$n_{\pl}=1,000$, $d_{u}=5$, $d_{r}=2$, and
$\rho=0.1\%$. The x-axis represents the expected second-step subsample
size $n$ and the corresponding sample overlap rate. (A) Log empirical MSE. (B) Ratio of estimated to
empirical MSE. Values below one indicate underestimation.}
    \label{fig:rmse-compare-ag}
\end{figure}

Taken together, these simulations suggest the following
recommendations within the settings considered. (1) When simplicity
and numerical stability are the primary considerations, or when the
quality of a pilot estimator cannot be assessed reliably, BL is a
natural default. (2) For two-step methods, a BL pilot is recommended, and the
final estimator should be fitted to the union of the pilot and
second-step samples. (3) When RC sampling is feasible and leaves
adequate budget for nonrare observations, RC-Aopt(BL) provides the
strongest overall performance, whereas RC-Lopt(BL) performs similarly
but has lower computational cost. (4) When RC sampling is infeasible,
RA-Lopt(BL) is preferable for overall and rare-feature
accuracy. Aopt(BL) remains competitive, particularly for
nonrare coefficients, but is more computationally demanding. Thus, the
choice of sampling design should be determined by the budget pressure
and the primary inferential target rather than by a universal ranking
of the methods.

\section{Real Data Applications}
\label{sec:real}

This section illustrates how the proposed subsampling designs can be
used in large-scale regression applications. We consider two datasets:
(i) the Covertype dataset, which contains categorical covariates with
low and heterogeneous prevalences, and (ii) an eye disease registry,
which represents a double-rarity setting with both an extremely rare
outcome and infrequent binary covariates. For each application, we
evaluate estimation accuracy relative to the corresponding full-data
estimator.

\subsection{Covertype Dataset}
\label{sec:covertype}

The Covertype dataset \citep{covertype_31} contains $N=581{,}012$
cartographic observations with 10 continuous environmental covariates
and 40 binary soil-type indicators. Following the dataset
documentation, we aggregate the 40 indicators into 11
geologic-climatic (GC) groups. Several groups remain extremely sparse
after aggregation, as shown in
Table~\ref{tab:covertype-gc-prevalence}. We define the response as an
indicator of lodgepole pine, which accounts for $48.76\%$ of the
original observations.  Under this coding, all $602$ observations in
GC\_5 have $Y=0$, so its observations are removed before model
fitting. The remaining GC groups are dummy-coded, with GC\_10, the
most prevalent group, used as the reference category.  These
preprocessing choices preserve interpretability while exposing the
practical difficulty of estimating coefficients for rare groups in
large-scale regression.

\begin{table}[tb]
  \centering 
  \caption{Counts and empirical prevalences of the GC groups in the
Covertype data. The full sample size is $N=581{,}012$. GC\_5 is removed
from model fitting, and GC\_10 is used as the reference category for
dummy coding.}
  \label{tab:covertype-gc-prevalence}
  \begin{tabular}{lrr}
    \toprule
    GC group & Count & Empirical prevalence (\%) \\
    \midrule
    GC\_2                 &        284 &  0.049 \\
    GC\_5 (removed)       &        602 &  0.104 \\
    GC\_3                 &      1{,}147 &  0.197 \\
    GC\_7                 &      1{,}899 &  0.327 \\
    GC\_6                 &      6{,}267 &  1.079 \\
    GC\_8                 &     14{,}118 &  2.430 \\
    GC\_1                 &     35{,}947 &  6.187 \\
    GC\_11                &     40{,}437 &  6.960 \\
    GC\_9                 &     91{,}125 & 15.684 \\
    GC\_4                 &     92{,}446 & 15.911 \\
    GC\_10 (reference)    &    296{,}740 & 51.073 \\
    \bottomrule
  \end{tabular}
\end{table}

We evaluate several promising designs to provide guidance for datasets
with this structure. The two-step methods use a BL pilot of size
$n_{\pl}=1{,}000$, and second-step subsample sizes
$n\in\{2{,}000,5{,}000,25{,}000,50{,}000\}$. Each setting is
replicated $1{,}000$ times. Performance is evaluated using RMSE
relative to the full-data estimator.  For the RC methods, the
rare-feature thresholds are $0.1\%$, $0.25\%$, $1.5\%$, and $5\%$,
respectively, retaining (GC\_2); (GC\_2, GC\_3); (GC\_2, GC\_3, GC\_6,
GC\_7); and (GC\_2, GC\_3, GC\_6, GC\_7, GC\_8). This calibration
keeps the retained rare observations below the second-step sample
size, but a relatively sparse group may not receive census retention
at smaller $n$.  In this application, RC-Uni and RC-Lopt(BL) exhibit
large identifiability failure rates, so we do not recommend them for
this dataset.  In contrast, methods such as BL, Aopt(BL), and
RA-Lopt(BL) do not require a prespecified rarity threshold and are not
exposed to this particular source of sensitivity. For comparability
across settings, the rare-feature RMSE is always evaluated on the
fixed set (GC\_2, GC\_3, GC\_6, GC\_7, GC\_8).

Figure~\ref{fig:rmse-compare-real}~(A) shows that Aopt(BL) has the
lowest overall and nonrare feature RMSE at every sample
size. RA-Lopt(BL) has the lowest rare-feature RMSE at $n=2{,}000$,
$5{,}000$, and $25{,}000$. At $50{,}000$, all five rare
features have entered the RC census, and RC-Aopt(BL) has the lowest
rare-feature RMSE.  Thus, Aopt(BL) is preferred when overall accuracy
is primary, whereas RA-Lopt(BL) is competitive when rare-feature
accuracy is the focus.

\begin{figure}[!tb]
  \centering
  \includegraphics[width=\textwidth]{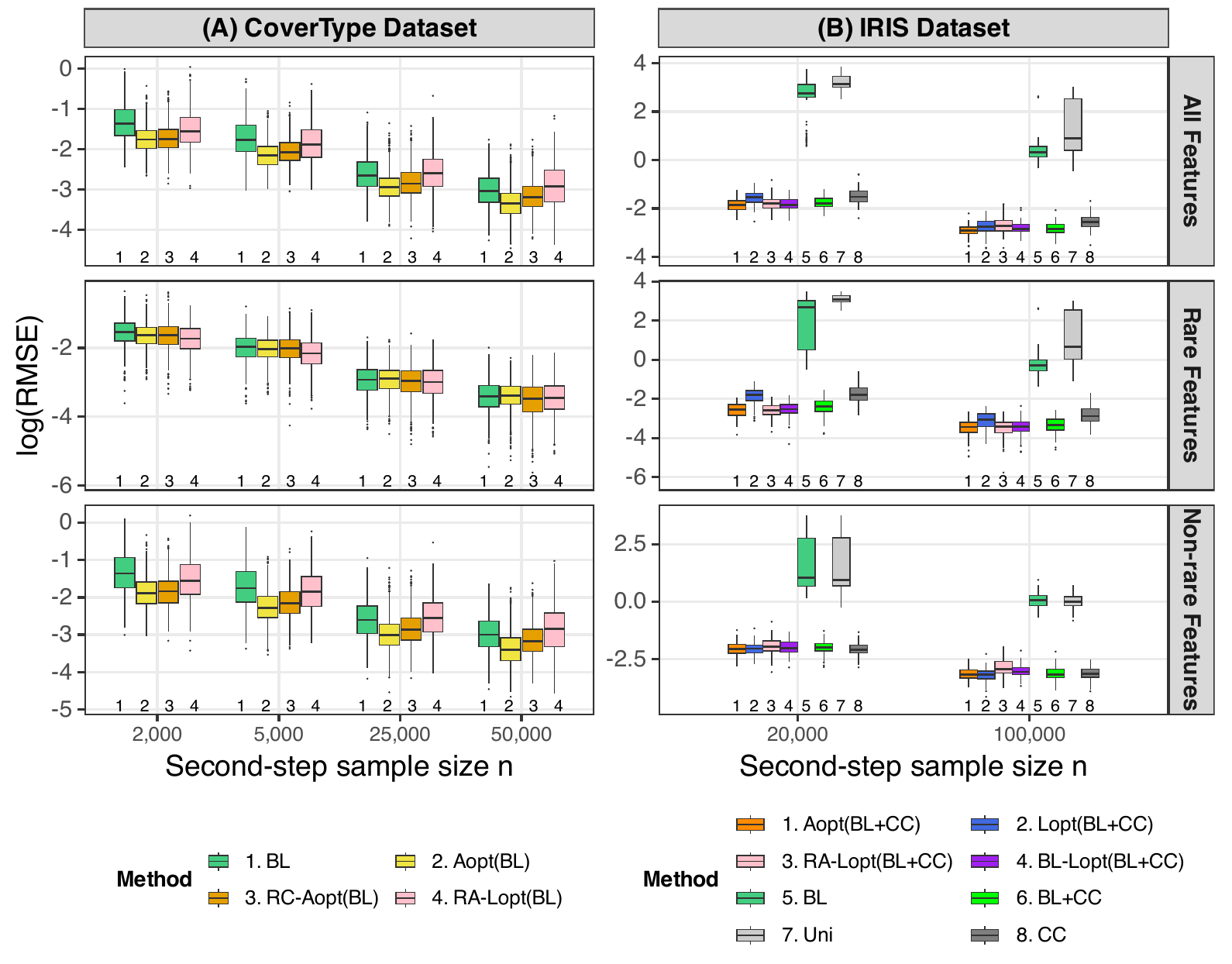}
\caption{Boxplots of RMSEs for the Covertype and IRIS experiments.
Panel (A) shows the Covertype dataset and Panel (B) shows the IRIS
dataset. In both datasets, $n_{\pl}=1{,}000$. In the
legend, ``(BL)'' and ``(BL+CC)'' indicate the BL and BL+CC pilot,
respectively.}
\label{fig:rmse-compare-real}
\end{figure}

\subsection{IRIS Registry: A Double-Rarity Case Study}
\label{sec:iris}

In this case study, we analyze a large-scale dataset that exhibits the
double-rarity structure described in
Section~\ref{sec:double-rarity}. The dataset is from the IRIS
(Intelligent Research in Sight) Registry, an electronic health
record-based eye disease registry in the United States. The binary
outcome $Y$ indicates thyroid eye disease, defined by at least two
ophthalmology visits coded with ICD-9 code 242.00 or ICD-10 code
E05.00. The full dataset contains $47{,}872{,}555$ observations, among
which only $0.086\%$ ($N_1=41{,}211$) are cases. We consider 14 features, including rare binary features
for Asian race ($2.902\%$) and other race ($0.882\%$). We do not
consider RC designs because retaining every observation expressing a
rare feature requires a substantially larger subsample than those
considered here.  Two-step methods use a pilot of size
$n_{\pl}=1{,}000$. We use the combined BL and CC pilot, denoted by
BL+CC, described in Section~\ref{sec:double-rarity}.  The final
estimator for two-step methods is the union-sample estimator. For
comparability, one-step methods use expected sample size
$n+n_{\pl}$. Each method was replicated 100 times by repeatedly
drawing subsamples from the full dataset.

Figure~\ref{fig:rmse-compare-real}~(B) shows that methods addressing
only one source of rarity perform poorly. BL does not account for the
extremely rare outcome, whereas CC 
does not protect rare feature information. Classical Lopt addresses
outcome imbalance through the residual-based sampling probabilities
but does not sufficiently amplify rare observations. In contrast, Aopt,
RA-Lopt, and BL-Lopt combine protection against rare outcomes with
preservation of rare-feature information and therefore achieve
substantially lower RMSEs when equipped with BL+CC pilots. Overall,
the results demonstrate that subsampling in large-scale logistic
regression requires simultaneous attention to outcome rarity and
feature rarity, and both the pilot and second-step designs should
account for them.

\section{Discussion}
\label{sec:discussion}

This paper studies subsampling for regression models with rare
features. We show that the estimators of coefficients associated with
rare features have convergence rates determined by the number of
observations expressing the corresponding features, rather than by the
full-data sample size. This finding explains why rare features remain
difficult even in massive datasets and why subsampling can further
aggravate the problem when rare-feature observations are
underrepresented. Motivated by this insight, we develop an RA subsampling
framework. We propose optimal RC designs, an
inverse-prevalence modification of Lopt, and a balanced design. We
also develop a union-sample aggregation procedure that incorporates
pilot information while accounting for overlap between pilot and
second-step samples. Simulation studies and applications to real
datasets show that the proposed methods reduce identifiability
failure, improve estimation efficiency, and provide more accurate
covariance estimation than existing subsampling approaches.

These results should be interpreted within the modeling framework
studied in this paper. In particular, the estimating equations are
assumed to be correctly specified, and the dimension of features
is fixed. The rare features are binary covariates, including
dummy-coded categorical features; highly correlated rare-feature
structures are not considered.

Several avenues remain for future research. First, since the
challenges associated with rare features stem directly from the
covariate structure, they warrant investigation beyond the model
framework considered here. A natural direction is to extend the
proposed methods to broader statistical settings, such as clustered,
longitudinal, and survival data models. Second, the impact of rare
features on inferential tasks beyond point and interval estimation,
including hypothesis testing, model selection, and variable screening,
remains to be explored. Third, high-dimensional regimes, particularly
settings where the number of rare features increases as the sample
size grows, will require new theory. Fourth, extending these scaling
and sampling schemes to non-parametric regression or tree-based
machine learning algorithms, where rare features are easily neglected
by default splitting rules, is another direction. Finally, the effects
of oversampling observations carrying rare features warrant further
study. A systematic comparison of estimation efficiency between such
oversampling designs and our proposed subsampling methods remains an
open question.

\bibliography{ref}
\bibliographystyle{natbib}

\clearpage

\doublespacing

\appendix
\setcounter{assumption}{0}
\renewcommand{\theassumption}{A.\arabic{assumption}}
\setcounter{lemma}{0}
\renewcommand{\thelemma}{A.\arabic{lemma}}
\setcounter{proposition}{0}
\renewcommand{\theproposition}{A.\arabic{proposition}}
\setcounter{corollary}{0}
\renewcommand{\thecorollary}{A.\arabic{corollary}}
\setcounter{theorem}{0}
\renewcommand{\thetheorem}{A.\arabic{theorem}}
\setcounter{algorithm}{0}
\renewcommand{\thealgorithm}{A.\arabic{algorithm}}
\setcounter{example}{0}
\renewcommand{\theexample}{A.\arabic{example}}
\setcounter{remark}{0}
\renewcommand{\theremark}{A.\arabic{remark}}

\setcounter{equation}{0}
\renewcommand{\theequation}{A.\arabic{equation}}

\setcounter{figure}{0}
\renewcommand{\thefigure}{A.\arabic{figure}}

\setcounter{table}{0}
\renewcommand{\thetable}{A.\arabic{table}}

\section{Proofs}
\label{app:proofs}

\subsection{Proof of Consistency of $\hbeta_{N}$}

\begin{lemma}
  \label{lem:full-consistency}
  Under Assumptions~\ref{asmp:regularity-1}-\ref{asmp:regularity-3},  $\hbeta_{N} \xrightarrow{p} \bbeta_{0}$.
\end{lemma}
\begin{proof}
Corresponding to the partition of $\X=(\U\tp,\Z\tp)\tp$, $\bbeta$ can
be written as $(\bbeta_{u}\tp, \bbeta_{r}\tp)\tp$, and similarly for
$\bbeta_{0}=(\bbeta_{0,u}\tp, \bbeta_{0,r}\tp)\tp$ and
$\hbeta_{N}=(\hbeta_{N,u}\tp, \hbeta_{N,r}\tp)\tp$.  Throughout the
proofs, we sometimes suppress the subscript $N$ when no confusion
arises for notational simplicity, such as when writing the
distribution of $(Y,\X)$ and some expectation operators. Let
$\Psi_{0,u}(\bbeta)=\Exp\{(Y - \mu(\X\tp \bbeta))\U\}$ and
$\Psi_{0,r,j}(\bbeta)=\Exp\{(Y - \mu(\X\tp \bbeta))Z_{\cdot j}\}$.
Similarly, denote
$\Psi_{N,u}(\bbeta) = N^{-1} \sumN\{(Y_{i} - \mu(\X_{i}\tp \bbeta))\U\}$
and
$\Psi_{N,r,j}(\bbeta) = N^{-1} \sumN\{(Y_{i} - \mu(\X_{i}\tp \bbeta))Z_{ij}\}$.
Define the set
$A=\{\bbeta\in\Theta:\|\bbeta-\bbeta_{0}\|\ge\epsilon\}$.  Define
event
$E=\left\{\sup_{\bbeta \in \Theta}\|D_{N}^{-2}\left\{\Psi_{N}(\bbeta) - \Psi_{0}(\bbeta)\right\}\|<c(\epsilon)/2\right\}$,
where $c(\epsilon)$ is defined in
Assumption~\ref{asmp:regularity-2}(iv). We show
\begin{align}
  \label{eq:full-consistency}
  \Pr(\hbeta_{N} \in A) \le \Pr(E^{c}) \to 0.
\end{align}
Notice
that
\begin{align*}
  \|D_{N}^{-2}\Psi_{N}(\bbeta)\|
  &= \|D_{N}^{-2}\left\{
    \Psi_{N}(\bbeta) - \Psi_{0}(\bbeta) + \Psi_{0}(\bbeta)\right\}\| 
   \ge \|D_{N}^{-2}\Psi_{0}(\bbeta)\| - \|D_{N}^{-2}\left\{\Psi_{N}(\bbeta) - \Psi_{0}(\bbeta)\right\}\|.
\end{align*}
By Assumption~\ref{asmp:regularity-2} (iv), we know that
$\|D_{N}^{-2}\Psi_{0}(\bbeta)\| \ge c(\epsilon)$ for any
$\bbeta \in A$. If event $E$ occurs, then we must have
$\|D_{N}^{-2}\Psi_{N}(\bbeta)\| \ge c(\epsilon) - c(\epsilon)/2 > 0$. However,
$\|D_{N}^{-2}\Psi_{N}(\hbeta_{N})\| = 0$, which means the event
$\{\hbeta_{N} \in A\}$ is contained in the event $E^{c}$. Thus,
$\Pr(\hbeta_{N} \in A) \le \Pr(E^{c})$.  Next, we show
$\Pr(E) \to 1$ as $N\to\infty$, that is, show
\begin{align}
  \label{eq:DN-2-Psi}
  \sup_{\bbeta \in \Theta}
\|D_{N}^{-2}\left\{\Psi_{N}(\bbeta) - \Psi_{0}(\bbeta) \right\}\|
\xrightarrow{p} 0.
\end{align}
It is sufficient to show that 
\begin{align}
  \label{eq:Psi-u}
  \sup_{\bbeta \in\Theta}
  \|\Psi_{N,u}(\bbeta) - \Psi_{0,u}(\bbeta)\|
  & \xrightarrow{p} 0 \text{ and }\\
  \sup_{\bbeta \in \Theta}
    \rho_{j,N}^{-1}|\Psi_{N,r,j}(\bbeta) - \Psi_{0,r,j}(\bbeta)|
   & \xrightarrow{p} 0, \quad j=1,\dots,d_{r}. \label{eq:Psi-r-j}
\end{align}

We first prove Equation~\eqref{eq:Psi-u}.  Since $\Theta$ is compact,
for any $\Delta>0$, there exists a finite set
$\{\bbeta_{1}, \dots, \bbeta_{K}\} \subset \Theta$ such that for every
$\bbeta \in \Theta$, there exists an index
$k(\bbeta) \in \{1,\dots,K\}$ satisfying
$\|\bbeta - \bbeta_{k(\bbeta)}\|<\Delta$. Notice that
\begin{align*}
  & \sup_{\bbeta \in \Theta} \|\Psi_{N,u}(\bbeta) - \Psi_{0,u}(\bbeta)\| \\
  =& \sup_{\bbeta \in \Theta} \|\Psi_{N,u}(\bbeta) - \Psi_{N,u}(\bbeta_{k(\bbeta)})
     + \Psi_{N,u}(\bbeta_{k(\bbeta)}) - \Psi_{0,u}(\bbeta_{k(\bbeta)})
     + \Psi_{0,u}(\bbeta_{k(\bbeta)}) - \Psi_{0,u}(\bbeta)\| \\
  \le& \sup_{\bbeta \in \Theta}\|\Psi_{N,u}(\bbeta) - \Psi_{N,u}(\bbeta_{k(\bbeta)})\|
    + \max_{1 \le k \le K}\|\Psi_{N,u}(\bbeta_{k}) - \Psi_{0,u}(\bbeta_{k})\|
    + \sup_{\bbeta \in \Theta}\|\Psi_{0,u}(\bbeta_{k(\bbeta)}) - \Psi_{0,u}(\bbeta)\| \\
  \equiv& \sup_{\bbeta \in \Theta}L_{1}
          + \max_{1 \le k \le K}L_{2}
          + \sup_{\bbeta \in \Theta}L_{3}.
\end{align*}
We prove that $\sup_{\bbeta \in \Theta}L_{1}=o_{p}(1)$,
$\max_{1 \le k \le K}L_{2}=o_{p}(1)$, and
$\sup_{\bbeta \in \Theta}L_{3}=o(1)$.

For $\max_{1 \le k \le K}L_{2}$, Assumption~\ref{asmp:regularity-2} (i) gives
$\sup_{N}\Exp\{(Y - \mu(\X\tp \bbeta_{k}))^{2}\|\U\|^{2}\} <\infty$ for
each fixed $\bbeta_{k}\in\{\bbeta_{1}, \dots, \bbeta_{K}\}$. Thus, we have
$\Psi_{N,u}(\bbeta_{k})-\Psi_{0,u}(\bbeta_{k})\xrightarrow{p}\0$ by
the $L^{2}$ weak law of large numbers (WLLN) for a triangular array. Therefore,
$\max_{1 \le k \le K}L_{2} = o_{p}(1)$ since $K$ is finite.

For $\sup_{\bbeta \in \Theta}L_{1}$, we have
\begin{align*}
  \sup_{\bbeta \in \Theta}L_{1}
  &= \sup_{\bbeta \in \Theta}\left\|\oneN\sumN
    \left[(Y_{i} - \mu(\X_{i}\tp \bbeta))\U_{i} - (Y_{i} - \mu(\X_{i}\tp \bbeta_{k(\bbeta)}))\U_{i}\right]\right\| \\
  &= \sup_{\substack{\bbeta\in\Theta\\ \|\bbeta-\bbeta_{k(\bbeta)}\|<\Delta}}
    \left\|\oneN\sumN
    \left[\{\mu(\X_{i}\tp \bbeta_{k(\bbeta)})-\mu(\X_{i}\tp \bbeta)\}\U_{i}\right]\right\| \\
  &\le \oneN\sumN
    \sup_{\substack{\bbeta\in\Theta\\ \|\bbeta-\bbeta_{k(\bbeta)}\|<\Delta}}
    \left\|\{\mu(\X_{i}\tp \bbeta_{k(\bbeta)})-\mu(\X_{i}\tp \bbeta)\}\U_{i}\right\| \\
  &\le \Delta \oneN\sumN \left\{  L(\X_{i})\|\U_{i}\|\right\}.
\end{align*}
Since $\sup_{N}\Exp\{L^2(\X)\|\U\|^2\}<\infty$ by
Assumption~\ref{asmp:regularity-2} (iii), we know that
$N^{-1}\sumN \left\{ L(\X_{i})\|\U_{i}\|\right\} - \Exp\{L(\X)\|\U\|\} \xrightarrow{p} 0$
by the triangular array WLLN, and hence 
$\sup_{\bbeta\in\Theta}L_1=O_p(\Delta)$.  For $L_{3}$, we have
\begin{align*}
  \sup_{\bbeta \in \Theta}L_{3}
  &= \sup_{\substack{\bbeta\in\Theta\\ \|\bbeta-\bbeta_{k(\bbeta)}\|<\Delta}}
    \left\|\Exp
    [\{\mu(\X\tp \bbeta_{k(\bbeta)})-\mu(\X\tp \bbeta)\}\U] \right\| 
  \le \Delta \Exp\left\{ L(\X)\|\U\| \right\},
\end{align*}
which is $O(\Delta)$ since $\sup_N \Exp\{L(\X)\|\U\|\}<\infty$.  Since
$\Delta$ can be arbitrarily small, letting $\Delta \to 0$ gives
$\sup_{\bbeta\in\Theta}L_1=o_{p}(1)$ and
$\sup_{\bbeta\in\Theta}L_3=o(1)$.  Thus, we proved
$\sup_{\bbeta \in \Theta}\|\Psi_{N,u}(\bbeta) - \Psi_{0,u}(\bbeta)\| =o_{p}(1) + o(1)$
so that Equation~\eqref{eq:Psi-u} is proved.

To show Equation~\eqref{eq:Psi-r-j}, i.e.,
$\sup_{\bbeta \in \Theta}|\Psi_{N,r,j}(\bbeta) - \Psi_{0,r,j}(\bbeta)| = o_{p}(\rho_{j,N})$
for $j=1,\dots,d_{r}$, write
\begin{align*}
  & \sup_{\bbeta \in \Theta} |\Psi_{N,r,j}(\bbeta) - \Psi_{0,r,j}(\bbeta)| \\
 \le& \sup_{\bbeta \in \Theta}|\Psi_{N,r,j}(\bbeta) - \Psi_{N,r,j}(\bbeta_{k(\bbeta)})|
    + \max_{1 \le k \le K}|\Psi_{N,r,j}(\bbeta_{k}) - \Psi_{0,r,j}(\bbeta_{k})|
    + \sup_{\bbeta \in \Theta}|\Psi_{0,r,j}(\bbeta_{k(\bbeta)}) - \Psi_{0,r,j}(\bbeta)| \\
  \equiv& \sup_{\bbeta \in \Theta}R_{1}
          + \max_{1 \le k \le K}R_{2}
          + \sup_{\bbeta \in \Theta}R_{3}.
\end{align*}
We now prove that $\sup_{\bbeta \in \Theta}R_{1}=o_{p}(\rho_{j,N})$,
$\max_{1 \le k \le K}R_{2}=o_{p}(\rho_{j,N})$, and
$\sup_{\bbeta \in \Theta}R_{3}=o(\rho_{j,N})$. For
$\max_{1 \le k \le K}R_{2}$, we show that for each 
$\bbeta_{k} \in \{\bbeta_{1},\dots,\bbeta_{K}\}$ and any $\eta>0$,
\begin{align*}
  & \Pr\left(|\Psi_{N,r,j}(\bbeta_{k}) - \Psi_{0,r,j}(\bbeta_{k})| > \rho_{j,N}\eta\right) \\
  =& \Pr\left(
     \left|\oneN \sumN
     \left[
     (Y_{i}-\mu(\X_{i}\tp \bbeta_{k}))Z_{ij}
     - \Exp\{(Y-\mu(\X\tp \bbeta_{k}))Z_{\cdot j}\}
     \right] \right| > \rho_{j,N}\eta \right) \\
  \le& \frac{\Var\{(Y-\mu(\X\tp\bbeta_{k}))Z_{\cdot j}\}}
       {N\rho_{j,N}^{2}\eta^{2}} \\
  \le& \frac{\Exp\{(Y-\mu(\X\tp\bbeta_{k}))^{2}Z_{\cdot j}\}}
       {N\rho_{j,N}^{2}\eta^{2}}\\
  =& \frac{\Exp\{(Y-\mu(\X\tp\bbeta_{k}))^{2}\mid Z_{\cdot j}=1\}}
     {N\rho_{j,N}\eta^{2}} \to 0,
\end{align*}
as $N\rho_{j,N}\to\infty$, where the following bound holds
\begin{align}
  \Exp\{(Y-\mu(\X\tp\bbeta_{k}))^{2}\mid Z_{\cdot j}=1\}
  &= \rho_{j,N}^{-1}\Exp\{(Y-\mu(\X\tp\bbeta_{k}))^{2}Z_{\cdot j}\} \nonumber\\
  &\le  \sup_{N}\Exp\{(Y-\mu(\X\tp\bbeta_{k}))^{2}\|D_{N}^{-1}\X\|^{2}\}<\infty.
    \label{eq:rho-Z-bound}
\end{align}
Thus, we have $\max_{1 \le k \le K}R_{2} = o_{p}(\rho_{j,N})$ since $K$
is finite.  For $\sup_{\bbeta \in \Theta}R_{1}$,
\begin{align*}
\sup_{\bbeta\in\Theta}R_1
  &=
    \sup_{\bbeta\in\Theta}|\Psi_{N,r,j}(\bbeta)-\Psi_{N,r,j}(\bbeta_{k(\bbeta)})|\\
  &\le
    \sup_{\bbeta\in\Theta} \oneN \sumN \{
|\mu(\X_i^\top\bbeta)-\mu(\X_i^\top\bbeta_{k(\bbeta)})|\,Z_{ij}\}\\
  &\le
    \oneN \sumN \sup_{\bbeta\in\Theta}
  \{|\mu(\X_i^\top\bbeta)-\mu(\X_i^\top\bbeta_{k(\bbeta)})|\,Z_{ij}\} \\
  &\le
    \Delta \oneN \sumN \{L(\X_{i}) Z_{ij}\}.
\end{align*}
By Cauchy--Schwarz inequality,
\begin{align}
\rho_{j,N}^{-1}\Exp\{L(\X) Z_{\cdot j}\}
\le
\rho_{j,N}^{-1}
\left[\Exp\{L^{2}(\X) Z_{\cdot j}\}\right]^{1/2}
\{\Exp(Z_{\cdot j})\}^{1/2} 
=
\left[
\rho_{j,N}^{-1}\Exp\{L^{2}(\X) Z_{\cdot j}\}
\right]^{1/2}
= O(1). \label{eq:rho-L-Z-bound}
\end{align}
By Chebyshev's inequality, 
we have that for any $\eta>0$, as $N\rho_{j,N} \to \infty$,
\begin{align*}
  \Pr \left(\left|\frac1N\sum_{i=1}^N L(\X_i)Z_{ij}-\Exp\{L(\X) Z_{\cdot j}\}\right|>\eta\rho_{j,N}\right) 
  \le \frac{\Exp\{L^{2}(\X) Z_{\cdot j}\}}{\eta^2 N\rho_{j,N}^{2}} 
    \to 0,
\end{align*}
Hence
$N^{-1}\sum_{i=1}^N L(\X_i)Z_{ij} = O(\rho_{j,N}) + o_p(\rho_{j,N}) = O_p(\rho_{j,N})$.
Therefore, 
$\sup_{\bbeta\in\Theta}R_1= o_p( \rho_{j,N})$ since 
$\Delta$ can be arbitrarily small. Similarly, for $\sup_{\bbeta \in \Theta}R_{3}$,
\begin{align*}
\sup_{\bbeta\in\Theta}R_3
&=\sup_{\bbeta\in\Theta}|\Psi_{0,r,j}(\bbeta_{k(\bbeta)})-\Psi_{0,r,j}(\bbeta)| \\
&\le \sup_{\bbeta\in\Theta}
  \Exp\left[Z_{\cdot j}|\mu(\X^\top\bbeta)-\mu(\X^\top\bbeta_{k(\bbeta)})|\right] \\
&\le \Delta\Exp\!\left[L(\X)Z_{\cdot j}\right] =o(\rho_{j,N}),
\end{align*}
Thus, we proved Equation~\eqref{eq:Psi-r-j} and thus
Equation~\eqref{eq:full-consistency}. This completes the proof.
\end{proof}

\subsection{Proof of Theorem~\ref{thm:full}}

\label{supp:thm-full}
\begin{proof}
Using a Taylor expansion of $\Psi_{N}(\hbeta_{N})$, 
\begin{align*}
  \Psi_{N}(\hbeta_{N})
  &= \Psi_{N}(\bbeta_{0})
    - \left[ \int_{0}^{1}
    \left\{ \oneN \sumN 
    \dot{\mu}\left(\X_{i}\tp\{\bbeta_{0}+\lambda(\hbeta_{N}-\bbeta_{0})\}\right) \X_{i}\X_{i}\tp
    \right\} \diff \lambda \right]
    (\hbeta_{N}-\bbeta_{0})  = \0.
\end{align*}
After scaling, we obtain
\begin{align*}
   \sqrt{N}D_{N}(\hbeta_{N}-\bbeta_{0})
  = H_{N}^{-1}  \onesqN \sumN D_{N}^{-1}  \psi(Y_{i},\X_{i},\bbeta_{0}),
\end{align*}
where
\begin{align*}
  & H_{N}
    = \int_{0}^{1}
    \left\{ \frac{1}{N}\sumN 
    \dot{\mu}\left(\X_{i}\tp\{\bbeta_{0} + \lambda(\hbeta_{N} - \bbeta_{0})\}\right)
    D_{N}^{-1} \X_{i}\X_{i}\tp D_{N}^{-1}
    \right\} \diff \lambda.
\end{align*}
We will prove that 
\begin{align}
  H_{N}
  & \xrightarrow{p} M_{r}(\bbeta_{0}), \label{eq:HN-Mr}\\
  \onesqN \sumN \{D_{N}^{-1} \psi(Y_{i},\X_{i},\bbeta_{0})\}
  & \xrightarrow{d}
    \Nor(\0, \Lambda_{r}(\bbeta_{0})). \label{eq:sumN-xi}
\end{align}
Then, we will have 
\begin{align*}
\sqrt{N}D_{N}(\hbeta_{N}-\bbeta_{0})
\xrightarrow{d}
\Nor(\0,
  \left\{M_{r}^{-1}(\bbeta_{0})\right\}
  \Lambda_{r}(\bbeta_{0})
  \left\{ M_{r}^{-1}(\bbeta_{0}) \right\} )
\end{align*}
by Slutsky's theorem and the continuous mapping theorem. For
\ref{eq:HN-Mr}, note that by Assumption~\ref{asmp:regularity-1}
and~\ref{asmp:regularity-2},
\begin{align*}
  & \|H_{N} - M_{r}(\bbeta_{0})\| 
   = \left\|H_{N}
    - M_{r,N}(\bbeta_{0}) + M_{r,N}(\bbeta_{0})
    - M_{r}(\bbeta_{0}) \right\| \\
 \le& \int_{0}^{1} \left[ \left\| 
    \oneN \sumN \left\{
      \dot{\mu}\left(\X_{i}\tp(\bbeta_{0}
      + \lambda(\hbeta_{N} - \bbeta_{0})\right)
    -\dot{\mu}(\X_{i}\tp\bbeta_{0})
    \right\}
    D_{N}^{-1} \X_{i}\X_{i}\tp D_{N}^{-1}
     \right\|  \right]
    \diff \lambda  \\
   +& \left\| M_{r,N}(\bbeta_{0}) - M_{r}(\bbeta_{0}) \right\| \\
 \le& \int_{0}^{1}\left[  
    \oneN \sumN \left\| \left\{
      \dot{\mu}\left(\X_{i}\tp(\bbeta_{0}
      + \lambda(\hbeta_{N} - \bbeta_{0})\right)
    -\dot{\mu}(\X_{i}\tp\bbeta_{0})
    \right\}
    D_{N}^{-1} \X_{i}\X_{i}\tp D_{N}^{-1}\right\| \right] 
    \diff \lambda
    + o_{p}(1) \\
 \le& \|\hbeta_{N} - \bbeta_{0}\|
        \left\{ \oneN \sumN L(\X_{i})
    \left\|D_{N}^{-1} \X_{i} \right\|^{2} \right\}
    \int_{0}^{1} \lambda \diff \lambda
    + o_{p}(1),
\end{align*}
where $\|M_{r,N}(\bbeta_{0}) - M_{r}(\bbeta_{0}) \|=o_{p}(1)$ follows from
Assumption~\ref{asmp:regularity-3}.  By
$\sup_{N}\Exp\{L^{2}(\X)\|D_{N}^{-1}\X\|^{2}\} < \infty$ in
Assumption~\ref{asmp:regularity-2} (iii) and
$\sup_{N}\Exp(\|D_{N}^{-1}\X\|^{2}) < \infty$, we know
$\sup_{N}\Exp\{L(\X)\|D_{N}^{-1}\X\|^{2}\} < \infty$ by
Cauchy--Schwarz inequality. Thus,
$N^{-1} \sumN L(\X_{i}) \left\|D_{N}^{-1} \X_{i} \right\|^{2}= O_{p}(1)$
by Markov's inequality.  We also have
$\|\hbeta_{N} - \bbeta_{0}\|=o_{p}(1)$ from
Lemma~\ref{lem:full-consistency}, thus
$\|H_{N} - M_{r}(\bbeta_{0})\|=o_{p}(1)$ and \eqref{eq:HN-Mr} is
proved.

To prove \eqref{eq:sumN-xi}, we use the Lindeberg--Feller central limit
theorem (LF-CLT) since the distribution of
$D_{N}^{-1} \psi(Y,\X,\bbeta)$ depends on $N$. We derive the asymptotic
expectation and covariance and then verify the Lindeberg
condition. Notice that
\begin{align*}
  \Exp \left\{D_{N}^{-1}\psi(Y,\X,\bbeta_{0}) \right\}
  =  D_{N}^{-1} \Psi_{0}(\bbeta_{0}) = \0.
\end{align*}
Thus, $\Exp (N^{-1/2}\sumN D_{N}^{-1}\psi(Y_{i},\X_{i},\bbeta_{0})) = \0$.
By Assumption~\ref{asmp:regularity-3}, the covariance is
\begin{align*}
  \Var \left[\onesqN  \sumN \left\{Y_{i} - \mu(\X_{i}\tp\bbeta_{0})\right\}
    D_{N}^{-1}\X_{i} \right] 
  & = \Exp \left[
    \oneN \sumN \left\{Y_{i} - \mu(\X_{i}\tp\bbeta_{0})\right\}^{2}
    D_{N}^{-1} \X_{i}\X_{i}\tp D_{N}^{-1} \right] \\
  & = \Exp \{\Lambda_{r,N}(\bbeta_{0})\} \to \Lambda_{r}(\bbeta_{0}).
\end{align*}
Therefore, the asymptotic covariance is $\Lambda_{r}(\bbeta_{0})$. Finally,
we check the Lindeberg condition: for every $\epsilon > 0$,
\begin{equation} \label{LF-condition-full}
  N \Exp \left\{
    \left\|\onesqN D_{N}^{-1}\psi(Y,\X,\bbeta_{0}) \right\|^{2}
    \bI \left( \left\| \onesqN  D_{N}^{-1}\psi(Y,\X,\bbeta_{0}) \right\| > \epsilon \right)
  \right\} \to 0.
\end{equation}
By Assumption~\ref{asmp:regularity-2} (ii), for any $\epsilon>0$,
\begin{align*} 
  & N \Exp \left\{
    \left\|\onesqN D_{N}^{-1} \psi(Y,\X,\bbeta_{0}) \right\|^{2}
    \bI \left( \left\| \onesqN  D_{N}^{-1}\psi(Y,\X,\bbeta_{0}) \right\| > \epsilon \right)
    \right\} \\
  &= \Exp \left\{
    \left\| D_{N}^{-1}\psi(Y,\X,\bbeta_{0}) \right\|^{2}
    \bI \left( \left\|  D_{N}^{-1}\psi(Y,\X,\bbeta_{0}) \right\| > \epsilon \sqrt{N}\right)
    \right\} \\
  &\le \frac{\Exp\|D_{N}^{-1}\psi(Y,\X,\bbeta_{0})\|^{2+\kappa}}{(\epsilon\sqrt{N})^{\kappa}}
\to 0.
\end{align*}
Hence \eqref{LF-condition-full} holds. This proves \eqref{eq:sumN-xi}
and completes the proof of Theorem~\ref{thm:full}.
\end{proof}

\subsection{Proof of Proposition~\ref{prop:lower-bound}.}
\label{supp:lem-prob-bound}

\begin{proof}

Denote events $\cA_{j} = \{ \sumN\delta_{i}Z_{ij}=0 \}$ for $j=1,\dots,d_{r}$.
Since $\cA_j\subset \cA$, we have
$\Pr(\cA)\ge \Pr(\cA_j)$ for every $j=1,\dots,d_{r}$.  Conditional on $\Fn$, the inclusion
indicators $\{\delta_{i}\}_{i=1}^{N}$ are independent, and therefore
\begin{align*}
\Pr(\cA_j\mid \Fn)
&=
\Pr\left(\sum_{i=1}^N\delta_i Z_{ij}=0\middle| \Fn\right) \\
&=
\prod_{i=1}^N \{1-\pi(Y_i,\X_i)\}^{Z_{ij}}
=
\exp\left[
\sum_{i=1}^N Z_{ij}\log\{1-\pi(Y_i,\X_i)\}
\right].
\end{align*}
Taking expectations and applying Jensen's inequality gives 
\begin{align*}
\Pr(\cA_j)
&=
\Exp\left(
\exp\left[
\sum_{i=1}^N Z_{ij}\log\{1-\pi(Y_i,\X_i)\}
\right]
\right)\\
&\ge
\exp\left(
\Exp\left[
\sum_{i=1}^N Z_{ij}\log\{1-\pi(Y_i,\X_i)\}
\right]
  \right) \\
&\ge
\exp\left(
- N \Exp\left[
\frac{Z_{\cdot j}\pi(Y,\X)}{1-\pi(Y,\X)}
\right]
  \right) \\
&\ge
\exp\left(
-\frac{N}{1-\zeta}
\Exp\{Z_{\cdot j}\pi(Y,\X)\}
\right)\\
&=
\exp\left(
-\frac{N\rho_{j,N}}{1-\zeta}
\Exp\{\pi(Y,\X)\mid Z_{\cdot j}=1\}
\right)\\
&=
\exp\left(
-\frac{N\rho_{j,N}\omega_{j,N}}{1-\zeta}
\right). 
\end{align*}
Since this bound holds for every $j\in\{1,\ldots,d_r\}$, we obtain
\begin{align*}
  \Pr(\cA)
\ge
\max_{1\le j\le d_r}
\exp\left(
-\frac{N\rho_{j,N}\omega_{j,N}}{1-\zeta}
\right).
\end{align*}
\end{proof}

\subsection{Proof of Consistency of $\cbeta_{n}$}

\begin{lemma}
  \label{lem:sub-consistency}
Under Assumptions~\ref{asmp:regularity-1}--\ref{asmp:samp-prob-2}, if
$Nh_N\rho_{j,N}\to\infty$ for $j=1,\dots,d_{r}$, then
$\cbeta_n\xrightarrow{p}\bbeta_0$.
\end{lemma}

\begin{proof}
  Similar to the notations defined in the proof of
  Lemma~\ref{lem:full-consistency}, denote
  \begin{align*}
    & \Psi_{n,u}(\bbeta) = \oneN \sumN \left\{  \frac{\delta_{i}}{\pi(Y_{i},\X_{i})}(Y_{i} - \mu(\X_{i}\tp \bbeta))\U_{i}\right\}, \\
    & \Psi_{n,r,j}(\bbeta) = \oneN \sumN \left\{  \frac{\delta_{i}}{\pi(Y_{i},\X_{i})} (Y_{i} - \mu(\X_{i}\tp \bbeta))Z_{ij}\right\}, \quad j=1,\dots,d_{r}.
  \end{align*}
Following the argument used to derive \eqref{eq:DN-2-Psi}, it remains
to show that
$\sup_{\bbeta \in\Theta} \|D_{N}^{-2}\left\{\Psi_{n}(\bbeta) - \Psi_{0}(\bbeta) \right\}\|\xrightarrow{p} 0$.
Notice that
\begin{align*}
  \|D_{N}^{-2}\left\{\Psi_{n}(\bbeta) - \Psi_{0}(\bbeta) \right\}\|
  \le 
  \|D_{N}^{-2}\left\{\Psi_{n}(\bbeta) - \Psi_{N}(\bbeta) \right\}\|
  +  
  \|D_{N}^{-2}\left\{\Psi_{N}(\bbeta) - \Psi_{0}(\bbeta) \right\}\|,
\end{align*}
and we already know that
$\sup_{\bbeta \in\Theta} \|D_{N}^{-2}\left\{\Psi_{N}(\bbeta) - \Psi_{0}(\bbeta) \right\}\| \xrightarrow{p} 0$
from Lemma~\ref{lem:full-consistency}, what we need to prove is
$\sup_{\bbeta \in\Theta} \|D_{N}^{-2}\left\{\Psi_{n}(\bbeta) - \Psi_{N}(\bbeta) \right\}\| \xrightarrow{p} 0$.
It is sufficient to prove that
\begin{align}
\label{eq:Psi-u-sub}
  \sup_{\bbeta \in\Theta}
    \|\Psi_{n,u}(\bbeta) - \Psi_{N,u}(\bbeta)\|
  & \xrightarrow{p} 0,\\
  \sup_{\bbeta \in \Theta}
    \rho_{j,N}^{-1}|\Psi_{n,r,j}(\bbeta) - \Psi_{N,r,j}(\bbeta)|
  & \xrightarrow{p} 0, \quad j=1,\dots,d_{r}. \label{eq:Psi-r-j-sub}
\end{align}
For \eqref{eq:Psi-u-sub}, fix $\Delta>0$, then there exists a finite set
$\{\bbeta_{1}, \dots, \bbeta_{K}\} \subset \Theta$ such that for every
$\bbeta \in \Theta$, there exists an index $k(\bbeta) \in \{1,\dots,K\}$ satisfying
$\|\bbeta - \bbeta_{k(\bbeta)}\|<\Delta$. Notice that 
\begin{align*}
  & \sup_{\bbeta \in \Theta} \|\Psi_{n,u}(\bbeta) - \Psi_{N,u}(\bbeta)\| \\
  =& \sup_{\bbeta \in \Theta} \|\Psi_{n,u}(\bbeta) - \Psi_{n,u}(\bbeta_{k(\bbeta)})
     + \Psi_{n,u}(\bbeta_{k(\bbeta)}) - \Psi_{N,u}(\bbeta_{k(\bbeta)})
     + \Psi_{N,u}(\bbeta_{k(\bbeta)}) - \Psi_{N,u}(\bbeta)\| \\
  \le& \sup_{\bbeta \in \Theta}\|\Psi_{n,u}(\bbeta) - \Psi_{n,u}(\bbeta_{k(\bbeta)})\|
    + \max_{1 \le k \le K}\|\Psi_{n,u}(\bbeta_{k}) - \Psi_{N,u}(\bbeta_{k})\|
    + \sup_{\bbeta \in \Theta}\|\Psi_{N,u}(\bbeta_{k(\bbeta)}) - \Psi_{N,u}(\bbeta)\| \\
  \equiv& \sup_{\bbeta \in \Theta}L_{\pi,1}
          + \max_{1 \le k \le K}L_{\pi,2}
          + \sup_{\bbeta \in \Theta}L_{1},
\end{align*}
where $L_{1}$ is defined in the proof of
Lemma~\ref{lem:full-consistency} and
$\sup_{\bbeta \in \Theta}L_{1}=o_{p}(1)$ was established there. It
remains to show $\max_{1 \le k \le K}L_{\pi,2} = o_{p}(1)$ and
$\sup_{\bbeta \in \Theta}L_{\pi,1}= o_{p}(1)$. For fixed $k$,
\begin{align*}
  \Psi_{n,u}(\bbeta_k)-\Psi_{N,u}(\bbeta_k)
  =\oneN \sumN \left[  \left(\frac{\delta_i}{\pi(Y_{i},\X_{i})}-1\right)
\{Y_i-\mu(\X_i^\top\bbeta_k)\} \U_i\right].
\end{align*}
Its expectation is $\0$ because
$\Exp\{\delta \pi^{-1}(Y,\X) -1 \mid Y,X\} = 0 $ and the covariance
is
\begin{align*}
  \Var\{\Psi_{n,u}(\bbeta_k)-\Psi_{N,u}(\bbeta_k)\}
  = \oneN\Exp\left[\frac{1-\pi(Y,\X)}{\pi(Y,\X)} \{Y-\mu(\X\tp\bbeta_k)\}^{2} \U\U\tp\right].
\end{align*}
Notice that
$\{1-\pi(Y,\X)\} / \pi(Y,\X) < \pi^{-1}(Y,\X) \le (\alpha h_{N})^{-1}$
by Assumption~\ref{asmp:samp-prob-1} (i). Therefore, we have
\begin{align*}
  \|\Var\{\Psi_{n,u}(\bbeta_k)-\Psi_{N,u}(\bbeta_k)\}\|
  \le \frac{1}{\alpha N h_{N}} \Exp[\{Y-\mu(\X\tp\bbeta_k)\}^{2} \|\U\|^{2}]
  \to 0
\end{align*}
as $Nh_{N}\to \infty$, where
$\sup_{N}\Exp[\{Y-\mu(\X\tp\bbeta_k)\}^{2} \|\U\|^{2}] < \infty$ is by
Assumption~\ref{asmp:regularity-2} (i).  Thus,
$\Psi_{n,u}(\bbeta_k)-\Psi_{N,u}(\bbeta_k)=o_{p}(1)$ and thus
$\max_{1 \le k \le K}L_{\pi,2} = o_{p}(1)$. For
$\sup_{\bbeta \in \Theta}L_{\pi,1}$, we have
\begin{align*}
  \sup_{\bbeta \in \Theta}L_{\pi,1}
  &= \sup_{\bbeta \in \Theta}\left\|
    \oneN\sumN
    \left[
    \frac{\delta_{i}}{\pi(Y_{i},\X_{i})}
    \left\{\psi_{u}(Y_{i},\X_{i},\bbeta) - \psi_{u}(Y_{i},\X_{i},\bbeta_{k(\bbeta)})
    \right\} \right] \right\| \\
  &\le \oneN\sumN
    \sup_{\substack{\bbeta\in\Theta\\ \|\bbeta-\bbeta_{k(\bbeta)}\|<\Delta}}
  \left\|\frac{\delta_{i}}{\pi(Y_{i},\X_{i})}
  \left\{\psi_{u}(Y_{i},\X_{i},\bbeta) - \psi_{u}(Y_{i},\X_{i},\bbeta_{k(\bbeta)})
  \right\} \right\| \\
  &\le \Delta \oneN\sumN
    \left\{\frac{\delta_{i}}{\pi(Y_{i},\X_{i})} L(\X_{i})\|\U_{i}\|\right\}.
\end{align*}
Under $Nh_{N}\to\infty$,
\begin{align*}
  & \Exp \left\{\oneN\sumN\frac{\delta_{i}}{\pi(Y_{i},\X_{i})} L(\X_{i})\|\U_{i}\|\right\}
    = \Exp\{L(\X)\|\U\|\} < \infty, \\
  & \Var \left\{\oneN\sumN\frac{\delta_{i}}{\pi(Y_{i},\X_{i})} L(\X_{i})\|\U_{i}\|\right\}
    \le \oneN \Exp\left\{\frac{1}{\pi(Y,\X)} L^{2}(\X)\|\U\|^{2} \right\}
    \le \frac{1}{\alpha N h_{N}} \Exp\{ L^{2}(\X)\|\U\|^{2}\} \to 0,
\end{align*}
where $\sup_N \Exp\{L^2(\X)\|\U\|^2\}<\infty$ is by
Assumption~\ref{asmp:regularity-2} (iii).
Therefore,
\begin{align*}
  \oneN \sumN
  \left\{
  \frac{\delta_{i}}{\pi(Y_{i},\X_{i})} L(\X_{i})\|\U_{i}\|
  \right\}
  =
  \Exp\{L(\X)\|\U\|\}+o_{p}(1)=O_{p}(1).
\end{align*}
As a result, for any fixed $\Delta>0$,
$\sup_{\bbeta\in\Theta}L_{\pi,1}=O_p(\Delta)$. Letting $\Delta\to 0$, we
have $\sup_{\bbeta\in\Theta}L_{\pi,1}=o_p(1)$. This proves
equation~\eqref{eq:Psi-u-sub}.

To prove Equation~\eqref{eq:Psi-r-j-sub}, i.e.,
$\sup_{\bbeta \in \Theta}|\Psi_{n,r,j}(\bbeta) - \Psi_{N,r,j}(\bbeta)| = o_{p}(\rho_{j,N})$,
note that
\begin{align*}
  & \sup_{\bbeta \in \Theta} |\Psi_{n,r,j}(\bbeta) - \Psi_{N,r,j}(\bbeta)| \\
 \le& \sup_{\bbeta \in \Theta}|\Psi_{n,r,j}(\bbeta) - \Psi_{n,r,j}(\bbeta_{k(\bbeta)})|
    + \max_{1 \le k \le K}|\Psi_{n,r,j}(\bbeta_{k}) - \Psi_{N,r,j}(\bbeta_{k})|
    + \sup_{\bbeta \in \Theta}|\Psi_{N,r,j}(\bbeta_{k(\bbeta)}) - \Psi_{N,r,j}(\bbeta)| \\
  \equiv& \sup_{\bbeta \in \Theta}R_{\pi,1}
          + \max_{1 \le k \le K}R_{\pi,2}
          + \sup_{\bbeta \in \Theta}R_{1}.
\end{align*}
We have shown $\sup_{\bbeta \in \Theta}R_{1}=o_{p}(\rho_{j,N})$ in the
proof of Lemma~\ref{lem:full-consistency}, so it remains to prove
$\sup_{\bbeta \in \Theta}R_{\pi,1}=o_{p}(\rho_{j,N})$ and
$\max_{1 \le k \le K}R_{\pi,2}=o_{p}(\rho_{j,N})$.  For
$\max_{1 \le k \le K}R_{\pi,2}$, we show that for each
$\bbeta_{k} \in \{\bbeta_{1},\dots,\bbeta_{K}\}$ and any $\eta>0$,
\begin{align*}
  & \Pr\left(|\Psi_{n,r,j}(\bbeta_{k}) - \Psi_{N,r,j}(\bbeta_{k})| > \rho_{j,N}\eta\right) \\
  =&
     \Pr\left(
     \left|\oneN \sumN
     \left[ \left\{ \frac{\delta_{i}}{\pi(Y_{i},\X_{i})} - 1 \right\}
     \{Y_{i}-\mu(\X_{i}\tp \bbeta_{k})\}Z_{ij}
     \right] \right| > \rho_{j,N}\eta
     \right) \\
  \le&
       \frac{\Var\left[
       \left\{ \frac{\delta_{i}}{\pi(Y_{i},\X_{i})} - 1 \right\}
       \{Y-\mu(\X\tp\bbeta_{k})\}Z_{\cdot j}\right]}
       {N\rho_{j,N}^{2}\eta^{2}} \\
  \le&
       \frac{\Exp[\{Y-\mu(\X\tp\bbeta_{k})\}^{2}Z_{\cdot j}]}
       {N\rho_{j,N}^{2}h_{N}\alpha\eta^{2}}\to 0,
\end{align*}
as $N\rho_{j,N}h_{N}\to\infty$, where we again use
inequality~\eqref{eq:rho-Z-bound}.  Thus, we have 
$\max_{1 \le k \le K}R_{\pi,2} = o_{p}(\rho_{j,N})$ since $K$ is
finite.  For $\sup_{\bbeta \in \Theta}R_{\pi,1}$,
\begin{align*}
\sup_{\bbeta\in\Theta}R_{\pi,1}
&=\sup_{\bbeta\in\Theta}|\Psi_{n,r,j}(\bbeta)-\Psi_{n,r,j}(\bbeta_{k(\bbeta)})|\\
&\le \sup_{\bbeta\in\Theta} \oneN \sumN \left\{
  \frac{\delta_{i}}{\pi(Y_{i},\X_{i})}
|\mu(\X_i^\top\bbeta)-\mu(\X_i^\top\bbeta_{k(\bbeta)})|\,Z_{ij}\right\} \\
&\le \oneN \sumN \sup_{\bbeta\in\Theta} \left\{ 
   \frac{\delta_{i}}{\pi(Y_{i},\X_{i})}
  |\mu(\X_i^\top\bbeta)-\mu(\X_i^\top\bbeta_{k(\bbeta)})|\,Z_{ij} \right\} \\
&\le \Delta \oneN \sumN
\left\{  \frac{\delta_{i}}{\pi(Y_{i},\X_{i})}
  L(\X_{i}) Z_{ij} \right\}.
\end{align*}
By inequality~\eqref{eq:rho-L-Z-bound} and
Assumption~\ref{asmp:regularity-2} (iii) we have that for any
$\eta>0$,
\begin{align*}
  & \Pr \left(\left|\oneN\sumN \left\{  \frac{\delta_{i}}{\pi(Y_{i},\X_{i})}
    L(\X_{i}) Z_{ij} \right\}
    -\rho_{j,N}\Exp\{L(\X)\mid Z_{\cdot j}=1\}\right|>\eta\rho_{j,N}\right) \\
  &\le \frac{\Exp\{L^{2}(\X) Z_{\cdot j}\}}{\alpha \eta^2 N\rho_{j,N}^{2}h_{N}} \\
  &\le \frac{\Exp\{L^{2}(\X)\mid Z_{\cdot j}=1\}}{\alpha \eta^2 N\rho_{j,N}h_{N}} 
    \to 0,
\end{align*}
where we use
$\sup_{N} [\rho_{j,N}^{-1}\Exp\{\delta \pi^{-1}(Y,\X)L^{2}(\X) Z_{\cdot j}\} ]
= \sup_{N}[\rho_{j,N}^{-1} \Exp\{L^{2}(\X)Z_{\cdot j}\}] < \infty$. Hence
$N^{-1}\sum_{i=1}^N \{\delta_{i} \pi^{-1}(Y_{i},\X_{i}) L(\X_i)Z_{ij}\}
=O(\rho_{j,N})+o_p(\rho_{j,N}) =O_p(\rho_{j,N})$.
Therefore, $\sup_{\bbeta\in\Theta}R_{\pi,1}= O_p(\Delta \rho_{j,N})$ and
$\sup_{\bbeta\in\Theta}R_{\pi,1}= o_p( \rho_{j,N})$ as $\Delta\to 0$. Thus,
we proved Equation~\eqref{eq:Psi-r-j-sub}. This completes the proof of
Lemma~\ref{lem:sub-consistency}.
\end{proof}

\subsection{Proof of Theorem~\ref{thm:sub1}}
\label{supp:thm-sub1}
\begin{proof}
Recall that
$\Psi_{n}(\bbeta) = N^{-1} \sumN
\left\{ \delta_{i} \pi^{-1}(Y_{i}, \X_{i}) \psi(Y_{i},\X_{i},\bbeta) \right\}$.
Using a Taylor expansion for $\Psi_{n}(\cbeta_{n})$,
\begin{align*}
  \Psi_{n}(\cbeta_{n})
  &= \Psi_{n}(\bbeta_{0})
  - \left[\int_{0}^{1}\left\{ \oneN \sumN \frac{\delta_{i}}{\pi(Y_{i}, \X_{i})}
    \dot{\mu}\left(
    \X_{i}\tp
    \{\bbeta_{0}+\lambda(\cbeta_{n}-\bbeta_{0})\}
    \right)
    \X_{i}\X_{i}\tp \right\} \diff \lambda \right]
  (\cbeta_{n}-\bbeta_{0}) \\
  & = \0.
\end{align*}
After scaling, we obtain
\begin{align*}
   \sqrt{N}\Omega_{N}D_{N}(\cbeta_{n}-\bbeta_{0})
  = M_{\pi,N}^{-1}(\cbeta_{n})  \frac{1}{\sqrt{N}} \sumN \eta_{N,i},
\end{align*}
where
\begin{align*}
   M_{\pi,N}(\cbeta_{n})
  &  = \int_{0}^{1}\left\{
    \oneN\sumN
    \frac{\delta_{i}}{\pi(Y_{i}, \X_{i})}
    \dot{\mu}\left(\X_{i}\tp\{
    \bbeta_{0} + \lambda(\cbeta_{n} - \bbeta_{0})\}\right)
    K_{N}(\X_{i}) \right\} \diff \lambda,\\
  \eta_{N,i}
  & = \frac{\delta_{i}}{\pi(Y_{i},\X_{i})}
  \left\{Y_{i} - \mu(\X_{i}\tp\bbeta_{0})  \right\}
  \Omega_{N} D_{N}^{-1} \X_{i}, \quad i=1,\dots,N.
\end{align*}
We will prove that 
\begin{align}
  M_{\pi,N}(\cbeta_{n})
  & \xrightarrow{p} M_{\pi}(\bbeta_{0}), \label{eq:M-pi-N}\\
  \frac{1}{\sqrt{N}} \sumN \eta_{N,i}
  & \xrightarrow{d} \Nor(\0, \Lambda_{\pi}(\bbeta_{0})), \label{eq:sum-eta}
\end{align}
then we will have
\begin{align*}
\sqrt{N}\Omega_{N}D_{N}(\cbeta_{n}-\bbeta_{0})
\xrightarrow{d}
\Nor\left(\0,
  \left\{M_{\pi}^{-1}(\bbeta_{0})\right\}
  \Lambda_{\pi}(\bbeta_{0})
  \left\{ M_{\pi}^{-1}(\bbeta_{0}) \right\}\tp  \right)
\end{align*}
by Slutsky's theorem and the continuous mapping theorem. To prove
\eqref{eq:M-pi-N}, it suffices to show that the $(k,l)$-th entry of
$M_{\pi,N}(\cbeta_{n}) - M_{\pi}(\bbeta_{0})$ is $o_{p}(1)$ for
$k,l=1,\dots,d$. Notice that
\begin{align*}
  & |M_{\pi,N}(\cbeta_{n}) - M_{\pi}(\bbeta_{0})|_{k,l} \\
  \le& \int_{0}^{1}\left| \oneN \sumN
     \frac{\delta_{i}}{\pi(Y_{i}, \X_{i})}
    \left\{  
      \dot{\mu}\left(\X_{i}\tp(\bbeta_{0}+\lambda(\cbeta_{n}-\bbeta_{0})\right)
      - \dot{\mu}(\X_{i}\tp\bbeta_{0}) \right\}
    K_{N}(\X_{i}) \right|_{k,l}
    \diff \lambda \\
  &+ \left| 
     \left\{ \oneN\sumN
    \left(\frac{\delta_{i}}{\pi(Y_{i}, \X_{i})}-1\right)
\dot{\mu}\left(\X_{i}\tp\bbeta_{0}\right)
    K_{N}(\X_{i})\right\}
    \right|_{k,l}  \\
  &+  \left| 
     \left\{ \oneN\sumN \dot{\mu}\left(\X_{i}\tp\bbeta_{0}\right)
    K_{N}(\X_{i})\right\}
    - M_{\pi}(\bbeta_{0}) \right|_{k,l} \\
  \le& \|\cbeta_{n}-\bbeta_{0}\| \left\{ \oneN\sumN
    \left|\frac{\delta_{i}}{\pi(Y_{i}, \X_{i})} L(\X_{i})
      K_{N}(\X_{i})
       \right|_{k,l}  \right\}
       \int_{0}^{1} \lambda \diff \lambda \\
  &+  \left| \left\{ \oneN \sumN 
      \left(\frac{\delta_{i}}{\pi(Y_{i}, \X_{i})} - 1 \right)
      \dot{\mu}\left(\X_{i}\tp\bbeta_{0}\right)
    K_{N}(\X_{i}) \right\}\right|_{k,l}\\
  &+  \left| 
     \left\{ \oneN\sumN \dot{\mu}\left(\X_{i}\tp\bbeta_{0}\right)
    K_{N}(\X_{i})\right\}
    - M_{\pi}(\bbeta_{0}) \right|_{k,l} \\ 
  =& \|\cbeta_{n}-\bbeta_{0}\| h_{k,l} \int_{0}^{1} \lambda \diff \lambda
     + |g_{k,l}| + o_{p}(1),
\end{align*}
where we denote 
\begin{align*}
  h_{k,l}
  &=  \oneN \sumN \left| \frac{\delta_{i}}{\pi(Y_{i}, \X_{i})} L(\X_{i})
  K_{N}(\X_{i})\right|_{k,l}, \\
  g_{k,l}
  &= 
     \left\{ \oneN\sumN
       \left(\frac{\delta_{i}}{\pi(Y_{i}, \X_{i})} - 1 \right)
      \dot{\mu}\left(\X_{i}\tp\bbeta_{0}\right)
    K_{N}(\X_{i})
    \right\}_{k,l},
\end{align*} 
and the third term is $o_p(1)$ by Assumption~\ref{asmp:samp-prob-1} (ii).
We already know $\|\cbeta_{n}-\bbeta_{0}\|=o_{p}(1)$ by
Lemma~\ref{lem:sub-consistency}. Additionally, we have $h_{k,l} = O_{p}(1)$
 by Assumption~\ref{asmp:samp-prob-2} (i):
\begin{align*}
  \Exp (h_{k,l})
  = \Exp\{|L(\X) K_{N}(\X)|_{k,l}\} 
  \le \sup_{N} \Exp \left\| L(\X) K_{N}(\X) \right\|_{F} < \infty.
\end{align*}
Therefore, Markov's inequality yields $h_{k,l} = O_{p}(1)$ and thus
$\|\cbeta_{n}-\bbeta_{0}\| h_{k,l} \int_{0}^{1} \lambda \diff \lambda = o_{p}(1)$. Hence,
it remains to show $g_{k,l} = o_{p}(1)$ for $k,l=1,\dots,d$. Notice
that
\begin{align*}
  \Exp (g_{k,l}) = \Exp \left[\oneN\sumN
\left(\frac{\delta_{i}}{\pi(Y_{i}, \X_{i})} - 1\right)
      \dot{\mu}\left(\X_{i}\tp\bbeta_{0}\right)
  \{K_{N}(\X_{i})\}_{k,l}\right]
  = \0
\end{align*}
since $\Exp\{\delta / \pi(Y, \X) \mid Y,\X\} = 1$. By
Assumption~\ref{asmp:samp-prob-2} (iii), we know that
\begin{align*}
  \Var(g_{k,l})
  &= \Var \left[ \oneN\sumN
\left(\frac{\delta_{i}}{\pi(Y_{i}, \X_{i})} - 1\right)
      \dot{\mu}\left(\X_{i}\tp\bbeta_{0}\right)
  \left\{ K_{N}(\X_{i}) \right\}_{k,l}\right] \\
 & = \oneN \Exp \left[
  \frac{1-\pi(Y, \X)}{\pi(Y, \X)}
  \dot{\mu}^{2}\left(\X\tp\bbeta_{0}\right)
    \left\{K_{N}(\X) \right\}^{2}_{k,l}\right] \\
 &\le \frac{1}{\alpha N h_{N}}
    \Exp \left[
  \dot{\mu}^{2}\left(\X\tp\bbeta_{0}\right)
    \left\{K_{N}(\X) \right\}^{2}_{k,l}\right] \\
 &\le \frac{1}{\alpha N h_{N}}
    \Exp \left\|
  \dot{\mu}\left(\X\tp\bbeta_{0}\right)
    K_{N}(\X)
    \right\|^{2}_{F}     \to 0.
\end{align*}
By Chebyshev's inequality, we have $g_{k,l}=o_{p}(1)$ and thus
$|g_{k,l}|=o_{p}(1)$. Equation~\eqref{eq:M-pi-N} is proved.

Next, we prove \eqref{eq:sum-eta}.  We will use LF-CLT again. We
first derive the asymptotic mean and variance for
$N^{-1/2}\sumN \eta_{N,i}$, and then check the Lindeberg
condition. Denote
$\eta_{N} = \delta \pi^{-1}(Y,\X)\Omega_{N} D_{N}^{-1}\psi(Y,\X,\bbeta_{0})$,
then
\begin{align*}
  \Exp (\eta_{N})
  = \Exp\left\{\frac{\delta}{\pi(Y,\X)}
      \Omega_{N} D_{N}^{-1}\psi(Y,\X,\bbeta_{0})\right\} 
  = \Omega_{N}D_{N}^{-1}\Exp\left\{\psi(Y,\X,\bbeta_{0})\right\} = \0.
\end{align*}
Hence, $\Exp \{ N^{-1/2}\sumN \eta_{N,i} \} = 0$. For the
covariance, by Assumption~\ref{asmp:samp-prob-1} (ii), 
\begin{align*}
   \Var \left( \frac{1}{\sqrt{N}} \sumN \eta_{N,i}\right)
  = & \frac{1}{N} \sumN \Exp(\eta_{N,i}\eta_{N,i}\tp) \\
  = & \frac{1}{N} \sumN
      \Exp\left[
      \frac{\delta_{i}\{Y_{i} - \mu(\X_{i}\tp\bbeta_{0})\}^{2}}{\pi^{2}(Y_{i},\X_{i})}
      \Omega_{N}D_{N}^{-1} \X_{i} \X_{i}\tp D_{N}^{-1}\Omega_{N} \right] \\
  = & \Exp\left[\frac{1}{N} \sumN
      \frac{\{Y_{i} - \mu(\X_{i}\tp\bbeta_{0})\}^{2}}{\pi(Y_{i},\X_{i})}
      \Omega_{N}D_{N}^{-1} \X_{i} \X_{i}\tp D_{N}^{-1}\Omega_{N} \right] \\
  = & \Exp\left[
      \Omega_{N}D_{N}^{-1}\Lambda_{\pi,N}(\bbeta_{0})D_{N}^{-1}\Omega_{N}
      \right]
  \to  \Lambda_{\pi}(\bbeta_{0}).
\end{align*}
Next, we check the Lindeberg condition.
It remains to verify
that for every $\epsilon > 0$,
\begin{equation*}
  N \Exp \left\{
    \left\| \frac{1}{\sqrt{N}} \eta_{N} \right\|^{2}
    \bI \left( \left\| \frac{1}{\sqrt{N}}  \eta_{N} \right\| > \epsilon \right)
  \right\} \to 0.
\end{equation*}
By Assumption~\ref{asmp:samp-prob-2} (ii), for any $\epsilon>0$,
\begin{align*} 
   N \Exp \left\{
    \left\|\onesqN \eta_{N}\right\|^{2}
    \bI \left( \left\| \onesqN \eta_{N} \right\| > \epsilon \right)
    \right\} 
  = \Exp \left\{
    \left\| \eta_{N} \right\|^{2}
    \bI \left( \left\| \eta_{N}\right\| > \epsilon \sqrt{N}\right)
    \right\} 
  \le \frac{\Exp\left\| \eta_{N}\right\|^{2+\kappa_{2}}}{(\epsilon\sqrt{N})^{\kappa_{2}}}
\to 0.
\end{align*}
This completes the proof of Equation~\eqref{eq:sum-eta} and 
the proof of Theorem~\ref{thm:sub1}.
\end{proof}

\subsection{Proof of Corollary~\ref{cor:rc-variance-decomposition}}
\label{supp:cor-rc-variance-decomposition}

\begin{proof}
Under an RC design, $\pi(Y_i,\X_i)=1$ for $i\in\RN$ and 
$\X_i=J_u\U_i$ for $i\in\RN^c$. Therefore,
\begin{align*}
\Lambda_{\pi,N}(\bbeta_0)
={}&\oneN\sumN\frac{\{Y_i-\mu(\X_i\tp\bbeta_{0})\}^2}{\pi(Y_i,\X_i)}\X_i\X_i\tp\\
={}&\Lambda_N(\bbeta_0)
+\oneN\sum_{i\in\RN^c}
\frac{1-\pi(Y_i,\U_i)}{\pi(Y_i,\U_i)}
\{Y_i-\mu(\U_i\tp\bbeta_{0,u})\}^2
J_u\U_i\U_i\tp J_u\tp.
\end{align*}
Substitution into
$\Var_N\{\cbeta_n\}=N^{-1}M_N^{-1}(\bbeta_0)
\Lambda_{\pi,N}(\bbeta_0)M_N^{-1}(\bbeta_0)$ gives
Equation~\eqref{eq:rc-variance-decomposition}.
\end{proof}

\subsection{Proof of Proposition~\ref{prop:rc-optimal-scores}}
\label{supp:prop-rc-optimal-scores}

\begin{proof}
In this proof, we abbreviate $\pi(Y_{i},\X_{i})$ to $\pi_{i}$. For
$i\in\RN^c$, define
\begin{align*}
t_i=
|Y_i-\mu(\U_i\tp\bbeta_{0,u})|
\left\|L M_N^{-1}(\bbeta_0)J_u\U_i\right\|.
\end{align*}
By Equation~\eqref{eq:rc-variance-decomposition}, conditional on $\Fn$
and subject to $\sum_{i\in\RN^c}\pi_{i}=n-|\RN|$, and $0<\pi_{i}\le 1$
for $i\in\RN^c$, minimizing $\tr\{\Var_N(L\cbeta_n)\}$ with respect to
the probabilities $\{\pi_{i}\}_{i \in \RN^{c}}$ is equivalent to
minimizing $\sum_{i\in\RN^c}(\pi_{i}^{-1}t_{i}^{2})$. The Lagrangian
is
\begin{align*}
  \sum_{i\in\RN^c}\frac{t_i^{2}}{\pi_i}
  + \lambda \left\{\left(\sum_{i\in \RN^c}\pi_i\right) - (n-|\RN|) \right\}
+ \sum_{i\in\RN^c}\nu_i(\pi_i-1),
\end{align*}
where $\lambda\in\mathbb R$ and $\nu_i\ge0$ are the Lagrangian
multipliers. Since the objective function is strictly convex and the
constraints are convex, the Karush--Kuhn--Tucker (KKT) conditions are
sufficient for optimality. The KKT conditions are
\begin{align*}
  - \frac{t_{i}^{2}}{\pi_{i}^{2}} + \lambda + \nu_{i} & =0, \\
  \sum_{i\in\RN^c}\pi_{i} &= n-|\RN|, \\
  \nu_i(\pi_i-1) &= 0,
\end{align*}
together with $0<\pi_i\le1$ and $\nu_i\ge0$.  The case $n=N$ is
trivial; otherwise, at least one $\pi_k<1$, so $\nu_k=0$ and
$\lambda=t_k^2/\pi_k^2>0$.  For $i\in\RN^c$, if $0<\pi_{i}<1$, then
$\nu_{i}=0$ and $\pi_{i}=\lambda^{-1/2}t_{i}$, i.e., $\pi_{i}$ is
proportional to $t_{i}$. For brevity, define $\tau=\lambda^{-1/2}$. If
$\pi_{i}=1$, the KKT conditions give
$t_{i}^{2} = \lambda + \nu_{i} \ge \lambda$, which is equivalent to
$\tau t_{i} \ge 1$. Therefore, the probabilities
$\{\pi_{i}\}_{i \in \RN^{c}}$ that minimize $\tr\{\Var_N(L\cbeta_n)\}$
satisfy
\begin{align*}
  \pi_i=
\begin{cases}
\tau t_i, & \tau t_i<1,\\
1, & \tau t_i\ge 1,
\end{cases}
\end{align*}
or, equivalently, $\pi_{i}=\min\{\tau t_{i}, 1\}$, where $\tau>0$ is
chosen so that $\sum_{i\in\RN^c}\pi_{i}=n-|\RN|$. Since we also 
define $\pi_i=\min\{h_N\varphi(Y_i,\X_i),1\}$, an optimal weight
function $\varphi(Y_i,\X_i)$ can therefore be chosen as
$\varphi(Y_i,\X_i)=(\tau/h_N)t_i$, and hence
$\varphi(Y_i,\X_i)\propto t_i$.  Setting $L=\I$ gives the RC-Aopt
weight function. Setting $L=M_N(\bbeta_0)$ and using
$\|J_u\U_i\|=\|\U_i\|$ gives the RC-Lopt weight function.
\end{proof}

\subsection{Proof of Proposition~\ref{prop:original-optimal}}
\label{supp:prop-original-optimal}
\begin{proof}
In this proof, we abbreviate $\pi(Y_i,\X_i)$ to $\pi_i$ and define
$t_i = |Y_i-\mu(\X_i\tp\bbeta_0)| \left\|L M_N^{-1}(\bbeta_0)\X_i\right\|$,
$i=1,\ldots,N$. By the definitions of $\Var_N(\cbeta_n)$ and
$\Lambda_{\pi,N}(\bbeta_0)$,
\begin{align*}
\tr\{\Var_N(L\cbeta_n)\}
=
\frac{1}{N^2}\sumN
\frac{
\{Y_i-\mu(\X_i\tp\bbeta_0)\}^2
\|L M_N^{-1}(\bbeta_0)\X_i\|^2
}{\pi_i} 
=
\frac{1}{N^2}\sumN\frac{t_i^2}{\pi_i}.
\end{align*}
Conditional on $\Fn$ and subject to $\sumN\pi_{i}=n$, and
$0<\pi_{i}\le 1$ for $i=1,\dots,N$, minimizing
$\tr\{\Var_N(L\cbeta_n)\}$ with respect to the probabilities
$\{\pi_{i}\}_{i=1}^{N}$ is equivalent to minimizing
$\sumN (\pi_{i}^{-1}t_{i}^{2})$.  Therefore, by similar KKT
argument in the proof of
Proposition~\ref{prop:rc-optimal-scores} gives
$\pi_i=\min\{\tau t_i,1\}$, $i=1,\ldots,N$, where $\tau>0$ is chosen
so that $\sumN\pi_i=n$. Since $\pi_i=\min\{h_N\varphi(Y_i,\X_i),1\}$,
an optimal weight function can be chosen as
$\varphi(Y_i,\X_i)=(\tau/h_N)t_i$, and hence
$\varphi(Y_i,\X_i) \propto |Y_i-\mu(\X_i\tp\bbeta_0)| \left\|L M_N^{-1}(\bbeta_0)\X_i\right\|$.
Taking $L=\I$ gives the Aopt weight function. Taking $L=M_N(\bbeta_0)$
gives the Lopt weight function.
\end{proof}

\subsection{Proof of Proposition~\ref{prop:bl-order}}
\label{supp:prop-bl-order}

\begin{proof}
Let $m=|\mathcal B|$, which is a fixed constant. Conditional on the
full data $\Fn$, write the BL sampling probabilities as
$\pi_i^{\mathrm{BL}}=\min\{h_N\tau_1b(\X_i),1\}$, $i=1,\dots,N$, where
$\tau_1>0$ is chosen so that $\sumN\pi_i^{\mathrm{BL}}=n$. Because the
observations are independent and identically distributed and the BL
construction is symmetric in the observations, the pairs
$(\pi_i^{\mathrm{BL}},Z_{ij})$, $i=1,\dots,N$, are exchangeable. We
will prove
$\omega_{j,N}^{\mathrm{BL}}\gtrsim \min\{h_N\rho_{j,N}^{-1},1\} $ and
$\omega_{j,N}^{\mathrm{BL}} \le \min\{h_N\rho_{j,N}^{-1},1\}$.

For feature $k\in\mathcal B$, there are $N\bar X_{\cdot k}$
observations with $X_{ik}=1$ and $N(1-\bar X_{\cdot k})$ observations
with $X_{ik}=0$. Therefore, the average balance score satisfies
\begin{align*}
\bar b_N
&=
\oneN\sumN b(\X_i) \\
&=
\oneN\sumN\sum_{k\in\mathcal B}
\frac{|X_{ik}-\bar X_{\cdot k}|}{\bar X_{\cdot k}} \\
&=
\sum_{k\in\mathcal B}
\left\{
\bar X_{\cdot k}
\frac{1-\bar X_{\cdot k}}{\bar X_{\cdot k}}
+
(1-\bar X_{\cdot k})
\right\} \\
&=
\sum_{k\in\mathcal B}2(1-\bar X_{\cdot k})
\le2m.
\end{align*}
Moreover, $n
  = \sumN
\min\left\{h_N\tau_{1}b(\X_i),1\right\}
\le
\sumN h_N\tau_{1}b(\X_i)
= n \tau_{1}\bar b_N$, so that 
\begin{align}
\label{eq:supp-bl-tau-lower}
\tau_1
&\ge
\frac{1}{\bar b_N}
\ge
\frac{1}{2m}.
\end{align}
For a fixed rare feature $j\in\{1,\dots,d_r\}$, conditional on
$Z_{1j}=1$, we have
\begin{align*}
\frac{\Exp(\bar Z_{\cdot j}\mid Z_{1j}=1)}{\rho_{j,N}}
&=
1+\frac{1-\rho_{j,N}}{N\rho_{j,N}}
\to 1, \\
\Var\left(
\frac{\bar Z_{\cdot j}}{\rho_{j,N}}
\,\middle|\,
Z_{1j}=1
\right)
&=
\frac{(N-1)(1-\rho_{j,N})}
     {N^2\rho_{j,N}}
\to 0,
\end{align*}
as $N\rho_{j,N}\to\infty$. Hence,
$\bar Z_{\cdot j}/\rho_{j,N}\xrightarrow{p}1$ conditional on
$Z_{1j}=1$. Moreoever, define event
$\mathcal E_{j,N} = \{\bar Z_{\cdot j}\le2\rho_{j,N}\}$.  The above
conditional convergence gives
$\Pr(\mathcal E_{j,N}\mid Z_{1j}=1)\to1$. Because $\rho_{j,N}\to0$, on
$\mathcal E_{j,N}$ and for all sufficiently large $N$, we have
$1-\bar Z_{\cdot j}\ge1/2$.  Next, notice that when $Z_{1j}=1$, the
contribution of rare feature $j$ to $b(\X_1)$ is
$(1-\bar Z_{\cdot j})/\bar Z_{\cdot j}$ and contributions from all
other binary features are nonnegative, so
$b(\X_1)\ge(1-\bar Z_{\cdot j})/\bar Z_{\cdot j}$.  Using
\eqref{eq:supp-bl-tau-lower}, we obtain
\begin{align*}
\omega_{j,N}^{\mathrm{BL}}
&=
\Exp(\pi_1^{\mathrm{BL}}\mid Z_{1j}=1) \\
&=
\Exp\left[\min\{h_N\tau_1b(\X_1),1\}\middle|Z_{1j}=1\right] \\
&\ge
  \Exp\left[
  \min\left\{h_N\frac{1-\bar Z_{\cdot j}} {2m\bar Z_{\cdot j}},1 \right\}
\bI(\mathcal E_{j,N})
 \middle| Z_{1j}=1
\right] \\
&\ge
\min\left\{\frac{h_N}{8m\rho_{j,N}},1\right\}
  \Pr(\mathcal E_{j,N}\mid Z_{1j}=1)
  \gtrsim
\min\{h_N\rho_{j,N}^{-1},1\}.
\end{align*}

For the upper bound, $\pi_i^{\mathrm{BL}}\le1$ gives
$\omega_{j,N}^{\mathrm{BL}}\le1$. Moreover, 
\begin{align*}
\Exp(\pi_1^{\mathrm{BL}})
=
\oneN\sumN\Exp(\pi_i^{\mathrm{BL}})
=
\oneN\Exp\left(\sumN\pi_i^{\mathrm{BL}}\right)
=h_N.
\end{align*}
Therefore, $\rho_{j,N}\omega_{j,N}^{\mathrm{BL}} = \Exp(\pi_1^{\mathrm{BL}}Z_{1j}) \le \Exp(\pi_1^{\mathrm{BL}}) =
h_N $,
which implies
$\omega_{j,N}^{\mathrm{BL}}\le h_N\rho_{j,N}^{-1}$.
Combining the two upper bounds gives $\omega_{j,N}^{\mathrm{BL}}
\le
\min\{h_N\rho_{j,N}^{-1},1\}$.
This proves $\omega_{j,N}^{\mathrm{BL}}
\asymp
\min\{h_N\rho_{j,N}^{-1},1\}$ and $N\rho_{j,N}\omega_{j,N}^{\mathrm{BL}}
\asymp
\min\{n,N\rho_{j,N}\}$.
\end{proof}

\subsection{Proof of Proposition~\ref{prop:ra-lopt}}
\label{supp:prop-ra-lopt}

\begin{proof}
Conditional on $\Fn$, apply Proposition~\ref{prop:original-optimal}
with $L=D_N^{-2}M_N(\bbeta_0)$.  Then the RA-Lopt optimal weight
function satisfies
\begin{align*}
\varphi^{\text{RA-Lopt}}(Y_i,\X_i)
\propto
|Y_i-\mu(\X_i\tp\bbeta_0)|
\left\|D_N^{-2}\X_i\right\|,
\qquad i=1,\ldots,N.
\end{align*}
\end{proof}

\subsection{Proof of Theorem~\ref{thm:ag}}
\label{supp:thm-ag}

\begin{proof}

First, we define quantities used in the oracle version of the union
procedure.  Let
$\pi^{0}(Y_i,\X_i)=\min\{h_N\varphi^{0}(Y_i,\X_i),1\}$, $i=1,\dots,N$,
denote the oracle second-step sampling probabilities, where
$\varphi^{0}(Y_i,\X_i)$ is obtained by replacing the pilot quantities
in $\tilde\varphi(Y_i,\X_i)$ by their limits.  Define the oracle
second-step inclusion indicator
$\delta_i^{0}=\bI\{V_i<\pi^{0}(Y_i,\X_i)\}$, where $V_i$ is the same
uniform variable used in $\tdelta_i$. The oracle union inclusion
indicators and inclusion probabilities are
\begin{align*}
   \delta_{\ag,i}^{0}
    &= \bI(\delta_{\pl,i} = 1 \text{ or } \delta_{i}^{0} = 1)
    = \delta_{\pl,i} + (1-\delta_{\pl,i})\delta_{i}^{0},\\
   \pi_{\ag}^{0}(Y_{i},\X_{i})
    &=\pi_{\pl}(Y_i,\X_i) + \{1-\pi_{\pl}(Y_i,\X_i)\}\pi^0(Y_{i},\X_{i}),
      \quad i=1,\ldots,N.
\end{align*}
The corresponding oracle union estimating equation is
\begin{align*}
  \Psi^{0}_{\ag}(\bbeta) =
  \oneN 
  \sumN \left\{\frac{\delta^{0}_{\ag,i}}{\pi^{0}_{\ag}(Y_{i}, \X_{i})}
  \psi(Y_{i},\X_{i},\bbeta) \right\}.
\end{align*}
Let
$h_{\ag,N}^0=\Exp\{\pi_{\ag}^{0}(Y,\X)\}$.
Let $\Omega_{\ag,N}^{0} = \operatorname{diag}\{(h_{\ag,N}^0)^{1/2}\mathbf{1}\tp_{d_{u}}, \{(\bomega_{\ag}^{0})^{1/2}\}\tp\}$,
where
$\bomega_{\ag}^{0} = (\omega_{\ag,1}^{0}, \dots, \omega_{\ag,d_{r}}^{0})\tp$
and $\omega_{\ag,j}^0=\Exp\{\pi_{\ag}^{0}(Y,\X) \mid Z_{\cdot j}=1\}$.
Let $K_{\ag,N}^{0}(\X)
  = \Omega_{\ag,N}^{0}D_N^{-1}\X\X\tp D_N^{-1}(\Omega_{\ag,N}^{0})^{-1}$ and 
\begin{align*}
  G_N(Y,\X,\bbeta_0)
  =& |\dot\mu(\X^\top\bbeta_0)|\|K_{\ag,N}^0(\X)\|_F
+\|D_N^{-2}\psi(Y,\X,\bbeta_0)\|, \\
  T_N(Y,\X,\bbeta_0)
  =& L(\X)\|K_{\ag,N}^0(\X)\|_F + \|D_N^{-1}\psi(Y,\X,\bbeta_0)\|^2 .
\end{align*}
Thus,  Assumption~\ref{asmp:ag} (ii) implies:
\begin{align}
  & \sup_N h_{\pl,N}^{-1}
    \Exp\{\pi_{\pl}(Y,\X)L_{\pl}(Y,\X)|\dot\mu(\X^\top\bbeta_0)|\|K_{\ag,N}^0(\X)\|_F\}
    <\infty, \label{eq:ag-GN-ineq-2}\\
  & \sup_N h_{\pl,N}^{-1}
    \Exp\{\pi_{\pl}(Y,\X)L_{\pl}(Y,\X)\|D_N^{-2}\psi(Y,\X,\bbeta_0)\|\}
    <\infty, \label{eq:ag-GN-ineq-3}\\
  & \sup_N \Exp\{L_{\pl}(Y,\X)L(\X)\|K_{\ag,N}^0(\X)\|_F\}<\infty,
  \label{eq:ag-TN-ineq-1}\\
  & \sup_N \Exp\{L_{\pl}(Y,\X)\|D_N^{-1}\psi(Y,\X,\bbeta_0)\|^2\}<\infty.
  \label{eq:ag-TN-ineq-2}
\end{align}
Recall the feasible union estimator $\cbeta_{\ag}$ is the solution of
\begin{align*}
  \tilde\Psi_{\ag}(\bbeta) =
  \oneN 
  \sumN \left\{\frac{\tdelta_{\ag,i}}{\tpi_{\ag}(Y_{i}, \X_{i})}
  \psi(Y_{i},\X_{i},\bbeta) \right\}  = \0,
\end{align*}
where
\begin{align*}
  \tdelta_{\ag,i}
  &= \bI(\delta_{\pl,i}=1 \text{ or } \tdelta_i=1)
   = \delta_{\pl,i}+(1-\delta_{\pl,i})\tdelta_i,\\
  \tpi_{\ag}(Y_i,\X_i)
  &= \pi_{\pl}(Y_i,\X_i)
    +\{1-\pi_{\pl}(Y_i,\X_i)\}\tpi(Y_i,\X_i),
     \quad i=1,\dots,N.
\end{align*}
We next establish the asymptotic
distribution of $\cbeta_{\ag}$.

Using a Taylor expansion for $\tilde\Psi_{\ag}(\cbeta_{\ag})$ at the point
$\bbeta_{0}$, we have 
\begin{align*}
  \tilde\Psi_{\ag}(\bbeta_{0}) 
  - \left( \int_{0}^{1}\left[  \oneN\sumN \left\{
    \frac{\tdelta_{\ag,i}}{\tpi_{\ag}(Y_{i}, \X_{i})}
    \dot{\mu}\left(\X_{i}\tp\{\bbeta_{0} + \lambda(\cbeta_{\ag} - \bbeta_{0})\}\right)
     \X_{i}\X_{i}\tp \right\} \right] \diff \lambda \right) (\cbeta_{\ag}-\bbeta_{0}) = \0.
\end{align*}
After scaling, we obtain
\begin{align}
  \sqrt{N}\Omega_{\ag,N}^{0}D_{N}(\cbeta_{\ag}-\bbeta_{0})
  =& \{\tilde M_{\ag,N}^{-1} (\cbeta_{\ag})\}
     \{\sqrt{N} \Omega_{\ag,N}^{0} D_{N}^{-1} \tilde\Psi_{\ag}(\bbeta_{0})\}
  \nonumber\\
  =& \{\tilde M_{\ag,N}^{-1} (\cbeta_{\ag})\}
     \{\sqrt{N} \Omega_{\ag,N}^{0} D_{N}^{-1} \Psi^{0}_{\ag}(\bbeta_{0})\}
  \nonumber\\
    &+ \{\tilde M_{\ag,N}^{-1}(\cbeta_{\ag})\} [\sqrt{N} \Omega_{\ag,N}^{0} D_{N}^{-1}
      \{\tilde\Psi_{\ag}(\bbeta_{0}) - \Psi^{0}_{\ag}(\bbeta_{0})\}],
      \label{eq:ag-decomp}
\end{align}
where 
\begin{align*}
  & \tilde M_{\ag,N} (\cbeta_{\ag})
    = \int_{0}^{1}\left[
    \oneN\sumN \left\{  
    \frac{\tdelta_{\ag,i}}{\tpi_{\ag}(Y_{i}, \X_{i})}
    \dot{\mu}\left(\X_{i}\tp
    \{\bbeta_{0} + \lambda(\cbeta_{\ag} - \bbeta_{0})\}
    \right)
    K_{\ag,N}^0(\X_{i})
    \right\} \right]\diff \lambda.
\end{align*}
Since $\pi_{\ag}^{0}(Y,\X)$ satisfies
Assumptions~\ref{asmp:samp-prob-1} and~\ref{asmp:samp-prob-2}, we can
apply the result \eqref{eq:sum-eta} in the proof of
Theorem~\ref{thm:sub1} to the term
$\sqrt{N} \Omega_{\ag,N}^{0} D_{N}^{-1} \Psi^{0}_{\ag}(\bbeta_{0})$ in
Equation~\eqref{eq:ag-decomp}, so that
\begin{align}
  \label{eq:sum-eta-ag}
  \sqrt{N} \Omega_{\ag,N}^{0} D_{N}^{-1} \Psi^{0}_{\ag}(\bbeta_{0})
  \xrightarrow{d}
  \Nor(\0, \Lambda_{\ag}^{0}(\bbeta_{0})).
\end{align}
It therefore remains to establish
\begin{align}
  \sqrt{N} \Omega_{\ag,N}^{0} D_{N}^{-1}
  \{\tilde\Psi_{\ag}(\bbeta_{0}) - \Psi^{0}_{\ag}(\bbeta_{0})\}
  &= o_{p}(1), \label{eq:ag-remainder} \\
  \tilde M_{\ag,N}(\cbeta_{\ag})
  & \xrightarrow{p} M_{\ag}^{0}(\bbeta_{0}). \label{eq:M-ag-N}
\end{align}
Equation~\eqref{eq:ag-remainder} is established in
Lemma~\ref{lem:remainder-variance}, and \eqref{eq:M-ag-N} is established in
Lemma~\ref{lem:M-ag}. Then we have
\begin{align*}
  & \sqrt{N}\Omega_{\ag,N}^{0}D_{N}(\cbeta_{\ag}-\bbeta_{0}) 
    = \{M_{\ag}^{0}(\bbeta_{0})\}^{-1}
    \{\sqrt{N} \Omega_{\ag,N}^{0} D_{N}^{-1} \Psi^{0}_{\ag}(\bbeta_{0})\} + o_{p}(1).
\end{align*}
The conclusion of Theorem~\ref{thm:ag} then follows from
\eqref{eq:sum-eta-ag}, the continuous mapping theorem, and Slutsky's
theorem.
\end{proof}

\begin{lemma} \label{lem:ag-W}
Define
\begin{align*}
  W_{i} = \frac{\tdelta_{\ag,i}}{\tpi_{\ag}(Y_{i},\X_{i})} - \frac{\delta^{0}_{\ag,i}}{\pi_{\ag}^{0}(Y_{i},\X_{i})}, \quad i=1,\dots,N.
\end{align*}
Then $W_{i} = W_{i,1} + W_{i,2}$, where 
\begin{align*}
  W_{i,1}
  =& \left[(1-\delta_{\pl,i})\pi_{\pl}(Y_{i},\X_{i})
     - \delta_{\pl,i}\{1-\pi_{\pl}(Y_{i},\X_{i})\}\right]
     \frac{\tpi(Y_{i},\X_{i}) -\pi^{0}(Y_{i},\X_{i})}
     {\tpi_{\ag}(Y_{i},\X_{i})\pi^{0}_{\ag}(Y_{i},\X_{i})}, \\
  W_{i,2}
  =&  (1-\delta_{\pl,i})\left(
    \frac{\tdelta_{i} - \tpi(Y_{i},\X_{i})}{\tpi_{\ag}(Y_{i},\X_{i})}
    - \frac{\delta_{i}^{0} - \pi^{0}(Y_{i},\X_{i})}{\pi^{0}_{\ag}(Y_{i},\X_{i})}
    \right).
\end{align*}
Under Assumptions~\ref{asmp:regularity-1}-\ref{asmp:regularity-3} and
the conditions in Theorem~\ref{thm:ag}, we have the following results:
\begin{align}
   |W_{i,1}| 
  & \le \frac{1}{\alpha^{2}h_{N}}\{\delta_{\pl,i} + \pi_{\pl}(Y_{i},\X_{i})\}
    L_{\pl}(Y_{i},\X_{i}) \|\cbeta_{\pl} - \bbeta_{0}\|, \label{eq:W1}\\
  \Exp\{ W_{i,2} \mid \Fn, \{V_{\pl,i}\}_{i=1}^{N} \}
  & = 0, \label{eq:W2} \\
  \Exp\{ |W_{i,2}| \mid \Fn, \{V_{\pl,i}\}_{i=1}^{N} \}
  &  \le
    \frac{4L_{\pl}(Y_{i},\X_{i})\|\cbeta_{\pl}-\bbeta_{0}\|}{\alpha},
    \label{eq:W2-first} \\
  \Exp\{ W_{i,2}^{2} \mid \Fn, \{V_{\pl,i}\}_{i=1}^{N} \}
  & \le \frac{4}{\alpha h_{N}},
  \label{eq:W2-second-looser}\\
  \Exp\{ W_{i,2}^{2} \mid \Fn, \{V_{\pl,i}\}_{i=1}^{N} \}
  & \le
  \frac{2L_{\pl}(Y_{i},\X_{i})\|\cbeta_{\pl}-\bbeta_{0}\|}{\alpha^{2}h_{N}},
    \quad i=1,\dots,N. \label{eq:W2-second}
\end{align}
\end{lemma}

\begin{proof}
We can decompose $W_{i} = W_{i,1}+W_{i,2}$ by the following
derivation:
\begin{align*}
  W_{i}
  =&
     \delta_{\pl,i}
    \left(\frac{1}{\tpi_{\ag}(Y_{i},\X_{i})} -
    \frac{1}{\pi^{0}_{\ag}(Y_{i},\X_{i})}  \right) 
    + (1-\delta_{\pl,i}) \left(
    \frac{\tdelta_{i}}{\tpi_{\ag}(Y_{i},\X_{i})} -
    \frac{\delta^{0}_{i}}{\pi^{0}_{\ag}(Y_{i},\X_{i})}  \right) \\
  =&
    - \delta_{\pl,i}\{1-\pi_{\pl}(Y_{i},\X_{i})\}
     \frac{\tpi(Y_{i},\X_{i}) -\pi^{0}(Y_{i},\X_{i})}
     {\tpi_{\ag}(Y_{i},\X_{i})\pi^{0}_{\ag}(Y_{i},\X_{i})}  \\
   &+
     (1-\delta_{\pl,i})  \left(
    \frac{\tpi(Y_{i},\X_{i})}{\tpi_{\ag}(Y_{i},\X_{i})}
     - \frac{\pi^{0}(Y_{i},\X_{i})}{\pi_{\ag}^{0}(Y_{i},\X_{i})}
     +
    \frac{\tdelta_{i} - \tpi(Y_{i},\X_{i})}{\tpi_{\ag}(Y_{i},\X_{i})}
    - \frac{\delta^{0}_{i} - \pi^{0}(Y_{i},\X_{i})}{\pi^{0}_{\ag}(Y_{i},\X_{i})}
    \right)  \\
  =&
     - \delta_{\pl,i}\{1-\pi_{\pl}(Y_{i},\X_{i})\}
     \frac{\tpi(Y_{i},\X_{i}) -\pi^{0}(Y_{i},\X_{i})}
     {\tpi_{\ag}(Y_{i},\X_{i})\pi^{0}_{\ag}(Y_{i},\X_{i})}
   + (1-\delta_{\pl,i}) \pi_{\pl}(Y_{i},\X_{i})
     \frac{\tpi(Y_{i},\X_{i}) -\pi^{0}(Y_{i},\X_{i})}
    {\tpi_{\ag}(Y_{i},\X_{i})\pi^{0}_{\ag}(Y_{i},\X_{i})} \\
   &
     + (1-\delta_{\pl,i})\left(
    \frac{\tdelta_{i} - \tpi(Y_{i},\X_{i})}{\tpi_{\ag}(Y_{i},\X_{i})}
    - \frac{\delta_{i}^{0} - \pi^{0}(Y_{i},\X_{i})}{\pi^{0}_{\ag}(Y_{i},\X_{i})}
     \right)\\
  =&  W_{i,1} + W_{i,2}.
\end{align*}
For $|W_{i,1}|$,
\begin{align*}
  & |W_{i,1}| \\ 
  & \le \frac{1}{\alpha^{2}h_{N}^{2}} \left[
  \pi_{\pl}(Y_{i},\X_{i}) (1 - \delta_{\pl,i})
    |\tpi(Y_{i},\X_{i}) -\pi^{0}(Y_{i},\X_{i})|
    +
    \delta_{\pl,i} \{1 - \pi_{\pl}(Y_{i},\X_{i})\}
    |\tpi(Y_{i},\X_{i}) -\pi^{0}(Y_{i},\X_{i})|
    \right] \nonumber \\
  & \le \frac{1}{\alpha^{2}h_{N}^{2}}
    \{\delta_{\pl,i} + \pi_{\pl}(Y_{i},\X_{i})\}
    |\tpi(Y_{i},\X_{i}) -\pi^{0}(Y_{i},\X_{i})| \\
  & \le \frac{1}{\alpha^{2}h_{N}}\{\delta_{\pl,i} + \pi_{\pl}(Y_{i},\X_{i})\}
    L_{\pl}(Y_{i},\X_{i}) \|\cbeta_{\pl} - \bbeta_{0}\|,
\end{align*}
where the last inequality is from Assumption~\ref{asmp:ag} (ii).

For \eqref{eq:W2}, we have
$\Exp\{ W_{i,2} \mid \Fn, \{V_{\pl,i}\}_{i=1}^{N} \} = 0$ because
$\Exp\{\tdelta_{i} - \tpi(Y_{i},\X_{i}) \mid \Fn, \{V_{\pl,i}\}_{i=1}^{N} \} = 0$
and
$\Exp\{\delta_{i}^{0} - \pi^{0}(Y_{i},\X_{i})\mid \Fn, \{V_{\pl,i}\}_{i=1}^{N} \} = 0$.
For $\Exp\{ |W_{i,2}| \mid \Fn, \{V_{\pl,i}\}_{i=1}^{N} \}$, note
that
$\Exp\{ |\tdelta_{i} - \delta^{0}_{i}| \mid \Fn, \{V_{\pl,i}\}_{i=1}^{N} \} = |\tpi(Y_{i},\X_{i}) -\pi^{0}(Y_{i},\X_{i})|$
and $\pi^{0}(Y_{i},\X_{i}) < \pi^{0}_{\ag}(Y_{i},\X_{i})$; hence,
\begin{align*}
  & \Exp\{ |W_{i,2}| \mid \Fn, \{V_{\pl,i}\}_{i=1}^{N} \} \\
  \le &
  \Exp\left\{   \left|
    \frac{\tdelta_{i} - \tpi(Y_{i},\X_{i})}{\tpi_{\ag}(Y_{i},\X_{i})}
    -
    \frac{\delta_{i}^{0} - \pi^{0}(Y_{i},\X_{i})}{\pi^{0}_{\ag}(Y_{i},\X_{i})}
    \right| \middle| \Fn, \{V_{\pl,i}\}_{i=1}^{N} \right\} \\
  \le &
    \Exp\left\{
    \frac{|\tdelta_{i} - \delta_{i}^{0}| + |\pi^{0}(Y_{i},\X_{i}) - \tpi(Y_{i},\X_{i})|}
    {\tpi_{\ag}(Y_{i},\X_{i})}
    \middle| \Fn, \{V_{\pl,i}\}_{i=1}^{N} \right\} \\
  & +
    \Exp\left\{ \left|\delta_{i}^{0} - \pi^{0}(Y_{i},\X_{i})\right|
    \left|
    \frac{\pi^{0}_{\ag}(Y_{i},\X_{i}) - \tpi_{\ag}(Y_{i},\X_{i})}
    {\tpi_{\ag}(Y_{i},\X_{i})\pi^{0}_{\ag}(Y_{i},\X_{i})}
    \right|
    \middle| \Fn, \{V_{\pl,i}\}_{i=1}^{N} \right\} \\
  \le & \frac{\Exp\left\{ |\tdelta_{i} - \delta_{i}^{0}| \middle| \Fn, \{V_{\pl,i}\}_{i=1}^{N} \right\}
      + |\pi^{0}(Y_{i},\X_{i}) - \tpi(Y_{i},\X_{i})|}
      {\tpi_{\ag}(Y_{i},\X_{i})}  \\
  & +
    \Exp\left\{ \left|\delta_{i}^{0} - \pi^{0}(Y_{i},\X_{i})\right|
    \middle| \Fn, \{V_{\pl,i}\}_{i=1}^{N} \right\}
  \frac{ \left|
    \pi^{0}_{\ag}(Y_{i},\X_{i}) - \tpi_{\ag}(Y_{i},\X_{i}) \right| }
    {\tpi_{\ag}(Y_{i},\X_{i})\pi^{0}_{\ag}(Y_{i},\X_{i})}
     \\
  \le &
      \frac{2|\pi^{0}(Y_{i},\X_{i}) - \tpi(Y_{i},\X_{i})|}
        {\tpi_{\ag}(Y_{i},\X_{i})}
        + \frac{2 \pi^{0}(Y_{i},\X_{i})  |\pi^{0}(Y_{i},\X_{i}) - \tpi(Y_{i},\X_{i})|}
        {\tpi_{\ag}(Y_{i},\X_{i})\pi^{0}_{\ag}(Y_{i},\X_{i})} \\
  \le & \frac{4}{\alpha}  L_{\pl}(Y_{i},\X_{i}) \|\cbeta_{\pl} - \bbeta_{0}\|.
\end{align*}

For $\Exp\{ W_{i,2}^{2} \mid \Fn, \{V_{\pl,i}\}_{i=1}^{N} \}$, the
first bound \eqref{eq:W2-second-looser} follows from
\begin{align*}
  & \Exp\{ W_{i,2}^{2} \mid \Fn, \{V_{\pl,i}\}_{i=1}^{N} \} \\
  & \le
    \frac{2 \Exp \left\{ (
    \tdelta_{i} - \tpi(Y_{i},\X_{i}))^{2}
    \middle| \Fn, \{V_{\pl,i}\}_{i=1}^{N}
    \right\} }
    {\{\tpi_{\ag}(Y_{i},\X_{i})\}^{2}}
  + \frac{2\Exp \left\{ \left(
     \delta_{i}^{0} - \pi^{0}(Y_{i},\X_{i})  \right)^{2}
    \middle| \Fn, \{V_{\pl,i}\}_{i=1}^{N}
    \right\}}
    {\{\pi^{0}_{\ag}(Y_{i},\X_{i})\}^{2}} \\
  & \le
     \frac{2 \tpi(Y_{i},\X_{i}) }
    {\{\tpi_{\ag}(Y_{i},\X_{i})\}^{2}}
  + \frac{2 \pi^{0}(Y_{i},\X_{i})}
    {\{\pi^{0}_{\ag}(Y_{i},\X_{i})\}^{2}} 
  \le \frac{4}{\alpha h_{N}}.
\end{align*}
To derive \eqref{eq:W2-second}, without
loss of generality, suppose
$\tpi(Y_{i},\X_{i}) > \pi^{0}(Y_{i},\X_{i})$. Then
\begin{align*}
  &\Exp\{ W_{i,2}^{2} \mid \Fn, \{V_{\pl,i}\}_{i=1}^{N} \} \\
  &\le \Exp \left\{ \left(
    \frac{\tdelta_{i} - \tpi(Y_{i},\X_{i})}{\tpi_{\ag}(Y_{i},\X_{i})}
    - \frac{\delta_{i}^{0} - \pi^{0}(Y_{i},\X_{i})}{\pi^{0}_{\ag}(Y_{i},\X_{i})}
  \right)^{2}
    \middle| \Fn, \{V_{\pl,i}\}_{i=1}^{N}
    \right\} \\
  & = \Var \left\{ \left(
    \frac{\tdelta_{i}}{\tpi_{\ag}(Y_{i},\X_{i})}
    - \frac{\delta_{i}^{0}}{\pi^{0}_{\ag}(Y_{i},\X_{i})}  \right)
    \middle| \Fn, \{V_{\pl,i}\}_{i=1}^{N}
    \right\} \\
  &\le
    \Exp \left\{ \left(
    \frac{\tdelta_{i}}{\tpi_{\ag}(Y_{i},\X_{i})}
    - \frac{\delta_{i}^{0}}{\pi^{0}_{\ag}(Y_{i},\X_{i})} \right)^{2}
    \middle| \Fn, \{V_{\pl,i}\}_{i=1}^{N}  \right\} \\
  & =
    \pi^{0}(Y_{i},\X_{i})
    \left( \frac{1}{\tpi_{\ag}(Y_{i},\X_{i})}
    - \frac{1}{\pi^{0}_{\ag}(Y_{i},\X_{i})} \right)^{2}
    + \frac{|\tpi(Y_{i},\X_{i}) - \pi^{0}(Y_{i},\X_{i})|}
    {\{\tpi_{\ag}(Y_{i},\X_{i})\}^{2}} \\
  & =
    \frac{\pi^{0}(Y_{i},\X_{i})\{1 - \pi_{\pl}(Y_{i},\X_{i})\}^{2}
    \{\tpi(Y_{i},\X_{i}) - \pi^{0}(Y_{i},\X_{i})\}^{2} }
    {\{\tpi_{\ag}(Y_{i},\X_{i})\pi^{0}_{\ag}(Y_{i},\X_{i})\}^{2}}
    + \frac{|\tpi(Y_{i},\X_{i}) - \pi^{0}(Y_{i},\X_{i})|}
    {\{\tpi_{\ag}(Y_{i},\X_{i})\}^{2}} \\
  & \le
    \frac{\pi^{0}(Y_{i},\X_{i})
    |\tpi(Y_{i},\X_{i}) - \pi^{0}(Y_{i},\X_{i})|^{2} }
    {\{\tpi_{\ag}(Y_{i},\X_{i})\pi^{0}_{\ag}(Y_{i},\X_{i})\}^{2}}
    + \frac{|\tpi(Y_{i},\X_{i}) - \pi^{0}(Y_{i},\X_{i})|}
    {\{\tpi_{\ag}(Y_{i},\X_{i})\}^{2}} \\
  & \le
    \frac{|\tpi(Y_{i},\X_{i}) - \pi^{0}(Y_{i},\X_{i})| }
    {\tpi_{\ag}(Y_{i},\X_{i}) \pi^{0}_{\ag}(Y_{i},\X_{i})}
    + \frac{|\tpi(Y_{i},\X_{i}) - \pi^{0}(Y_{i},\X_{i})|}
    {\{\tpi_{\ag}(Y_{i},\X_{i})\}^{2}} \\
  & \le
    \frac{2L_{\pl}(Y_{i},\X_{i})\|\cbeta_{\pl}-\bbeta_{0}\|}{\alpha^{2}h_{N}}.
\end{align*}
\end{proof}

\begin{lemma} \label{lem:remainder-variance}
Let
$Q_{i} = \Omega_{\ag,N}^{0} D_{N}^{-1}\psi(Y_{i},\X_{i},\bbeta_{0})$, $i=1,\dots,N$.
By Lemma~\ref{lem:ag-W}, the left side of \eqref{eq:ag-remainder} can
be written as $N^{-1/2}\sumN (W_{i,1} Q_{i} + W_{i,2} Q_{i})$.  Then,
under Assumptions~\ref{asmp:regularity-1}-\ref{asmp:regularity-3} and
the conditions in Theorem~\ref{thm:ag}, we have
$N^{-1/2}\sumN (W_{i,1} Q_{i})=o_{p}(1)$ and
$N^{-1/2}\sumN (W_{i,2} Q_{i})=o_{p}(1)$.
\end{lemma}

\begin{proof}
By
\eqref{eq:W1}, for $i=1,\dots,N$, we obtain
\begin{align*}
  \left\| \frac{1}{\sqrt{N}} \sumN (W_{i,1} Q_{i}) \right\|
  & \le  \frac{1}{\sqrt{N}} \sumN (|W_{i,1}| \| Q_{i} \|) \\
  & \le
    \frac{\|\cbeta_{\pl}-\bbeta_{0}\|}{\alpha^{2} h_{N}}
  \frac{1}{\sqrt{N}} \sumN \left[  
    \{\delta_{\pl,i}+\pi_{\pl}(Y_{i},\X_{i})\}
    L_{\pl}(Y_{i},\X_{i}) \|Q_{i}\| \right]\\
  & \le
    \frac{\|\cbeta_{\pl}-\bbeta_{0}\|}{\alpha^{2} h_{N}}
  \frac{1}{\sqrt{N}} \sumN \left[ 
    \{\delta_{\pl,i}+\pi_{\pl}(Y_{i},\X_{i})\} L_{\pl}(Y_{i},\X_{i})
    \|D_{N}^{-1}\psi(Y_{i},\X_{i},\bbeta_{0})\|  \right],
\end{align*}
where we used
$ \|Q_{i}\| \le \|D_{N}^{-1}\psi(Y_{i},\X_{i},\bbeta_{0})\|$ because
the diagonal entries of $\Omega_{\ag,N}^0$ are at most one.  By
\eqref{eq:ag-GN-ineq-3}, we have
\begin{align*}
  \Exp\left[ \pi_{\pl}(Y,\X) L_{\pl}(Y,\X) \|D_{N}^{-1}\psi(Y,\X,\bbeta_{0})\| \right]
  \le \Exp\left[ \pi_{\pl}(Y,\X) L_{\pl}(Y,\X) \|D_{N}^{-2}\psi(Y,\X,\bbeta_{0})\| \right]
  = O(h_{\pl,N}).
\end{align*}
By Assumption~\ref{asmp:ag} (i),
$\|\cbeta_{\pl}-\bbeta_{0}\| = O_{p}(n_{\pl}^{-1/2})$. Therefore, by
Markov's inequality and $\sqrt{h_{\pl,N}} / h_{N} \to 0$, we have 
\begin{align*}
  \left\| \frac{1}{\sqrt{N}}\sumN (W_{i,1} Q_{i})  \right\|
  \le O_{p}\left(\frac{\sqrt{N}h_{\pl,N}}{\sqrt{n_{\pl}}h_{N}}  \right)
  = O_{p}\left(\frac{\sqrt{h_{\pl,N}}}{h_{N}}  \right)
  \to 0.
\end{align*}

For $N^{-1/2}\sumN (W_{i,2} Q_{i})$, note that, conditional on
$(\Fn, \{V_{\pl,i}\}_{i=1}^{N})$, $\{W_{i,2} Q_{i}\}_{i=1}^{N}$ are
independent. Therefore, for $i=1,\dots,N$,
\begin{align*}
  & \Exp(W_{i,2} Q_{i} \mid \Fn, \{V_{\pl,i}\}_{i=1}^{N}) \\
  & = \Exp\left[ \left\{(1-\delta_{\pl,i})\left(
    \frac{\tdelta_{i} - \tpi(Y_{i},\X_{i})}{\tpi_{\ag}(Y_{i},\X_{i})}
    - \frac{\delta_{i}^{0} - \pi^{0}(Y_{i},\X_{i})}{\pi^{0}_{\ag}(Y_{i},\X_{i})}
  \right) Q_{i} \right\}  \middle| \Fn, \{V_{\pl,i}\}_{i=1}^{N}\right] \\
  & =  (1-\delta_{\pl,i})\left(
    \frac{\tpi(Y_{i},\X_{i}) - \tpi(Y_{i},\X_{i})}{\tpi_{\ag}(Y_{i},\X_{i})}
    - \frac{\pi^{0}(Y_{i},\X_{i}) - \pi^{0}(Y_{i},\X_{i})}{\pi^{0}_{\ag}(Y_{i},\X_{i})}
  \right) Q_{i}   = \0,
\end{align*}
Thus, $\Exp\{N^{-1/2}\sumN (W_{i,2} Q_{i})\} = \0$. For the variance,
first consider its conditional variance. By \eqref{eq:W2-second},
\begin{align*}
  & \Var\left\{
    \frac{1}{\sqrt{N}}\sumN (W_{i,2} Q_{i}) \middle|
    \Fn, \{V_{\pl,i}\}_{i=1}^{N} \right\} \\
  &=
    \oneN \sumN
    \Exp\left\{W_{i,2}^{2} \|Q_{i}\|^{2} \mid \Fn, \{V_{\pl,i}\}_{i=1}^{N} \right\}  
    \le
    \frac{2\|\cbeta_{\pl}-\bbeta_{0}\|}{\alpha^{2}h_{N}}
    \oneN \sumN  \{L_{\pl}(Y_{i},\X_{i})\|Q_{i}\|^{2}\}.
\end{align*}
Moreover, by Assumption~\ref{asmp:ag} (ii) and \eqref{eq:ag-TN-ineq-2},
for $i=1,\dots,N$,
\begin{align*}
  \sup_{N}\Exp\{L_{\pl}(Y_{i},\X_{i}) \|Q_{i}\|^{2}\}
  \le \sup_{N}\Exp\{L_{\pl}(Y_{i},\X_{i})
  \|D_{N}^{-1}\psi(Y_{i},\X_{i},\bbeta_{0})\|^{2}\}
  < \infty,
\end{align*}
we
have
\begin{align*}
  \Var\left\{N^{-1/2}\sumN (W_{i,2} Q_{i}) \middle| \Fn, \{V_{\pl,i}\}_{i=1}^{N} \right\}
  = O_{p}\left( \frac{1}{h_{N}\sqrt{n_{\pl}}} \right) \to 0.
\end{align*}
Thus, $\Var\{N^{-1/2}\sumN (W_{i,2} Q_{i})\} \to 0$. This completes
the proof of Lemma~\ref{lem:remainder-variance}.
\end{proof}

\begin{lemma} \label{lem:ag-feasible-consistency-step1}
Let $J_{i}(\bbeta) = D_{N}^{-2}\psi(Y_{i},\X_{i},\bbeta)$, $i=1,\dots,N$. Under
Assumptions~\ref{asmp:regularity-1}-\ref{asmp:regularity-3} and the
conditions in Theorem~\ref{thm:ag},
  \begin{align*}
  \sup_{\bbeta \in\Theta}
    \|D_{N}^{-2}\{\tilde\Psi_{\ag}(\bbeta) - \Psi_{\ag}^{0}(\bbeta) \} \|
    = \sup_{\bbeta \in\Theta} \left\|\oneN \sumN W_{i} J_{i}(\bbeta) \right\|
    \xrightarrow{p} 0.
  \end{align*}
\end{lemma}

\begin{proof}
Note that
$\|J_{i}(\bbeta)\| \le \|J_{i}(\bbeta_{0})\| + \|J_{i}(\bbeta) - J_{i}(\bbeta_{0})\|$
and
$\|J_{i}(\bbeta) - J_{i}(\bbeta_{0})\| \le L(\X_{i})\|\bbeta - \bbeta_{0}\| \|D_{N}^{-2}\X\|$
by Assumption~\ref{asmp:regularity-2} (iii), thus
\begin{align*}
  \sup_{\bbeta \in\Theta}\left\|\oneN \sumN W_{i} J_{i}(\bbeta) \right\|
  \le
  \left\| \oneN \sumN \{W_{i} J_{i}(\bbeta_{0})\}  \right\|
  +
  \left( \sup_{\bbeta \in\Theta}\|\bbeta - \bbeta_{0}\| \right)
  \oneN \sumN \{ |W_{i}| L(\X_{i})\|D_{N}^{-2}\X_{i}\| \}.
\end{align*}
Since $\sup_{\bbeta \in\Theta}\|\bbeta - \bbeta_{0}\| < \infty$
because $\Theta$ is compact, it suffices to prove
\begin{align}
  & \left\|  \oneN \sumN \{W_{i,1} J_{i}(\bbeta_{0})\}\right\| = o_{p}(1), \label{eq:wq11}\\
  & \oneN \sumN \left\{ |W_{i,1}| L(\X_{i})\|D_{N}^{-2}\X_{i}\|\right\} = o_{p}(1), \label{eq:wq12} \\
  & \left\|\oneN \sumN \{W_{i,2} J_{i}(\bbeta_{0})\} \right\| = o_{p}(1), \label{eq:wq21}\\
  & \oneN \sumN \left\{ |W_{i,2}| L(\X_{i})\|D_{N}^{-2}\X_{i}\|\right\} = o_{p}(1). \label{eq:wq22}
\end{align}
For \eqref{eq:wq11}, using \eqref{eq:W1}, we obtain
\begin{align*}
  \left\|\oneN \sumN W_{i,1} J_{i}(\bbeta_{0}) \right\|
  & \le
  \oneN \sumN |W_{i,1}| \left\| J_{i}(\bbeta_{0}) \right\| \\
  & \le 
  \frac{\|\cbeta_{\pl} - \bbeta_{0}\|}{\alpha^{2} N h_{N}}
  \sumN \{\delta_{\pl,i}+\pi_{\pl}(Y_{i},\X_{i})\} L_{\pl}(Y_{i},\X_{i}) \|J_{i}(\bbeta_{0})\|  \\
  & =
    O_{p}\left( \frac{h_{\pl,N}}{h_{N}\sqrt{n_{\pl}}} \right)
  \to 0,
\end{align*}
where we use  \eqref{eq:ag-GN-ineq-3}:
$\sup_{N}h_{\pl,N}^{-1}\Exp\{\pi_{\pl}(Y,\X) L_{\pl}(Y,\X) \|D_{N}^{-2}\psi(Y,\X,\bbeta)\|\} < \infty$,
$\sqrt{h_{\pl,N}} / h_{N} \to 0$, $n_{\pl} \to \infty$ and Markov's
inequality.  Similarly, for \eqref{eq:wq12}, 
\begin{align*}
  \oneN \sumN |W_{i,1}| L(\X_{i})\|D_{N}^{-2}\X_{i}\|
  & \le 
  \frac{\|\cbeta_{\pl} - \bbeta_{0}\|}{\alpha^{2} N h_{N}}
  \sumN \{\delta_{\pl,i}+\pi_{\pl}(Y_{i},\X_{i})\} L_{\pl}(Y_{i},\X_{i}) L(\X_{i})\|D_{N}^{-2}\X_{i}\|  \\
  & =
    O_{p}\left( \frac{1}{h_{N}\sqrt{n_{\pl}}} \right)
  \to 0,
\end{align*}
where we use \eqref{eq:ag-TN-ineq-1}: 
\begin{align*}
  \sup_{N}\Exp\{\pi_{\pl}(Y,\X) L_{\pl}(Y,\X) L(\X) \|D_{N}^{-2}\X\|\}
  \le \sup_{N}\Exp\{ L_{\pl}(Y,\X) L(\X) \|K_{\ag,N}^0(\X)\|_F\}
  < \infty.
\end{align*}
Next, for \eqref{eq:wq21}, $\Exp[N^{-1} \sumN \{W_{i,2} J_{i}(\bbeta_{0})\}] =0$ and 
\begin{align*}
  & \Var\left(\oneN \sumN \{W_{i,2} J_{i}(\bbeta_{0})\} 
  \middle| \Fn, \{V_{\pl,i}\}_{i=1}^{N} \right) \\
  & =
    \frac{1}{N^{2}} \sumN \Var\left( W_{i,2} J_{i}(\bbeta_{0})
  \middle| \Fn, \{V_{\pl,i}\}_{i=1}^{N} \right) \\
  & =
    \frac{1}{N^{2}}\sumN \Exp\left( W_{i,2}^{2} \left\| J_{i}(\bbeta_{0}) \right\|^{2}
    \middle| \Fn, \{V_{\pl,i}\}_{i=1}^{N}  \right) \\
  & \le \frac{2\|\cbeta_{\pl} - \bbeta_{0}\|}{\alpha^{2} N h_{N} } 
  \oneN \sumN \{L_{\pl}(Y_{i},\X_{i}) \|J_{i}(\bbeta_{0})\|^{2}\}.
\end{align*}
If we denote
$\rho_{N,\min} = \min\{\rho_{1,N},\dots,\rho_{d_{r},N}\}$ and $J(\bbeta) = D_{N}^{-2}\psi(Y,\X,\bbeta)$, then by
\eqref{eq:ag-TN-ineq-2}, as $Nh_{N}\rho_{N,\min} \to \infty$, 
\begin{align*}
  \frac{1}{Nh_{N}}\Exp\{L_{\pl}(Y,\X) \|J(\bbeta_{0})\|^{2}\}
  \le
  \frac{1}{Nh_{N}\rho_{N,\min}}\Exp\{L_{\pl}(Y,\X)
  \|D_{N}^{-1}\psi(Y,\X,\bbeta_{0})\|^{2}\} \to 0.
\end{align*}
Combining with $\|\cbeta_{\pl} - \bbeta_{0}\| = o_{p}(1)$, we have 
$\Var(N^{-1} \sumN \{W_{i,2} J_{i}(\bbeta_{0})\}) = o(1)$
and \eqref{eq:wq21} is proved. Lastly, for \eqref{eq:wq22}, by
\eqref{eq:ag-TN-ineq-1} and \eqref{eq:W2-first},
\begin{align*}
  & \Exp\left[ \oneN \sumN \{ |W_{i,2}| L(\X_{i})\|D_{N}^{-2}\X_{i}\| \}
  \middle| \Fn, \{V_{\pl,i}\}_{i=1}^{N}   \right] \\
  & \le
    \frac{4\|\cbeta_{\pl} - \bbeta_{0}\|}{\alpha}
    \oneN\sumN[L_{\pl}(Y_{i},\X_{i}) L(\X_{i})\|D_{N}^{-2}\X_{i}\|] \\
  & \le
    \frac{4\|\cbeta_{\pl} - \bbeta_{0}\|}{\alpha}
    \oneN\sumN[L_{\pl}(Y_{i},\X_{i}) L(\X_{i})\|K_{\ag,N}^0(\X_{i})\|_F] \\  
  &  =
    O_{p}\left(\frac{1}{\sqrt{n_{\pl}}}  \right) =o_{p}(1).
\end{align*}
Therefore, by Markov's inequality, \eqref{eq:wq22} is proved.  The
proof of Lemma~\ref{lem:ag-feasible-consistency-step1} is completed.
\end{proof}

\begin{lemma} \label{lem:ag-feasible-consistency}
Under Assumptions~\ref{asmp:regularity-1}-\ref{asmp:regularity-3} and
the conditions in Theorem~\ref{thm:ag}, if
$N h_{N} \rho_{j,N} \to \infty$, $j=1,\dots,d_{r}$, then
$\cbeta_{\ag} \xrightarrow{p} \bbeta_{0}$.
\end{lemma}
\begin{proof}
Following the argument used to derive \eqref{eq:DN-2-Psi}, it remains to show that
$\sup_{\bbeta \in\Theta} \|D_{N}^{-2}\{\tilde\Psi_{\ag}(\bbeta) - \Psi_{0}(\bbeta)\}\|\xrightarrow{p} 0$.
First, since $\pi^{0}_{\ag}(Y,\X)$ satisfies
Assumptions~\ref{asmp:samp-prob-1} and \ref{asmp:samp-prob-2}, the
proof of Lemma~\ref{lem:sub-consistency} applies verbatim after
replacing
$\{\delta_{i},\pi(Y_{i},\X_{i}),\Omega_{N}, \Psi_{n}(\bbeta)\}$ with
$\{\delta^{0}_{\ag,i}, \pi_{\ag}^{0}(Y_{i},\X_{i}),\Omega_{\ag,N}^{0}, \Psi_{\ag}^{0}(\bbeta)\}$.
Also, note that under $\sqrt{h_{\pl,N}}/h_N\to0$, we have
$h_{\pl,N}=o(h_N^2)$, hence $h_{\pl,N}=o(h_N)$ since $h_{N}\le 1$. Therefore
$h_{\ag,N}^0$ is of the same order as $h_N$. Thus, if
$N h_{N} \rho_{j,N} \to \infty$, $j=1,\dots,d_{r}$, we have
$\sup_{\bbeta \in\Theta} \|D_{N}^{-2}\{\Psi_{\ag}^{0}(\bbeta) - \Psi_{0}(\bbeta) \}\| \xrightarrow{p} 0$
and $\cbeta_{\ag}^{0} \xrightarrow{p} \bbeta_{0}$. Note that
\begin{align*}
  \|D_{N}^{-2}\{\tilde\Psi_{\ag}(\bbeta) - \Psi_{0}(\bbeta)\}\|
  \le \|D_{N}^{-2}\{\tilde\Psi_{\ag}(\bbeta) - \Psi_{\ag}^{0}(\bbeta) \}\|
  + \|D_{N}^{-2}\{\Psi_{\ag}^{0}(\bbeta) - \Psi_{0}(\bbeta)\}\|.
\end{align*}
Thus,
it remains to show that
$\sup_{\bbeta \in\Theta} \|D_{N}^{-2}\{\tilde\Psi_{\ag}(\bbeta) - \Psi_{\ag}^{0}(\bbeta) \} \| \xrightarrow{p} 0$,
which follows from
Lemma~\ref{lem:ag-feasible-consistency-step1}.
\end{proof}

\begin{lemma} \label{lem:M-ag}
Under Assumptions~\ref{asmp:regularity-1}-\ref{asmp:regularity-3} and
the conditions in Theorem~\ref{thm:ag}, we have 
$\tilde M_{\ag,N}(\cbeta_{\ag}) \xrightarrow{p} M_{\ag}^{0}(\bbeta_{0})$.
\end{lemma}
\begin{proof}
It suffices to show that for $k,l=1,\dots,d$,
$\{\tilde M_{\ag,N}(\cbeta_{\ag}) - M_{\ag}^{0}(\bbeta_{0})\}_{k,l}$,
the $(k,l)$-th element of the matrix
$\{\tilde M_{\ag,N}(\cbeta_{\ag}) - M_{\ag}^{0}(\bbeta_{0})\}$, is
$o_{p}(1)$.  We write
$\tilde M_{\ag,N}(\cbeta_{\ag}) - M_{\ag}^{0}(\bbeta_{0}) = G_{1} + G_{2} + G_{3}$,
where
\begin{align*}
 G_{1}=& \int_{0}^{1}\left[\oneN\sumN\left\{
    \left( \frac{\tdelta_{\ag,i}}{\tpi_{\ag}(Y_{i}, \X_{i})}
    - \frac{\delta^{0}_{\ag,i}}{\pi_{\ag}^{0}(Y_{i}, \X_{i})}  \right)
         \dot{\mu}\left(
         \X_{i}\tp\{\bbeta_{0} + \lambda(\cbeta_{\ag} - \bbeta_{0})
         \}\right)
          K_{\ag,N}^{0}(\X_{i})
          \right\} \right]\diff \lambda \\
 G_{2}=& \int_{0}^{1}\left[  \oneN\sumN \left\{
    \frac{\delta^{0}_{\ag,i}}{\pi_{\ag}^{0}(Y_{i}, \X_{i})}  
           \left\{ \dot{\mu}\left(
           \X_{i}\tp\{\bbeta_{0} + \lambda(\cbeta_{\ag} - \bbeta_{0})\}
           \right)
    - \dot{\mu}(\X_{i}\tp\bbeta_{0})\right\}
    K_{\ag,N}^{0}(\X_{i})\right\} \right] \diff \lambda \\
 G_{3}=&  \oneN\sumN \left\{  
    \frac{\delta^{0}_{\ag,i}}{\pi_{\ag}^{0}(Y_{i}, \X_{i})}
    \dot{\mu}(\X_{i}\tp\bbeta_{0})
    K_{\ag,N}^{0}(\X_{i})\right\}
    -M_{\ag}^{0}(\bbeta_{0}).
\end{align*}
Thus,
$|\tilde M_{\ag,N}(\cbeta_{\ag}) - M_{\ag}^{0}(\bbeta_{0})|_{k,l} \le |G_{1}|_{k,l} + |G_{2}|_{k,l} + |G_{3}|_{k,l}$.
We now prove all three terms are $o_{p}(1)$. For $|G_{1}|_{k,l}$,
using Lemma~\ref{lem:ag-W}, we have
\begin{align*}
  & |G_{1}|_{k,l} \\ 
  \le& 
  \int_{0}^{1}\oneN\sumN\left\{ |W_{i,1}|
       \left|\dot{\mu}\left(
       \X_{i}\tp\{\bbeta_{0} + \lambda(\cbeta_{\ag} - \bbeta_{0})\}
       \right)
       - \dot{\mu}(\X_{i}\tp\bbeta_{0}) \right| 
       |K_{\ag,N}^{0}(\X_{i})|_{k,l} \right\}  \diff \lambda \\
  &+ \oneN\sumN\left\{ |W_{i,1}|
       \left| \dot{\mu}(\X_{i}\tp\bbeta_{0}) \right| 
       |K_{\ag,N}^{0}(\X_{i})|_{k,l} \right\}   \\
     &+ \int_{0}^{1}
       \left| \oneN\sumN 
        \left[  W_{i,2} \left\{  
       \dot{\mu}\left(\X_{i}\tp\{\bbeta_{0} + \lambda(\cbeta_{\ag} - \bbeta_{0})\}\right)
       - \dot{\mu}(\X_{i}\tp\bbeta_{0})  \right\}
       K_{\ag,N}^{0}(\X_{i})  \right]\right|_{k,l} \diff \lambda \\
  &+ \left| \oneN\sumN\left\{W_{i,2}
        \dot{\mu}(\X_{i}\tp\bbeta_{0}) K_{\ag,N}^{0}(\X_{i}) \right\}  \right| \\
  \le& G_{11} + G_{12} + G_{13} + G_{14},
\end{align*}
where
\begin{align*}
  G_{11} &= \|\cbeta_{\ag}-\bbeta_{0}\|
       \oneN\sumN\left\{ |W_{i,1}|
           L(\X_{i}) |K_{\ag,N}^{0}(\X_{i})|_{k,l} \right\}, \\
  G_{12}&= \oneN\sumN\left\{ |W_{i,1}|
       \left| \dot{\mu}(\X_{i}\tp\bbeta_{0}) \right| 
          |K_{\ag,N}^{0}(\X_{i})|_{k,l} \right\}, \\
  G_{13} &= \|\cbeta_{\ag}-\bbeta_{0}\|
       \left| \oneN\sumN\left\{ W_{i,2}
       L(\X_{i})
           K_{\ag,N}^{0}(\X_{i}) \right\}  \right|_{k,l}, \\
  G_{14}&=\left|  \oneN\sumN\left\{ W_{i,2}
    \dot{\mu}(\X_{i}\tp\bbeta_{0}) K_{\ag,N}^{0}(\X_{i}) \right\}\right|_{k,l}.
\end{align*}
For $G_{11}$, using the upper bound for $|W_{i,1}|$ in \eqref{eq:W1},
\begin{align*}
  G_{11}
  \le 
  \frac{\|\cbeta_{\ag}-\bbeta_{0}\| \|\cbeta_{\pl}-\bbeta_{0}\|}{\alpha^{2}h_{N}}
  \oneN \sumN \left[
  \{\delta_{\pl,i} + \pi_{\pl}(Y_{i},\X_{i})\} L_{\pl}(Y_{i},\X_{i}) L(\X_{i})
  \| K_{\ag,N}^{0}(\X_{i})\|_{F} \right].
\end{align*}
By \eqref{eq:ag-TN-ineq-1},
\begin{align*}
  \Exp[ \{\delta_{\pl} + \pi_{\pl}(Y,\X)\}
  L_{\pl}(Y,\X) L(\X) \| K_{\ag,N}^{0}(\X)\| ] 
  & = 2\Exp\left\{ \pi_{\pl}(Y,\X)
    L_{\pl}(Y,\X) L(\X) \| K_{\ag,N}^{0}(\X)\|_{F}\right\} \\
  & \le 2\Exp\left\{
    L_{\pl}(Y,\X) L(\X) \| K_{\ag,N}^{0}(\X)\|_{F} \right\}
    <\infty,
\end{align*}
 and since $h_{N}\sqrt{n_{\pl}} \to \infty$,
$\|\cbeta_{\ag}-\bbeta_{0}\| = o_{p}(1)$ and
$\|\cbeta_{\pl}-\bbeta_{0}\|=O_{p}(n_{\pl}^{-1/2})$, we have
\begin{align*}
  G_{11}
  \le o_{p}(1)O_{p}\left( \frac{1}{h_{N}\sqrt{n_{\pl}}} \right)
  =  o_{p}(1)
\end{align*}
by Markov's inequality.  For $G_{12}$, similarly, by \eqref{eq:W1},
\eqref{eq:ag-GN-ineq-2} and Assumption~\ref{asmp:ag} (i),
\begin{align*}
  G_{12}
  & \le \|\cbeta_{\pl}-\bbeta_{0}\| \frac{1}{\alpha^{2}h_{N}} \oneN \sumN
    \left[
    \{\delta_{\pl,i} + \pi_{\pl}(Y_{i},\X_{i})\}
    L_{\pl}(Y_{i},\X_{i})
    |\dot{\mu}(\X_{i}\tp\bbeta_{0})| \|
    K_{\ag,N}^{0}(\X_{i})\|_{F} \right] \\
  & = O_{p}\left( \frac{h_{\pl,N}}{h_{N}\sqrt{n_{\pl}}} \right) 
  =  o_{p}(1).
\end{align*}
To prove $G_{13}=o_{p}(1)$, it suffices to show that
$N^{-1}\sumN\left\{ W_{i,2} L(\X_{i}) K_{\ag,N}^{0}(\X_{i}) \right\}_{k,l}=O_{p}(1)$.
By \eqref{eq:W2-first} and \eqref{eq:ag-TN-ineq-1},
\begin{align*}
  & \Exp\left[
    \left|
    \oneN\sumN
    \left\{ W_{i,2} L(\X_{i}) K_{\ag,N}^{0}(\X_{i}) \right\}_{k,l}
    \right|
    \middle| \Fn, \{V_{\pl,i}\}_{i=1}^{N}
    \right] \\
  & \le
    \Exp\left[  \oneN\sumN
    |W_{i,2}| L(\X_{i}) \|K_{\ag,N}^{0}(\X_{i})\|_{F}
    \middle| \Fn, \{V_{\pl,i}\}_{i=1}^{N}
    \right] \\
  & \le \frac{4\|\cbeta_{\pl}-\bbeta_{0}\|}{\alpha} \oneN \sumN
    L_{\pl}(Y_{i}, \X_{i}) L(\X_{i}) \|K_{\ag,N}^{0}(\X_{i})\|_{F}
    = O_{p}(n_{\pl}^{-1/2})O_{p}(1) = o_{p}(1).
\end{align*}
Therefore,
$N^{-1}\sumN\left\{ W_{i,2} L(\X_{i}) K_{\ag,N}^{0}(\X_{i}) \right\}_{k,l}=O_{p}(1)$
by Markov's inequality.  For $G_{14}$, we have
$\Exp\left[ N^{-1}\sumN\left\{ W_{i,2} \dot{\mu}(\X_{i}\tp\bbeta_{0}) K_{\ag,N}^{0}(\X_{i}) \right\}_{k,l} \right] = 0$.
Using \eqref{eq:W2-second-looser},
\begin{align*}
  & \Var\left[ \oneN\sumN\left\{ W_{i,2}
  \dot{\mu}(\X_{i}\tp\bbeta_{0}) K_{\ag,N}^{0}(\X_{i}) \right\}_{k,l}
   \middle| \Fn, \{V_{\pl,i}\}_{i=1}^{N}\right] \\
  & \le
    \frac{1}{N^{2}} \sumN \Exp\left\{
    W_{i,2}^{2}
    \middle| \Fn, \{V_{\pl,i}\}_{i=1}^{N} \right\}  \|\dot{\mu}(\X_{i}\tp\bbeta_{0}) K_{\ag,N}^{0}(\X_{i})\|_{F}^{2} \\ 
  & \le
    \frac{4}{\alpha N h_{N}} \oneN \sumN
    \|\dot{\mu}(\X_{i}\tp\bbeta_{0}) K_{\ag,N}^{0}(\X_{i})\|_{F}^{2} = o_{p}(1),
\end{align*}
where the last step follows from Assumption~\ref{asmp:samp-prob-2} (iii). 
This implies that $|G_{14}|=o_{p}(1)$. Thus, $|G_{1}|=o_{p}(1)$.

For $|G_{2}|_{k,l}$, we have
\begin{align*}
  |G_{2}|_{k,l}
  &\le \|\cbeta_{\ag} - \bbeta_{0}\|
    \left|\oneN \sumN \left\{ \frac{\delta^{0}_{\ag,i}}{\pi_{\ag}^{0}(Y_{i}, \X_{i})}
    L(\X_{i}) K_{\ag,N}^{0}(\X_{i}) \right\}
    \right|_{k,l}.
\end{align*}
By Lemma~\ref{lem:ag-feasible-consistency},
$\|\cbeta_{\ag} - \bbeta_{0}\|=o_{p}(1)$. By
Assumption~\ref{asmp:samp-prob-2} (i),
$\sup_{N} \Exp\{ L(\X) \| K_{\ag,N}^{0}(\X) \|_{F} \} < \infty$. Thus,
\begin{align*}
  & \Exp\left\{\frac{\delta^{0}_{\ag,i}}{\pi_{\ag}^{0}(Y_{i}, \X_{i})}
   \left|L(\X_{i})K_{\ag,N}^{0}(\X_{i}) \right|_{k,l} \right\} 
  = \Exp\{ \left| L(\X) K_{\ag,N}^{0}(\X) \right|_{k,l}\} 
  \le \sup_{N}
  \Exp\left\|L(\X)K_{\ag,N}^{0}(\X) \right\|_{F}
  <  \infty.
\end{align*}
Thus, we obtain $|G_{2}|_{k,l} = o_{p}(1)$ for $k,l=1,\dots,d$. Finally,
to prove $|G_{3}|_{k,l}=o_{p}(1)$, we observe that $G_{3}=G_{31} + G_{32}$, where
\begin{align*}
  G_{31}
  & = \oneN\sumN \left\{  
    \left(\frac{\delta^{0}_{\ag,i}}{\pi_{\ag}^{0}(Y_{i}, \X_{i})} - 1 \right)
    \dot{\mu}(\X_{i}\tp\bbeta_{0})
    K_{\ag,N}^{0}(\X_{i})\right\}, \\
  G_{32}
  & = \oneN\sumN \left\{  
    \dot{\mu}(\X_{i}\tp\bbeta_{0})
    K_{\ag,N}^{0}(\X_{i})\right\}
    -M_{\ag}^{0}(\bbeta_{0}).
\end{align*}
By Assumption~\ref{asmp:samp-prob-2} (iii) and the same argument in
the proof of \eqref{eq:M-pi-N} in Theorem~\ref{thm:sub1}, we have
$G_{31}=o_{p}(1)$, and $G_{32}=o_{p}(1)$ is implied by
Assumption~\ref{asmp:samp-prob-1} (ii).  Thus, \eqref{eq:M-ag-N} is
proved.
\end{proof}

\subsection{Aggregation Methods Implementation}
\label{supp:ag-implement}
Recall that the pilot and second-step inclusion probabilities are
$\{\pi_{\pl}(Y_i,\X_i)\}_{i=1}^{N}$ and
$\{\tpi(Y_i,\X_i)\}_{i=1}^{N}$, respectively.  To distinguish
realized and expected sample sizes, let $n_{\pl}^{\star}$ and
$n^{\dagger}$ be the realized sample sizes after the implementation of
pilot sampling and second-step sampling. Correspondingly, let
$\{(Y_{j}^{\star}, \X_{j}^{\star})\}_{j=1}^{n_{\pl}^{\star}}$ and
$\{(Y_{k}^{\dagger}, \X_{k}^{\dagger})\}_{k=1}^{n^{\dagger}}$ be the
realized observations in the pilot sample and second-step sample.
Therefore, the pilot and second-step estimating equations are
\begin{align*}
  & \Psi_{\pl}^{\star}(\bbeta) = 
    \oneN 
  \sum_{j=1}^{n_{\pl}^{\star}}
    \frac{\psi(Y_{j}^{\star}, \X_{j}^{\star},\bbeta)}
    {\pi_{\pl}(Y_{j}^{\star}, \X_{j}^{\star})}, \qquad
   \tilde\Psi^{\dagger}(\bbeta) = 
    \oneN \sum_{k=1}^{n^{\dagger}}
    \frac{\psi(Y_{k}^{\dagger},\X_{k}^{\dagger},\bbeta)}
    {\tpi(Y_{k}^{\dagger}, \X_{k}^{\dagger})}.
\end{align*}
The solutions are $\cbeta_{\pl}$ and $\tilde\bbeta_{n}$, respectively.
The pooled-sample estimator
$\cbeta_{\mathrm{pool}}$ solves
\begin{align*}
  \Psi_{\mathrm{pool}}(\bbeta) = 
  \frac{n_{\pl}^{\star}}{n_{\pl}^{\star}+n^{\dagger}}
  \Psi_{\pl}^{\star}(\bbeta)
  +  \frac{n^{\dagger}}{n_{\pl}^{\star}+n^{\dagger}}
  \tilde\Psi^{\dagger}(\bbeta) 
    = \0.
\end{align*}
The estimated covariance matrix of $\cbeta_{\mathrm{pool}}$ is
\begin{align*}
  & \Var(\cbeta_{\mathrm{pool}}) 
   = \frac{1}{n_{\pl}^{\star}+n^{\dagger}}
   \left\{ \dot{\Psi}_{\mathrm{pool}}(\cbeta_{\mathrm{pool}}) \right\}^{-1} \\
  & \Biggl[ \frac{1}{N^{2}}
    \Biggl\{ 
  \frac{(n_{\pl}^{\star})^{2}}{n_{\pl}^{\star}+n^{\dagger}} \sum_{j=1}^{n_{\pl}^{\star}}
    \frac{\{\psi(Y_{j}^{\star},\X_{j}^{\star},\cbeta_{\mathrm{pool}})\}^{\otimes2}}
    {\pi_{\pl}^{2}(Y_{j}^{\star}, \X_{j}^{\star})} 
  +
    \frac{(n^{\dagger})^{2}}{n_{\pl}^{\star}+n^{\dagger}}
    \sum_{k=1}^{n^{\dagger}} \frac{\{\psi(Y_{k}^{\dagger},\X_{k}^{\dagger},\cbeta_{\mathrm{pool}})\}^{\otimes2}}
    {\tpi^{2}(Y_{k}^{\dagger}, \X_{k}^{\dagger})} 
   \Biggr\}
  \Biggr] \\
  & \left\{ \dot{\Psi}_{\mathrm{pool}}(\cbeta_{\mathrm{pool}}) \right\}^{-1}.
\end{align*}
Regarding the weighting procedure, after deriving $\cbeta_{\pl}$ and
$\tilde\bbeta_{n}$, we calculate the weighting matrices by
$w_{\pl} = n_{\pl}^{\star}(n_{\pl}^{\star}+n^{\dagger})^{-1}\dot\Psi_{\pl}^{\star}(\cbeta_{\pl})$
and
$w = n^{\dagger}(n_{\pl}^{\star}+n^{\dagger})^{-1}\dot{\tilde\Psi}^{\dagger}(\tilde\bbeta_{n})$,
then the weighted estimator is
$\cbeta_{\mathrm{we}}=(w_{\pl}+w)^{-1}(w_{\pl}\cbeta_{\pl} + w \tilde\bbeta_{n})$,
and its estimated covariance matrix is 
\begin{align*}
  & \Var(\cbeta_{\mathrm{we}})
   =  \frac{1}{n_{\pl}^{\star}+n^{\dagger}} (w_{\pl}+w)^{-1} \\
  & \left[ \frac{1}{N^{2} }
   \left\{  
   \frac{(n_{\pl}^{\star})^{2}}{n_{\pl}^{\star}+n^{\dagger}} \sum_{j=1}^{n_{\pl}^{\star}}
    \frac{\{\psi(Y_{j}^{\star},\X_{j}^{\star},\cbeta_{\pl})\}^{\otimes2}}
    {\pi_{\pl}^{2}(Y_{j}^{\star}, \X_{j}^{\star})}
  +  \frac{(n^{\dagger})^{2}}{n_{\pl}^{\star}+n^{\dagger}}
    \sum_{k=1}^{n^{\dagger}} \frac{\{\psi(Y_{k}^{\dagger},\X_{k}^{\dagger},\tilde\bbeta_{n})\}^{\otimes2}}
    {\tpi^{2}(Y_{k}^{\dagger}, \X_{k}^{\dagger})} 
   \right\}
  \right] (w_{\pl}+w)^{-1}.
\end{align*}

\end{document}